\documentclass[a4paper, 11pt]{article}
\usepackage[T1]{fontenc}
\usepackage{comment}
\usepackage{amsmath}
\usepackage{mathtools}
\usepackage{amssymb}
\usepackage{amsthm}
\usepackage{setspace}
\usepackage{placeins}
\usepackage{minibox}
\usepackage{graphicx}
\usepackage{subcaption}
\PassOptionsToPackage{dvipsnames}{xcolor}
\usepackage{tikz}
\usepackage{pgfplots}
\pgfplotsset{compat=newest}
\usepackage{xcolor}
\usepackage[round]{natbib}
\usepackage{caption}
\usepackage{booktabs}
\usepackage{multirow}
\usepackage{authblk}
\usepackage{epstopdf}
\usepackage{algorithm}
\usepackage{algpseudocode}
\usepackage{mdframed}
\usepackage{listings}
\usepackage{circledsteps}
\usepackage{fancyhdr}
\usepackage{lastpage}
\usepackage{hyperref}
\usepackage{cleveref}
\usepackage{autonum}
\usepackage{soul}
\usepackage{enumitem}
\usepackage[edges]{forest}
\usepackage[hmargin = 2cm, vmargin = 2cm]{geometry}
\usepackage{color}
\usepackage[dvipsnames]{xcolor}

\newcommand{\R}{\mathbb{R}}

\newcommand{\E}{\mathbb E}

\newcommand{\T}{\mathcal E}

\DeclareMathOperator*{\argmin}{arg\,min}

\newtheorem{theorem}{Theorem}
\newtheorem{proposition}{Proposition}
\newtheorem{lemma}{Lemma}

\newtheorem{assumption}{Assumption}
\newtheorem{definition}{Definition}
\theoremstyle{definition}

\Crefname{appendix}{Supplement}{Supplements}
\Crefname{subappendix}{Supplement}{Supplements}
\Crefname{subsubappendix}{Supplement}{Supplements}

\title{Wasserstein Exponential Smoothing for Distributional Time Series Forecasting}

\author{
	Takuo Matsubara \quad Peiwen Jiang \quad Minh-Ngoc Tran \quad Wilson Ye Chen
}

\date{
	\small{The University of Sydney Business School, Australia}
}

\begin{document}

\maketitle

\begin{abstract}
Distributional time series arise when each temporal observation is a probability distribution rather than a scalar. 
We propose Wasserstein exponential smoothing (WES), a one-parameter recursive forecasting method for distributional time series on $\R$. 
The method adapts the practical logic of classical exponential smoothing to probability distributions by updating forecast distributions along Wasserstein geodesics. 
This yields a simple filter that can be applied directly to empirical distributions without parametric density modeling. 
We estimate the smoothing parameter by minimizing an in-sample Wasserstein prediction loss and establish consistency under a distributional local-level data-generating process. 
In applications to high-frequency equity-index return distributions and household electricity-demand distributions, WES attains the lowest one-step-ahead Wasserstein prediction error among existing distributional autoregressive and regression-based benchmarks for all $20$ series considered, and is retained in the $90\%$ model confidence set in every case.

\vspace{10pt}

\noindent\textbf{Keywords:} Distributional time series; Distributional forecasting; Optimal transport; Symbolic data analysis; Empirical distributions.
\end{abstract}


\section{Introduction}

Many forecasting problems in business, economics and finance now involve observations that are more naturally represented as distributions than as scalars. High-frequency financial data, for example, generate a daily empirical distribution of intraday returns, whose evolution contains information about volatility, skewness and tail behavior. Smart-meter data similarly produce within-day distributions of household electricity demand, where the object of interest is not only the daily average but also the shape and variability of consumption over the day. In such settings, forecasting a scalar summary can discard information that is directly relevant for risk management, demand planning and operational decision-making. This motivates the development of time-series models for distribution-valued observations.

Forecasting entire distributions is closely related to the broader literature on density and probabilistic forecasting \citep{Diebold1998Density,TayWallis2000,GneitingKatzfuss2014}. A key distinction, however, is that in distributional time series the observed data object at each time point is itself a probability distribution. The task is therefore to model and forecast a sequence $\nu_1,\nu_2,\ldots$ whose elements lie in a non-Euclidean space of probability measures. This perspective is also related to functional time-series forecasting, where each observation is a curve or function rather than a scalar \citep{HyndmanUllah2007,HorvathKokoszka2012}. In financial applications, quantile-function-valued time-series models have also been developed to forecast dynamic distributional features of high-frequency returns \citep{Chen2022DQF}. Distributional time series impose an additional structural constraint: each observed curve must represent a valid probability distribution. Consequently, standard vector-space operations such as addition and scalar multiplication are not directly applicable, and the geometric structure of the space of probability distributions must be respected.

Wasserstein geometry provides a natural framework for this problem because it compares distributions through optimal transport rather than pointwise vertical differences. Recent work has developed several autoregressive and regression-based models for distributional time series and related distribution-valued data. For example, \citet{zhang2022wasserstein} proposed Wasserstein autoregressive models based on tangent-space representations, \citet{Zhu2023} developed autoregressive models for optimal transport maps, and \citet{ghodrati2024distributional} studied distributional autoregression through iterated transportation. Related regression-based approaches include \citet{ghodrati2022distribution} and \citet{chen2023wasserstein}. These contributions show that temporal dependence between distributions can be modeled in a way that respects Wasserstein geometry. At the same time, most existing approaches are autoregressive or regression-based, and therefore typically require the estimation of regression operators, transport-map dynamics or other structured dependence mechanisms.

In scalar time-series forecasting, autoregressive models are not the only benchmark. Exponential smoothing (ES) is one of the most widely used forecasting methods because it provides a simple recursive update, governed by a small number of smoothing parameters, that often performs well across a wide range of data-generating mechanisms \citep{Gardner1985,Gardner2006,Hyndman2008}. In its simplest form, ES updates the current forecast by moving part of the way from the previous forecast toward the latest observation. This makes it especially attractive when the goal is adaptive one-step-ahead forecasting rather than detailed structural modeling. Despite the central role of ES in scalar forecasting, the existing toolkit for distribution-valued time series offers no direct analogue of exponential smoothing that operates on probability distributions under Wasserstein geometry.

This paper proposes Wasserstein exponential smoothing (WES), a one-parameter recursive forecasting model for distributional time series on $\mathbb R$. The method adapts the updating principle of classical ES to probability distributions by replacing Euclidean interpolation with \emph{displacement interpolation} in the Wasserstein space. WES updates the forecast by moving it a fraction $\theta$ of the way along the optimal-transport path from the previous forecast distribution to the latest observed distribution. The resulting model preserves the simplicity and adaptivity of classical ES while respecting the intrinsic geometry of the space of probability distributions. In one dimension, the method has a particularly simple quantile representation, so it can be applied directly to empirical distributions without parametric density modeling.

The contribution of the paper is threefold. First, we formulate WES as an exponential-smoothing model for distribution-valued time series, providing a parsimonious alternative to existing Wasserstein autoregressive and regression-based approaches. Second, we estimate the smoothing parameter by minimizing an in-sample Wasserstein prediction loss and establish consistency under a distributional local-level data-generating process. Third, we evaluate the method in two empirical forecasting applications of direct relevance to business and economic statistics: high-frequency equity-index return distributions and household electricity-demand distributions. The latter connects the proposed methodology to the literature on probabilistic electricity-demand forecasting \citep{HongFan2016}. Across both applications, WES attains the lowest out-of-sample Wasserstein prediction error among the autoregressive and regression-based benchmarks considered.

The remainder of the paper is organized as follows. \Cref{sec:background} reviews the necessary background on exponential smoothing and related distributional time-series methods. \Cref{sec:methodology} introduces the WES model and its estimation procedure. \Cref{sec:theory} presents theoretical properties of the WES process and the smoothing-parameter estimator. \Cref{sec:simulation} studies the finite-sample behavior of the estimator. \Cref{sec:experiment} reports the empirical forecasting results. \Cref{sec:conclusion} concludes and discusses possible extensions.


\section{Background} \label{sec:background}

This section reviews the background required for the proposed forecasting model. 
We first recall the simple exponential-smoothing recursion that motivates WES.
We then position WES relative to existing time-series models for distribution-valued data.

\subsection{Exponential Smoothing}

Exponential smoothing (ES) is a standard forecasting methodology built around simple recursive updating. 
Its practical appeal is that accurate one-step-ahead forecasts can often be obtained from a small number of smoothing parameters, without specifying a detailed autoregressive structure \citep{Gardner1985,Gardner2006,Hyndman2008}. 
In this paper, we focus on simple exponential smoothing, the level-only form of ES, because it provides the clearest analogue for a one-parameter distributional forecasting model.

Let $\{ y_t \}_{t=1}^{T}$ denote a real-valued time series.
Let $x_t$ denote the latent predictor at time $t$, which serves as the forecast of the next observation $y_{t+1}$.
Simple ES can be expressed in component form as
\begin{align}
    y_t &= x_{t-1} + e_t, \label{eq:es component form 1}\\
    x_t &= (1-\theta)x_{t-1} + \theta y_t, \label{eq:es component form 2}
\end{align}
where $\theta \in [0,1]$ is the smoothing parameter and $\{ e_t \}_{t=1}^{T}$ is an innovation sequence, typically assumed to satisfy $\E[e_t] = 0$ for all $t$.
The predictor $x_t$ is obtained by taking a convex combination of the previous predictor $x_{t-1}$ and the current observation $y_t$. 

The term `exponential smoothing' arises because, by solving the recursion in \eqref{eq:es component form 1}--\eqref{eq:es component form 2}, the one-step-ahead forecast $x_t$ can be expanded into an exponentially weighted average of past observations.
Because simple ES contains only a level component, its multi-step-ahead forecasts are constant, and the method is best suited to series without pronounced trend or seasonality.

Several equivalent representations of ES underpin its statistical interpretation.
In particular, ES can be rewritten in error-correction form as
\begin{align}
    y_t &= x_{t-1} + e_t, \label{eq:es error correction form 1}\\
    x_t &= x_{t-1} + \theta e_t, \label{eq:es error correction form 2}
\end{align}
where $e_t = y_t - x_{t-1}$ serves as the one-step-ahead forecast error. 
This representation shows that the predictor is updated by correcting the previous estimate in the direction of the latest innovation. 
It also establishes the link between ES and innovation state-space models \citep{Hyndman2008}.

Simple ES is also forecast-equivalent to an ARIMA$(0,1,1)$ model in the additive case \citep{Gardner2006,Hyndman2008}. 
From a state-space perspective, simple ES corresponds to the local-level model in steady state, and is optimal under the random-walk-plus-noise framework studied by \citet{Muth1960}. 
These equivalences help explain why ES often performs well in one-step-ahead forecasting, even when the true data-generating process is unknown.

\subsection{Related Time-Series Models for Distribution-Valued Data}
\label{sec:related_work}

The closest methodological context for WES is the recent literature on distribution-valued time series under Wasserstein geometry. 
Much of this literature has developed autoregressive and regression-based approaches for modeling temporal dependence between probability distributions. 
An important early contribution is the Wasserstein autoregressive model (WAR) of \citet{zhang2022wasserstein}, which constructs autoregressive models for density time series by lifting distributions to a tangent space at the Wasserstein mean. 
This yields a flexible analogue of classical autoregression, with temporal dependence represented through tangent-space linear operators.
A related but more intrinsic approach is the autoregressive transport map model (ATM) of \citet{Zhu2023}. 
Rather than working in tangent-space coordinates, ATM imposes the autoregressive structure directly on optimal transport maps, such as maps from a barycenter to individual distributions or maps between successive observations.

The distribution-on-distribution regression model (DoDRM) of \citet{ghodrati2022distribution} links the conditional Fr\'echet mean of a response distribution to a predictor distribution through optimal transport. 
Building on this line of work, \citet{ghodrati2024distributional} study a distributional autoregressive model (DAM) based on iterated transportation, providing a Markovian framework for autoregressive dynamics in Wasserstein space. 
These regression-based methods, together with WAR and ATM, demonstrate the usefulness of Wasserstein geometry for modeling temporal dependence between distributions and serve as the empirical benchmarks in our applications.

A related strand of work comes from symbolic data analysis, where observations may be intervals, histograms or other aggregate distributional summaries rather than scalars \citep{BillardDiday2003,BillardDiday2006}. 
This perspective has led to forecasting methods for interval-valued and histogram-valued time series. 
For interval-valued time series, \citet{ArroyoEspinolaMate2011} compare forecasting approaches based on lower and upper bounds, center and radius representations, and interval-arithmetic extensions of classical forecasting methods, while \citet{MaiaCarvalho2011} develop Holt-type exponential smoothing and neural-network models for interval-valued data. 
For histogram-valued time series, \citet{ArroyoMate2009} develop nearest-neighbor forecasting methods using dissimilarities between histograms, including Mallows and Wasserstein distances, and \citet{GonzalezRiveraArroyo2012} study the dependence structure of daily histograms of intraday S\&P 500 returns. 
Closely related in motivation, \citet{ArroyoGonzalezRiveraMateMunoz2011} develop smoothing methods for histogram-valued time series using barycentric histograms. 
The present paper adopts a complementary formulation in which each observation is treated as a probability distribution in the Wasserstein space.
The resulting filter applies directly to empirical distributions without fixing histogram bins and provides a natural basis for estimating the smoothing parameter through Wasserstein prediction loss.

The proposed method is also related to quantile-function and functional time-series approaches. 
For example, \citet{Chen2022DQF} develop dynamic quantile-function models for high-frequency financial returns, in which evolving distributional information is represented through a parametric dynamic model for quantile functions. 
More broadly, functional time-series methods model temporally ordered curves and functions \citep{HyndmanUllah2007,HorvathKokoszka2012}.
WES differs from these approaches in how the defining constraints of a probability distribution are maintained.
Linear operations on unconstrained curves need not return a valid quantile or density function, and parametric quantile-function models enforce validity only within a chosen family, whereas the geodesic update of WES always yields a valid forecast distribution without any parametric restriction.

Taken together, these literatures motivate WES.
The proposed method draws on standard tools from Wasserstein statistics, including pushforward maps, displacement interpolation, and quantile representations on $\R$.
Building on these, we develop a parsimonious recursive forecasting model for distribution-valued time series, supported by a practical estimator of the smoothing parameter and its consistency theory.
Governed by a single smoothing parameter, WES is best viewed as a distributional analogue of simple exponential smoothing and as a complement to existing Wasserstein autoregressive and regression-based approaches.


\section{Wasserstein Exponential Smoothing} \label{sec:methodology}

This section introduces Wasserstein exponential smoothing (WES) as a recursive forecasting filter for distribution-valued time series. 
The construction follows the updating logic of simple exponential smoothing: each new observation pulls the forecast distribution partway toward it. 
The key difference is that the movement is taken along a Wasserstein geodesic rather than a Euclidean line segment.

We first recall the Wasserstein distance and displacement interpolation for probability distributions on $\R$.
We then define the WES filter as an operational forecasting method for an observed sequence of distributions. 
We next describe estimation of the smoothing parameter by minimizing an in-sample Wasserstein prediction loss. 
Finally, we introduce a population WES process and its quantile error-correction representation, which provide the theoretical interpretation used in \Cref{sec:theory}.

\subsection{Wasserstein Geometry on the Real Line}

The role of Wasserstein geometry in this paper is operational: it provides both the loss function used to compare forecast and observed distributions, and the interpolation path used to update a forecast distribution. 
Let $\mathcal P_2(\R)$ denote the space of probability distributions on $\R$ with finite second moment. 
Let $( T )_\# \mu$ denote the pushforward of a distribution $\mu$ by a map $T$.

For $\mu,\nu\in\mathcal P_2(\R)$, the squared 2-Wasserstein distance is defined by
\begin{equation}\label{eq:Wasserstein dis}
W_2^2(\mu,\nu)
=
\inf_{\pi\in\Pi(\mu,\nu)}
\int_{\R\times\R} |x-y|^2\,\pi(dx,dy),
\end{equation}
where $\Pi(\mu,\nu)$ is the set of probability measures on $\R\times\R$ with marginals $\mu$ and $\nu$ \citep{villani2009optimal}.
The Wasserstein space $\mathbb{W}_2(\R)$ is $\mathcal P_2(\R)$ equipped with this distance.
The infimum in \eqref{eq:Wasserstein dis} is attained through a map that rearranges the probability mass of $\mu$ into $\nu$ at minimal mean-squared displacement.
This map, denoted $T_{\mu \to \nu}$, is called the \emph{optimal transport map} and satisfies $\nu = ( T_{\mu \to \nu} )_{\#} \mu$.

In one dimension, the Wasserstein distance admits a particularly useful closed-form representation through quantile functions \citep{Panaretos2020}. 
Let $U$ and $V$ denote the quantile functions of $\mu$ and $\nu$, respectively. 
Since distributions in $\mathbb{W}_2(\R)$ have finite second moments, their quantile functions belong to $L^2((0,1))$.
Then the Wasserstein distance satisfies
\begin{equation}\label{eq:wasserstein_quantile}
    W_2^2(\mu,\nu)
    =
    \int_0^1 ( U(q)-V(q) )^2\,dq .
\end{equation}
This representation is central to the estimation and computation of WES, since the empirical distributions in our applications can be evaluated directly on a common quantile grid.

Wasserstein geometry also provides a natural analogue of linear interpolation, called \emph{displacement interpolation} \citep{villani2009optimal}.
Given the optimal transport map $T_{\mu \to \nu}$, the displacement interpolation between $\mu$ and $\nu$ is the path of distributions
\begin{equation}\label{eq:wasserstein_geodesic}
    \mu^\theta
    =
    \left( (1-\theta) \cdot \operatorname{Id}
    +
    \theta \cdot T_{\mu \to \nu} \right)_{\#}\mu,
    \qquad
    \theta\in[0,1],
\end{equation}
where $\operatorname{Id}$ is the identity map $\operatorname{Id}(x) = x$.
It recovers $\mu$ at $\theta = 0$ and $\nu$ at $\theta=1$, representing the geodesic (i.e., the shortest path) between the two distributions in the Wasserstein space \citep{Ambrosio2024}.
In one dimension, this path of distributions $\mu^\theta$ admits a quantile representation:
\begin{equation}\label{eq:wasserstein_geodesic_quantile}
    U^\theta(q) = (1-\theta)U(q)+\theta V(q), \qquad q\in(0,1).
\end{equation}
\Cref{fig:displacement_interpolation} illustrates displacement interpolation; as $\theta$ increases, the interpolated distribution moves along the Wasserstein geodesic, with its probability mass optimally transported from $\mu$ to $\nu$.

\begin{figure}[t]
    \centering
    \begin{subfigure}[b]{0.40\textwidth}
        \centering
        \begin{tikzpicture}
            \draw[black, thick] (0,0) ellipse (2.35 and 1.4);
            \node at (0,0.95) {$\mathbb{W}_2(\R)$};
            \fill (-1.2,-0.5) circle (2pt);
            \node[below left=-2pt, xshift=-1.5pt] at (-1.2,-0.5) {$\mu$};
            \fill (1.35,0.2) circle (2pt);
            \node[above right=-2pt] at (1.35,0.2) {$\nu$};
            \draw[->, line width=0.7pt, shorten >=3pt] (-1.2,-0.5) to[bend right=12] node[pos=0.4, circle, fill, inner sep=0pt, minimum size=4pt] {} node[pos=0.4, below=0pt] {$\mu^\theta$} (1.35,0.2);
        \end{tikzpicture}
        \caption{Geodesic in the Wasserstein space}
    \end{subfigure}
    \hfill
    \begin{subfigure}[b]{0.56\textwidth}
        \centering
        \begin{tikzpicture}
        \begin{axis}[
            width=\textwidth, height=120pt,
            axis x line=bottom, axis y line=none,
            axis line style={-},
            xtick=\empty, xmin=-4.6, xmax=5.4, ymin=0, ymax=1.0,
            clip=false,
            declare function={gauss(\x,\m,\s)=exp(-(\x-\m)^2/(2*\s^2))/(\s*2.5066);}
        ]
            \addplot[gray!70, line width=1.1pt, smooth, samples=140, domain=-4.6:5.4] {gauss(x,-1.2,0.55)};
            \addplot[gray!70, line width=1.1pt, smooth, samples=140, domain=-4.6:5.4] {gauss(x,0,0.65)};
            \addplot[gray!70, line width=1.1pt, smooth, samples=140, domain=-4.6:5.4] {gauss(x,1.2,0.75)};
            \addplot[black, line width=1.1pt, smooth, samples=140, domain=-4.6:5.4] {gauss(x,-2.4,0.45)};
            \addplot[black, line width=1.1pt, smooth, samples=140, domain=-4.6:5.4] {gauss(x,2.4,0.85)};
            \node at (axis cs:-2.65,0.95) {$\mu$};
            \node at (axis cs:2.4,0.53) {$\nu$};
            \node at (axis cs:0,0.72) {$\mu^\theta$};
        \end{axis}
        \end{tikzpicture}
        \caption{Interpolated distributions on $\R$}
    \end{subfigure}
    \caption{
        Displacement interpolation between $\mu$ and $\nu$. 
        Panel (a) illustrates the displacement interpolation \eqref{eq:wasserstein_geodesic} tracing the Wasserstein geodesic connecting $\mu$ and $\nu$.
        Panel (b) visualizes the interpolated distribution $\mu^\theta$ at $\theta = 0.25, 0.5, 0.75$ for two distinct Gaussian $\mu$ and $\nu$.
    }
    \label{fig:displacement_interpolation}
\end{figure}
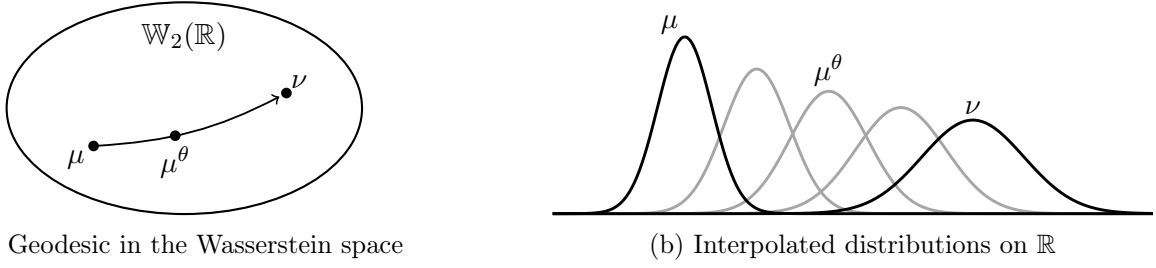

Throughout the paper, distribution-valued observations are treated as random elements of $\mathbb{W}_2(\R)$.
Their data-generating mechanism is specified in \Cref{sec:sample} through perturbations of distributions by random pushforward maps, as is standard in the analysis of random measures and optimal transport \citep{Panaretos2020}. 
The WES filter and the smoothing-parameter estimation, which we describe next, operate directly on realized data and do not depend on this stochastic specification.

\subsection{Observed Distributional Time Series and the WES Filter} \label{sec:setup} \label{sec:estimate}

In standard time-series analysis, we observe a sequence of scalars $y_1,\ldots,y_T$ indexed by time.
In distributional time-series analysis, we instead observe a sequence of probability distributions $\nu_1,\ldots,\nu_T$ indexed by time.
The framework of WES imposes no parametric form on $\nu_t$, nor does it require $\nu_t$ to admit a density.
In our applications, each $\nu_t$ is represented by an empirical distribution based on a finite sample.

The goal is to construct a sequence of forecast distributions $\mu_0^\theta,\mu_1^\theta,\ldots,\mu_T^\theta$, where $\mu_t^\theta$ is used as the one-step-ahead forecast for the next observable distribution $\nu_{t+1}$.
The parameter $\theta\in(0,1)$ controls the degree of smoothing.
Small values of $\theta$ produce smoother forecast trajectories, while larger values make the forecast distribution more adaptive to the latest observation.

\begin{definition}[WES Filter]
	Given a parameter $\theta \in (0, 1)$, a fixed initial distribution $\mu_0$, and a data sequence $\{ \nu_t \}_{t=1}^{T}$, the predictor $\mu_t^{\theta}$ is defined recursively from $t = 1$ to $T$ by
	\begin{align}\label{eq:WES model update}
        \mu_{t}^{\theta} & = \left( (1 - \theta) \cdot \operatorname{Id} + \theta \cdot T_{\mu_{t-1}^{\theta} \to \nu_{t}} \right)_{\#} \mu_{t-1}^{\theta} .
	\end{align}
	The resulting sequence of predictors $\{ \mu_t^{\theta} \}_{t=0}^{T}$ is called the WES filter of the data sequence $\{ \nu_t \}_{t=1}^{T}$.
\end{definition}

The update in \eqref{eq:WES model update} is the direct Wasserstein analogue of the convex averaging step in simple ES.
The previous forecast distribution $\mu_{t-1}^{\theta}$ is moved a fraction $\theta$ of the way toward the newly observed distribution $\nu_t$ along the optimal-transport path from $\mu_{t-1}^{\theta}$ to $\nu_t$.
The WES filter therefore requires only the observed distributional time series and the smoothing parameter $\theta$.
It does not require the estimation of autoregressive operators, transport-map dynamics, or parametric density models.

\subsection{Estimating the Smoothing Parameter}

We estimate the smoothing parameter $\theta$ by minimizing the in-sample one-step-ahead Wasserstein prediction error of the WES filter.
At each time $t$, the predictor $\mu_t^{\theta}$ serves as a one-step-ahead forecast for the next observation $\nu_{t+1}$. 
Its accuracy is measured by the squared Wasserstein error $W_2^2( \mu_t^{\theta}, \nu_{t+1} )$.
Averaging this error across the sample yields the estimation criterion.

\begin{definition}[Minimum Wasserstein Estimator] \label{def:minimum_Wasserstein_estimator}
    Let $\{ \mu_t^\theta \}_{t=0}^{T}$ denote the WES filter of a given data sequence $\{ \nu_t \}_{t=1}^{T}$ for each smoothing parameter $\theta \in (0, 1)$.
	The estimator $\theta_T$, defined as
	\begin{align}\label{eq:Minimum Wasserstein Estimator}
		\theta_T := \argmin_{\theta \in (0, 1)} L_T(\theta) \quad \text{where} \quad L_T(\theta) := \frac{1}{T} \sum_{t=0}^{T-1} W_2^2\big( \mu_t^{\theta}, \nu_{t+1} \big) ,
	\end{align}
    is termed the minimum Wasserstein estimator of the WES smoothing parameter.
\end{definition}

The Wasserstein error at each time $t$ can be computed efficiently.
Specifically, it reduces to the $L^2$ distance between the quantile functions of $\mu_t^{\theta}$ and $\nu_{t+1}$, denoted by $U_t^\theta$ and $V_{t+1}$ respectively:
\begin{align}
	W_2^2\big( \mu_t^{\theta}, \nu_{t+1} \big) = \int_{(0,1)} \left( U_t^\theta(q) - V_{t+1}(q) \right)^2 dq .
\end{align}
This integral can be evaluated accurately using standard numerical quadrature methods.
Once the estimate $\theta_T$ is obtained, the corresponding final predictor $\mu_T^{\theta_T}$ serves as the one-step-ahead forecast for the subsequent observation $\nu_{T+1}$.

\subsection{Population WES Process} \label{sec:sample} 

The WES filter above is the operational forecasting method applied to a realized distributional time series.
To study its theoretical behavior, we introduce a population data-generating process under which observable distributions are generated from latent predictor distributions.
This population process provides a distribution-valued analogue of the local-level interpretation of simple ES.

A standard modeling assumption in distributional data analysis \citep{Panaretos2020} is that observed distributions arise by pushing an underlying distribution forward through random maps.

\begin{definition}[Random Pushforward Map]
	A random pushforward map $\T$ is a random measurable map from $\R$ to $\R$ whose realizations are almost surely non-decreasing.
	For any predictor distribution $\mu\in\mathbb{W}_2(\R)$ under consideration, the pushforward $\T_{\#}\mu$ is assumed to belong to $\mathbb{W}_2(\R)$ almost surely.
\end{definition}

By almost sure monotonicity, each realization of a random pushforward map represents an optimal transport map between a distribution and its pushforward image.
Our population process defines each observable distribution as the pushforward of the previous latent predictor distribution under such a random transport map, denoted $\T_t$ at time $t$.

\begin{definition}[WES Process] \label{def:wesp}
	Given a parameter $\theta_* \in (0, 1)$ and a fixed initial distribution $\mu_0$, a random observable $\nu_t$ and a random predictor $\mu_t$ are defined recursively from $t = 1$ to $T$ by
	\begin{align}
	    \nu_t &= \left( \T_t \right)_{\#} \mu_{t-1} , \label{eq:WES_process_formula 1} \\ 
        \mu_t &= \left( (1 - \theta_*) \cdot \operatorname{Id} + \theta_* \cdot T_{\mu_{t-1} \to \nu_{t}} \right)_{\#} \mu_{t-1} . \label{eq:WES_process_formula 2}
    \end{align}
    The sequence of random observables $\{ \nu_t \}_{t=1}^{T}$ is called the WES process.
\end{definition}

For simplicity of exposition, we assume the initial condition $\mu_0$ is deterministic.
This entails no loss of generality: all results extend readily to settings with a random $\mu_0$.
We use the terms \emph{WES process} and \emph{WES filter} to distinguish between the population data-generating construction in \Cref{def:wesp} and the operational forecasting recursion in \eqref{eq:WES model update}.
Notably, the filter requires no knowledge of the random pushforward maps $\T_t$ generating the WES process, as it depends only on the realized distributions and the smoothing parameter.

\Cref{fig:WES demonstration} illustrates the structure of the WES process.
As demonstrated by the theoretical analysis in \Cref{sec:theory} and the simulations in \Cref{sec:simulation}, this simple recursion can generate nonstationary distributional dynamics while remaining governed by a single scalar parameter $\theta_* \in (0,1)$.

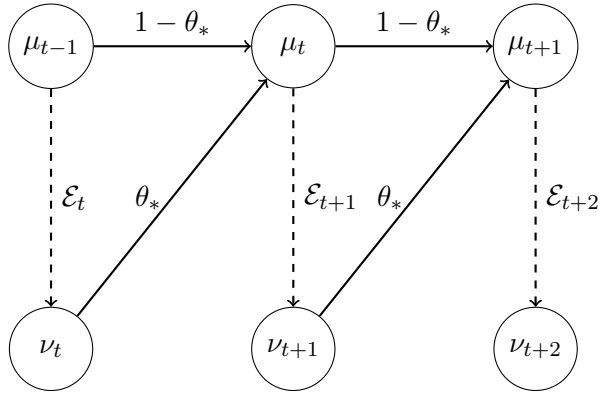
\begin{figure}[t]
\centering
\begin{tikzpicture}[
    node distance=3.2cm,
    state/.style={draw, circle, minimum size=1.1cm},
    arrow/.style={->, thick}
]

\node[state] (xprev) {$\mu_{t-1}$};
\node[state] (yprev) [below of=xprev, yshift=-0.8cm] {$\nu_{t}$};

\node[state] (xt) [right of=xprev] {$\mu_t$};
\node[state] (yt) [below of=xt, yshift=-0.8cm] {$\nu_{t+1}$};

\node[state] (xnext) [right of=xt] {$\mu_{t+1}$};
\node[state] (ynext) [below of=xnext, yshift=-0.8cm] {$\nu_{t+2}$};

\draw[arrow] (xprev) -- node[above] {$1-\theta_*$} (xt);
\draw[arrow] (xt) -- node[above] {$1-\theta_*$} (xnext);

\draw[arrow] (yprev) -- node[left] {$\theta_*$} (xt);
\draw[arrow] (yt) -- node[left] {$\theta_*$} (xnext);

\draw[arrow, dashed] (xprev) -- node[right] {$\mathcal{E}_{t}$} (yprev);
\draw[arrow, dashed] (xt) -- node[right] {$\mathcal{E}_{t+1}$} (yt);
\draw[arrow, dashed] (xnext) -- node[right] {$\mathcal{E}_{t+2}$} (ynext);

\end{tikzpicture}
\caption{The state-space representation of the WES process.}\label{fig:WES demonstration}
\end{figure}

\subsection{Quantile Error-Correction Representation} \label{sec:WES_error_correction_form}

Finally, the WES process admits a quantile representation that directly parallels the error-correction form of classical ES.
This representation is useful both computationally and theoretically.
Because $\nu_t$ and $\mu_t$ are random elements in $\mathbb{W}_2(\R)$, their quantile functions, denoted $V_t$ and $U_t$, are random elements of $L^2((0,1))$.
Under the pushforward relation \eqref{eq:WES_process_formula 1}, these random quantile functions satisfy
\begin{align}
    V_t(q) = \T_t\big( U_{t-1}(q) \big) \quad \text{for each} \quad q\in(0,1) . 
\end{align}

Central to our theoretical analysis is the concept of the \emph{quantile residual} defined below.

\begin{definition}[Quantile Residual] \label{def:quantile_residual}
	At each time $t \ge 1$, we define the quantile residual $F_t$ by
	\begin{align}
        F_t(q) := V_t(q) - U_{t-1}(q) = \T_t\big( U_{t-1}(q) \big) - U_{t-1}(q) \quad \text{for each} \quad q \in (0, 1) , \label{eq:quantile residual}
	\end{align}
    which is a random element in $L^2((0,1))$.
\end{definition}

The WES process can then be written in error-correction form.

\begin{proposition}[Error Correction Form of WES] \label{lem:quantile_process}
    The WES process in \Cref{def:wesp} is equivalent to the sequence of random quantile functions $V_t$ and $U_t$ defined recursively from $t = 1$ to $T$ by
	\begin{align}
		V_t(q)  &= U_{t-1}(q) + F_t(q) , \label{eq:WES error correction form 1}\\
        U_t(q) &= U_{t-1}(q) + \theta_* F_{t}(q) , \label{eq:WES error correction form 2}
	\end{align}
	where $U_0$ is the deterministic quantile function of the fixed initial condition $\mu_0$.
\end{proposition}

This result makes explicit the connection between WES and classical ES.
The recursion in \eqref{eq:WES_process_formula 1}--\eqref{eq:WES_process_formula 2} is the Wasserstein-geometric counterpart of the component form in \eqref{eq:es component form 1}--\eqref{eq:es component form 2}, while the quantile representation in \eqref{eq:WES error correction form 1}--\eqref{eq:WES error correction form 2} corresponds to the error-correction form in \eqref{eq:es error correction form 1}--\eqref{eq:es error correction form 2}.
As demonstrated in \Cref{apx:preparation}, under the regularity conditions of \Cref{sec:theory}, the quantile residuals are centered, i.e., $\E[F_{t}(q)]=0$, and mutually uncorrelated, mirroring the innovation conditions of classical ES.


\section{Theoretical Analysis} \label{sec:theory}

This section studies the statistical behavior of the WES estimator under the population WES process introduced in \Cref{sec:sample}. 
The main result is consistency of the minimum Wasserstein estimator. 
The argument follows the standard M-estimation logic; we show that the sample prediction-loss criterion converges to a deterministic population criterion that is uniquely minimized at the true smoothing parameter.

Throughout this section, $\{ \nu_t \}_{t=1}^{T}$ denotes the data sequence following the WES process with fixed parameter $\theta_* \in (0, 1)$.
Let $\{ \mu_t^\theta \}_{t=0}^{T}$ denote the WES filter conditional on the data sequence $\{ \nu_t \}_{t=1}^{T}$ for a given parameter $\theta \in (0, 1)$.
The following assumption imposes moment and independence conditions on the random pushforward maps $\{ \T_t \}_{t=1}^{T}$ associated with the underlying WES process.

\begin{assumption}[Random transport innovations] \label{ass:transport_innovations}
The random pushforward maps $\{ \T_t \}_{t=1}^{T}$ are mutually independent.
For each time $t \ge 1$, there exist positive constants $\sigma_t$ and $\kappa_t$ such that, for all $x\in \R$, the random pushforward map $\T_t$ satisfies:
\begin{align}
    \E\big[ \T_t(x)-x \big] &= 0, \\
    \E\big[ \{ \T_t(x)-x \}^2 \big] &= \sigma_t^2, \\
    \E\big[ \{ \T_t(x)-x \}^4 \big] &\leq \kappa_t .
\end{align}
Furthermore, the constants $\sigma_t$ and $\kappa_t$ admit time-independent bounds $\sigma_t \leq \sigma$ and $\kappa_t \leq \kappa$ for all $t$.
The time average of the innovation variances converges: $(1 / T) \sum_{t=1}^{T}\sigma_t^2 \to \bar{\sigma}^2$ for some positive $\bar{\sigma}^2$.
\end{assumption}

The zero-mean condition is the Wasserstein analogue of the zero-mean innovation assumption in the scalar local-level representation of simple exponential smoothing.
The bounded second and fourth moment conditions control the size of the random transport perturbations.
The final condition is a mild long-run non-degeneracy requirement: it ensures that the innovation variance does not vanish on average, which is needed to identify the smoothing parameter.

\subsection{Regularity of WES Sample Paths} \label{sec:sample_path_regularity}

The WES process is generally nonstationary, as in the scalar local-level model underlying simple exponential smoothing.
Nevertheless, it has a useful first-moment property.
Because the state space of the WES process is not a vector space, we use the Fr\'{e}chet mean and variance \citep{Frechet1948}, whose existence is well established in Wasserstein space \citep{Panaretos2020}.

\begin{definition}[Fr\'{e}chet Mean and Variance]
	For a random measure $\mu$ in $\mathbb{W}_2(\mathbb{R})$, its Fr\'{e}chet mean and variance are, respectively, a distribution $E(\mu) \in \mathbb{W}_2(\mathbb{R})$ and a scalar $V(\mu) \in \R$, defined by
	\begin{align}
		E(\mu) := \argmin_{ \rho \in \mathbb{W}_2(\mathbb{R}) } \E \big[ W_2^2(\rho, \mu) \big] \quad \text{and} \quad V(\mu) := \E \big[ W_2^2(E(\mu), \mu) \big] .
	\end{align}
\end{definition}

The following result shows that the WES process has a constant Fr\'echet mean under zero-mean transport innovations.
This is a first-moment property only.

\begin{proposition}[Constant Fr\'{e}chet Mean] \label{prop:stationarity}
	The WES process $\{ \nu_t \}_{t=1}^{T}$ satisfies $E(\nu_t) = {\mu}_0$ for all $t \ge 1$, where ${\mu}_0$ is the fixed initial distribution.
    Furthermore, $V(\nu_t) < \infty$ for all $t \ge 1$.
\end{proposition}

The proof is provided in \Cref{apx:proof_stationarity}.
The result shows that, under \Cref{ass:transport_innovations}, the observed distributions are centered around the initial distribution in the Fr\'echet-mean sense.
However, as in the scalar local-level model, the dispersion of the process may increase over time.
Indeed, the Fr\'{e}chet variance $V(\nu_t)$ is finite at each $t$ but need not be bounded uniformly in the horizon.

We next study the one-step-ahead prediction errors generated by the WES filter.
Let $U_t^\theta: (0, 1) \to \R$ denote the quantile function of the predictor $\mu_t^\theta$. 
Recall that $V_t: (0, 1) \to \R$ denotes the quantile function of the observable $\nu_t$.
The predictive error between the predictor $\mu_t^\theta$ and the observable $\nu_{t+1}$ can be expressed through their quantile functions:
\begin{align}
	W_2^2(\mu_t^\theta, \nu_{t+1}) = \int_{(0,1)} \big( U_t^\theta(q) - V_{t+1}(q) \big)^2 dq . \label{eq:model_residual_expression}
\end{align}
We refer to the difference $U_t^\theta - V_{t+1}$ as the \emph{model-quantile residual}.
This residual is the distributional analogue of the one-step-ahead forecast error in scalar exponential smoothing.

We examine the autocovariance of the model-quantile residuals, defined via the $L^2((0,1))$ inner product between residuals at different time points.
In what follows, $\| \cdot \|$ and $\langle \cdot, \cdot \rangle$ denote the norm and inner product of $L^2((0, 1))$, respectively.

\begin{proposition}[Residual Autocovariance] \label{prop:residual_autocovariance}
	For any $\theta \in (0, 1)$ and time $t > s \ge 0$, we have
    \begin{align}
        \E\Big[ \big\langle U_{t}^\theta - V_{t+1}, U_{s}^\theta - V_{s+1} \big\rangle \Big] 
        & = (1 -\theta)^{t+s} \| \Delta_0 \|^2
        + (\theta-\theta_*)^2 (1-\theta)^{t-s}m_s^\theta +\frac{\theta - \theta_*}{1 - \theta} (1 - \theta)^{t-s} \sigma_{s+1}^2,
    \end{align}
    where we define $\Delta_0 := U_0^\theta - U_0$ and $m_s^\theta := \sum_{i=1}^{s} (1 - \theta)^{2(s-i)} \sigma_i^2$.
\end{proposition}

The proof is provided in \Cref{apx:proof_residual_autocovariance}.
The autocovariance of the model-quantile residuals is not strictly translation-invariant, so the residual process is not weakly stationary in the usual time-series sense.
Nevertheless, for any fixed $\theta\in(0,1)$, the dependence induced by the filtering recursion decays geometrically.
This short-memory structure is the key ingredient used below to establish convergence of the Wasserstein prediction-loss criterion.

\subsection{Consistency of the Minimum Wasserstein Estimator} \label{sec:convergence}

This section establishes consistency of the minimum Wasserstein estimator.
Recall the definition
\begin{align}
	\theta_T := \argmin_{\theta \in (0, 1)} L_T(\theta) \quad \text{where} \quad L_T(\theta) := \frac{1}{T} \sum_{t=0}^{T-1} W_2^2( \mu_t^\theta, \nu_{t+1} ) . \label{eq:random_error}
\end{align}
Here $L_T$ is a stochastic function because it depends on the data $\{ \nu_t \}_{t=1}^{T}$ following the population WES process.
By \eqref{eq:model_residual_expression}, $L_T(\theta)$ is the time average of the squared $L^2$ norms of the model-quantile residuals.

The next result identifies the deterministic population limit of the sample prediction-loss criterion.

\begin{theorem}[Pointwise Convergence] \label{thm:pointwise_convergence}
	Under \Cref{ass:transport_innovations}, for each $\theta \in (0, 1)$,
	\begin{align}
		L_T(\theta) \overset{p}{\to} A_\infty(\theta) (\theta - \theta_*)^2 + C_\infty ,
	\end{align}
	where $A_\infty(\theta)$ is a strictly positive deterministic function and $C_\infty$ is a non-negative constant.
\end{theorem}

The limiting criterion is uniquely minimized at $\theta=\theta_*$ because $A_\infty(\theta)>0$ for every $\theta\in(0,1)$.
Thus, \Cref{thm:pointwise_convergence} provides identification of the true smoothing parameter through the population Wasserstein prediction loss.

Pointwise convergence alone is not sufficient for consistency of the minimizer.
We therefore establish uniform convergence over compact subsets of the parameter space.

\begin{proposition}[Uniform Convergence] \label{thm:uniform_convergence}
	Under \Cref{ass:transport_innovations}, the convergence in probability established in \Cref{thm:pointwise_convergence} is uniform over any compact subset of $(0, 1)$.
\end{proposition}

The proofs of \Cref{thm:pointwise_convergence} and \Cref{thm:uniform_convergence} are provided in \Cref{apx:proof_pointwise_convergence,apx:proof_uniform_convergence}, respectively.
The uniform convergence result is a uniform law of large numbers for the prediction-loss criterion.
It follows from the bounded moment conditions and the geometrically decaying dependence structure of the model-quantile residuals.

For the consistency argument, we impose the standard compactness condition in the asymptotic theory of M-estimation \citep{Newey1994,Vaart1998}: we assume $\theta_*$ lies in some compact set $\Omega \subset (0,1)$ and regard the estimator as the minimizer of $L_T(\theta)$ over $\Omega$.
In practice, this amounts to minimizing $L_T(\theta)$ over a closed interval covering most of the unit interval but ruling out the endpoints $0$ and $1$.
We now state the consistency result.

\begin{theorem}[Consistency] \label{thm:consistency}
    Let $\Omega \subset (0,1)$ be compact with $\theta_*\in\Omega$, and define the estimator $\theta_T$ by minimizing $L_T(\theta)$ over $\Omega$.
	The estimator $\theta_T$ is consistent under \Cref{ass:transport_innovations}, that is,
	\begin{align}
	    \theta_T \overset{p}{\to} \theta_* .
	\end{align}
\end{theorem}

The proof is provided in \Cref{apx:proof_consistency}.
The result shows that, under the WES population process, the smoothing parameter can be consistently recovered from the one-step-ahead Wasserstein prediction-loss criterion.
The theory does not require the WES process itself to be stationary in the usual econometric sense; instead, consistency is driven by identification of the population prediction loss and uniform convergence of the sample criterion.


\section{Simulation Studies} \label{sec:simulation}

We next examine the finite-sample behavior of the minimum Wasserstein estimator through a simulation study. 
The simulations serve two purposes: first, to illustrate that the WES process can generate qualitatively different distributional dynamics; and second, to assess whether the estimator $\theta_T$ recovers the true smoothing parameter $\theta_*$ across a range of parameter values and sample sizes.

\begin{figure}[b]
    \centering
    \begin{minipage}{0.48\textwidth}
        \centering
        \includegraphics[width=\textwidth]{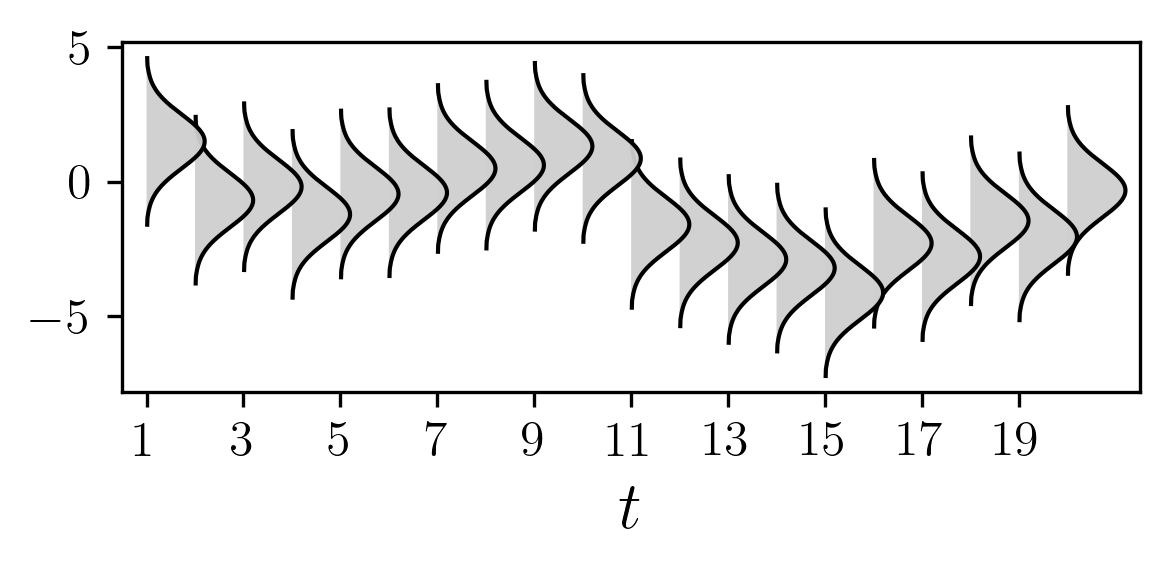}
    \end{minipage}
    \hfill
    \begin{minipage}{0.48\textwidth}
        \centering
        \includegraphics[width=\textwidth]{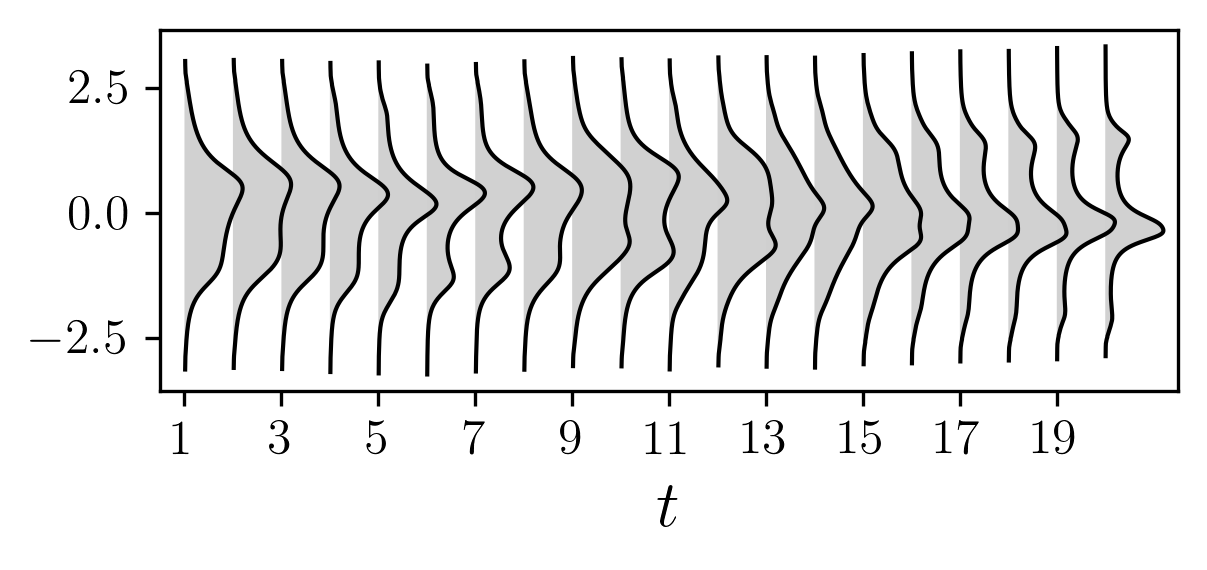}
    \end{minipage}
    \caption{Illustrative sample paths from the WES process, both simulated with $\theta_* = 0.8$, under the random shift map $\T_t^{\mathrm{Shift}}$ (left) and the random sine map $\T_t^{\mathrm{Sine}}$ (right).}
    \label{fig:wes_path}
\end{figure}

We consider two specifications of the random pushforward map. 
The first is a random shift map,
\begin{align}
    \T_t^{\mathrm{Shift}}(x) = x + B_t, \qquad B_t \sim \mathrm{Normal}(0,1),
\end{align}
which produces purely translational perturbations of the latent predictor distribution. 
The resulting observable distributions retain the same overall shape as the predictor, but their locations fluctuate randomly over time. 
The second is a random sine map,
\begin{align}
    \T_t^{\mathrm{Sine}}(x) = \sum_{j=1}^3 W_{t,j}\left\{x - \frac{a\sin\bigl(\pi(x-C_{t,j})\bigr)}{\pi}\right\},
\end{align}
where $a=0.3$, the random centers $C_{t,1},C_{t,2},C_{t,3}$ are independent $\mathrm{Uniform}(-1,1)$ variables, and the weights $W_{t,1},W_{t,2},W_{t,3}$ are random, positive and normalized to sum to one. 
Unlike the shift specification, this map introduces local nonlinear deformations while preserving monotonicity, thereby generating more varied changes in shape over time. 
Both maps satisfy the moment assumptions introduced in \Cref{sec:theory}; verification of these properties is provided in \Cref{apx:simulation_maps}.

\Cref{fig:wes_path} provides representative sample paths under the two map specifications. 
While these examples are intended only as illustrations, they show that the WES process can accommodate markedly different forms of distributional dynamics. 
Under the shift map, the evolution is driven mainly by random location changes, whereas under the sine map the observable distributions may also undergo more irregular shape deformations. 
The two specifications thus probe the estimator under both simple and more complex perturbations.

\begin{figure}[t]
    \centering

    \begin{subfigure}{0.32\textwidth}
        \centering
        \includegraphics[height=140pt]{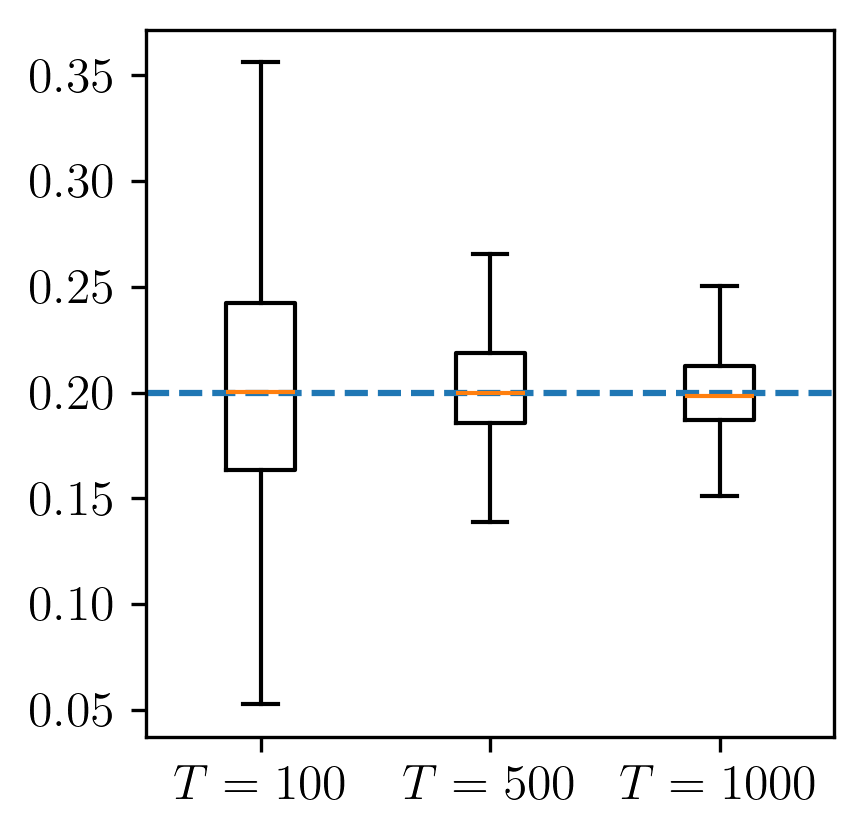}
        \caption{Shift, $\theta_*=0.2$}
    \end{subfigure}
    \begin{subfigure}{0.32\textwidth}
        \centering
        \includegraphics[height=140pt]{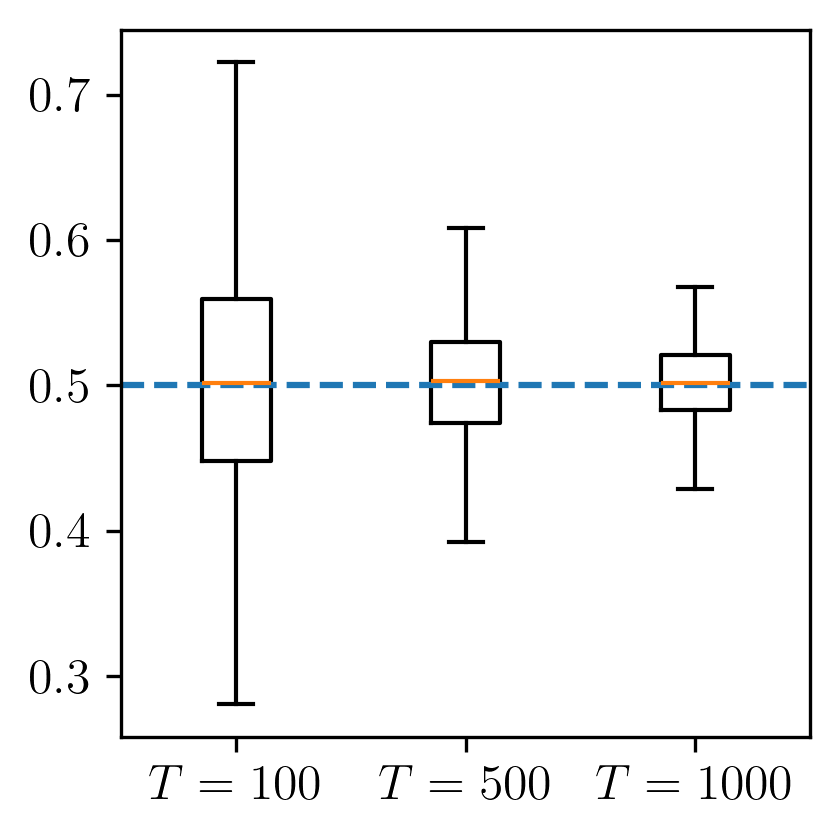}
        \caption{Shift, $\theta_*=0.5$}
    \end{subfigure}
    \begin{subfigure}{0.32\textwidth}
        \centering
        \includegraphics[height=140pt]{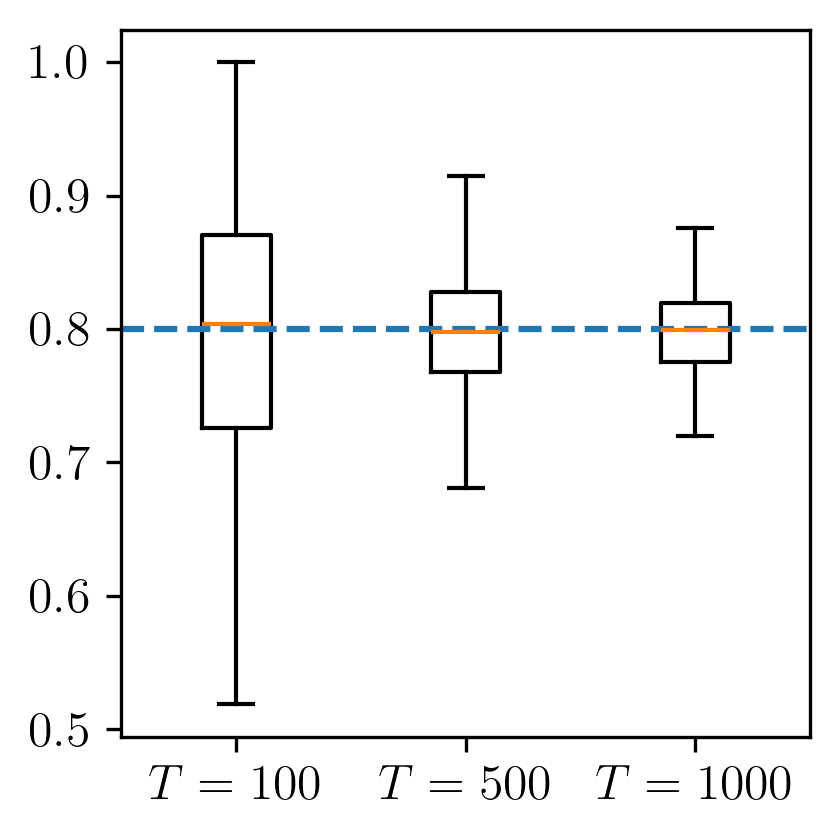}
        \caption{Shift, $\theta_*=0.8$}
    \end{subfigure}

    \vspace{0.5em}

    \begin{subfigure}{0.32\textwidth}
        \centering
        \includegraphics[height=140pt]{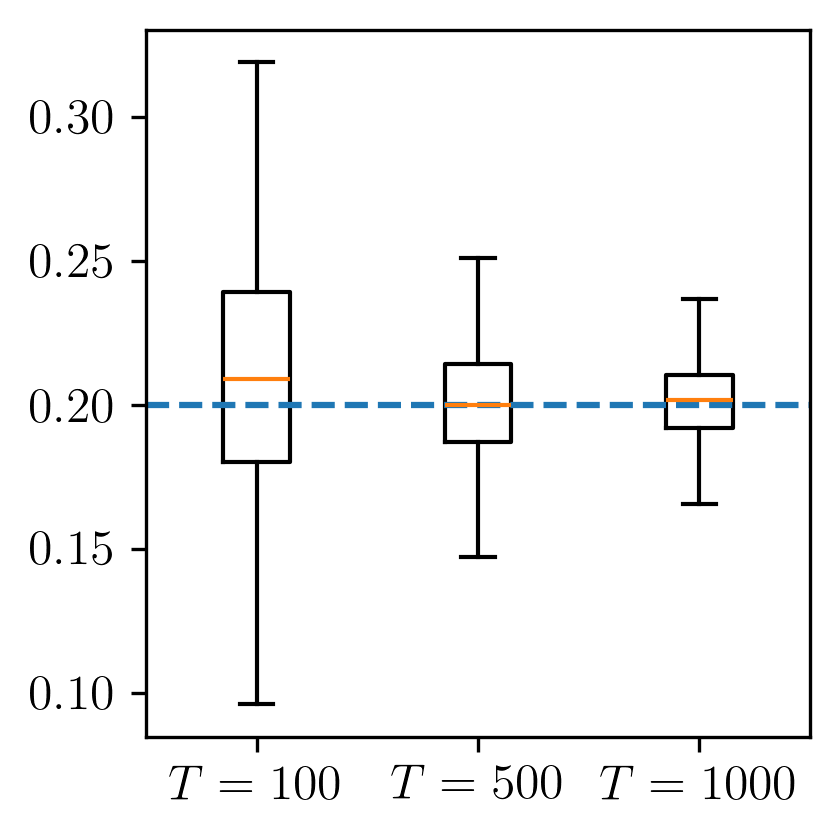}
        \caption{Sine, $\theta_*=0.2$}
    \end{subfigure}
    \begin{subfigure}{0.32\textwidth}
        \centering
        \includegraphics[height=140pt]{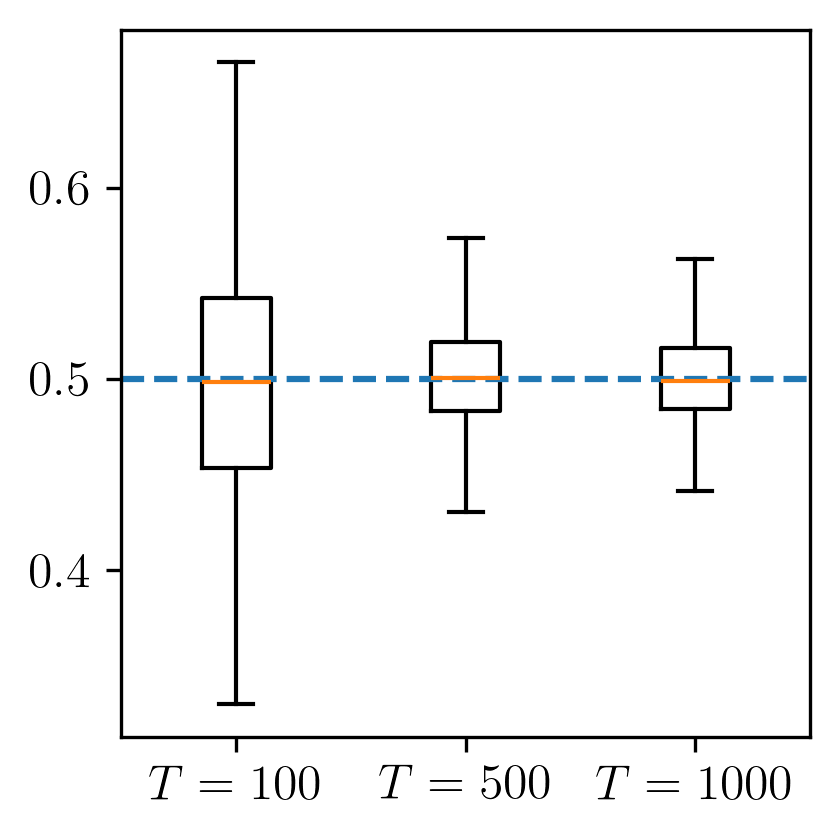}
        \caption{Sine, $\theta_*=0.5$}
    \end{subfigure}
    \begin{subfigure}{0.32\textwidth}
        \centering
        \includegraphics[height=140pt]{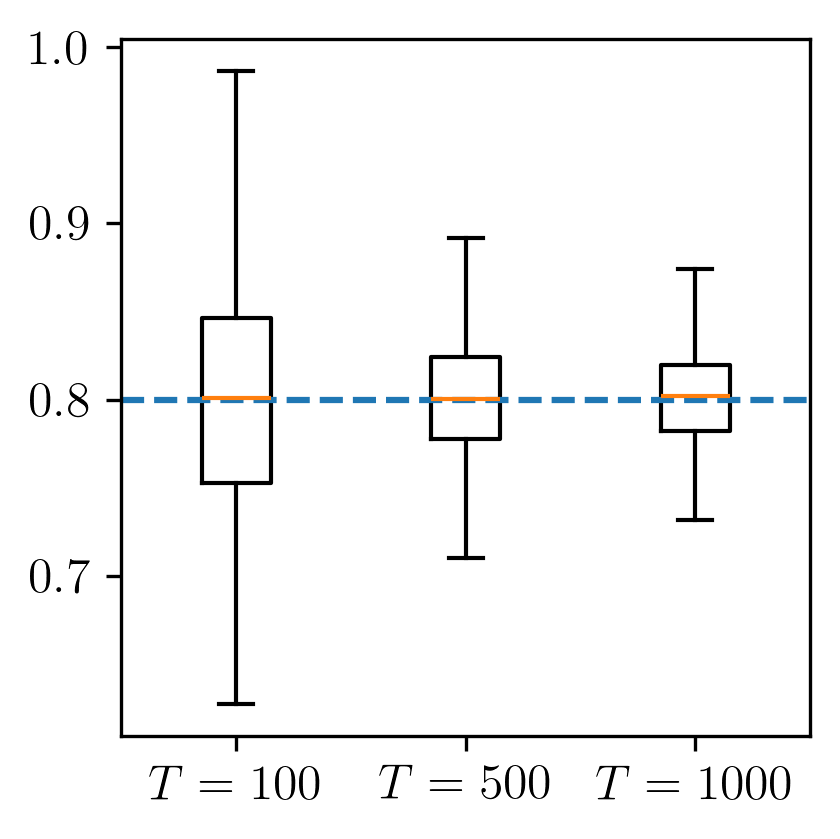}
        \caption{Sine, $\theta_*=0.8$}
    \end{subfigure}

    \caption{Boxplots of the estimated smoothing parameters $\theta_T$ over $500$ replications. The six panels correspond to the three values $\theta_* \in \{0.2,0.5,0.8\}$ and the two random map specifications $\T_t^{\mathrm{Shift}}$ and $\T_t^{\mathrm{Sine}}$. Within each panel, the three boxplots correspond to sample sizes $T=100$, $T=500$ and $T=1000$. The dashed horizontal line indicates the true value of $\theta_*$.}
    \label{fig:that_box}
\end{figure}

To assess the performance of the estimator, we consider combinations of smoothing parameter $\theta_* \in \{0.2,0.5,0.8\}$, random map $\T_t \in \{\T_t^{\mathrm{Shift}},\T_t^{\mathrm{Sine}}\}$, and sample size $T \in \{100,500,1000\}$. 
This yields $3 \times 2 \times 3 = 18$ simulation settings in total. 
For each setting, we obtain the estimate $\theta_T$ by fitting the WES filter to the simulated sample path and minimizing the empirical loss over $\theta \in (0,1)$.
The initial predictor is the standard normal distribution throughout.
We repeat this procedure $500$ times, generating independent sample paths from the WES process.

\Cref{fig:that_box} summarizes the empirical sampling distributions of $\theta_T$ obtained from the repeated experiment in each setting.
Several features are apparent. 
As $T$ increases from $100$ to $1000$, sampling variability consistently decreases, as indicated by the tighter interquartile ranges and whiskers.
Across all cases, the estimates are well centered around the true smoothing parameter $\theta_*$, exhibiting no substantial finite-sample bias even under the nonlinear sine map. 
These patterns provide empirical support for the theoretical consistency established in \Cref{sec:theory}, and suggest that the minimum Wasserstein estimator is accurate and stable across the finite-sample settings considered.

\section{Applications} \label{sec:experiment}

This section evaluates the one-step-ahead forecasting performance of WES in two empirical applications. 
The first involves daily distributions of high-frequency equity-index returns, and the second concerns household energy-demand forecasting using smart-meter electricity consumption data. 
These two applications provide complementary forecasting settings. 
In the equity-index application, the distributional dynamics are associated with volatility clustering, tail behavior and changing asymmetry in high-frequency returns. 
In the household energy-demand application, after regular daily and weekly consumption cycles are removed, the relevant dynamics concern whether a household's residual-demand distribution remains unusually concentrated or dispersed, or shifted upward or downward, across neighboring days.

We compare WES with the four benchmark methods reviewed in \Cref{sec:related_work}: WAR \citep{zhang2022wasserstein}, ATM \citep{Zhu2023}, DAM \citep{ghodrati2024distributional}, and DoDRM \citep{ghodrati2022distribution}. 
The aim is to assess whether the parsimonious recursive structure of WES delivers competitive one-step-ahead forecasts against existing autoregressive and regression-based methods.
In what follows, $\lfloor x \rfloor$ denotes the floor of a scalar $x$.

\subsection{Forecasting Design and Evaluation}

For both applications, we use the following expanding-window forecasting design. 
Let $\nu_1,\ldots,\nu_T$ denote the observed distributional time series of length $T$.
The first $70\%$ of the time series is reserved for the initial in-sample window, where we set $T_0 := \lfloor 0.7 \times T \rfloor$ and $\mathcal{T}_{\mathrm{out}} := \{ T_0+1, \ldots, T \}$.
For each out-of-sample time $t \in \mathcal{T}_{\mathrm{out}}$, each competing method, denoted $m$, constructs a one-step-ahead forecast distribution $\nu_{m,t|t-1}$ using only information available up to time $t-1$. 
We then observe the realized distribution at time $t$ and record the forecast loss.

At each forecast origin, let
\begin{align}
    \ell_{m,t} := W_2\left(\nu_{m,t|t-1},\nu_t\right),
    \qquad t \in \mathcal{T}_{\mathrm{out}} .
\end{align}
The squared Wasserstein loss $\ell_{m,t}^2$ is used for estimation and for the Diebold--Mariano (DM) test \citep{diebold1995comparing} and model confidence set (MCS) procedures described below.
For interpretability, the empirical tables report the mean Wasserstein prediction error (MWPE) of each method $m$:
\begin{equation}
\label{eq:mwpe}
    \operatorname{MWPE}_m := \frac{1}{|\mathcal{T}_{\mathrm{out}}|} \sum_{t\in\mathcal{T}_{\mathrm{out}}} \ell_{m,t} .
\end{equation}
The MWPE is expressed in the same units as the data; lower values indicate more accurate distributional forecasts.

To assess whether differences in out-of-sample performance are statistically meaningful, we supplement the reported MWPE values with pairwise DM tests and the MCS procedure. 
Both procedures are computed from the squared Wasserstein loss sequences $\{\ell_{m,t}^2\}_{t \in \mathcal{T}_{\mathrm{out}}}$. 
The DM test compares each benchmark against WES using the squared-loss differential, with the null hypothesis of equal expected prediction loss. 
We use a two-sided test, since the comparison is intended to detect predictive differences in either direction. 
Serial dependence in the loss differential is accounted for using a Newey--West long-run variance estimator with lag $\lfloor 4(n/100)^{2/9}\rfloor$, where $n := |\mathcal{T}_{\mathrm{out}}|$.
The MCS of \citet{hansen2011model} is applied to the full matrix of squared Wasserstein losses across all competing methods. 
We use the $90\%$ MCS, computed from $10{,}000$ stationary-bootstrap replications with mean block length $\lfloor\sqrt n\rfloor$. 
The retained MCS is insensitive to the stationary-bootstrap block length: varying the mean block length over the range $[n^{1/3},2\sqrt n]$ left it unchanged in both applications.
Since both the DM and MCS procedures are applied to realized scalar loss sequences, they can be used directly in this distributional forecasting setting.

The smoothing parameter $\theta$ of WES is estimated over the current in-sample period by minimizing the average squared Wasserstein one-step-ahead prediction loss. 
The same expanding-window principle is applied to the benchmark methods using their corresponding estimation procedures. 
Model parameters are re-estimated every $20$ out-of-sample forecast origins on the expanding in-sample period, using the same refitting schedule for all methods.
The Python implementation used for the simulations and empirical applications is described in \Cref{appendix:software}.
The path figures below display kernel-density visualizations of the observed and forecast empirical distributions; this smoothing is for visualization only and plays no role in the loss calculations or statistical tests.

\subsection{High-Frequency Equity Index Returns} \label{sec:equity}

The first application concerns the forecasting of daily return distributions for major equity indices. 
This provides a natural setting for distributional time-series methods: each trading day yields an empirical distribution of intraday returns, and the day-to-day evolution of these distributions reflects changing volatility, asymmetry and tail behavior.
The dataset consists of $5$-minute log returns spanning approximately $30$ years for $10$ major global indices: Nikkei 225, Hang Seng Index, TAIEX Index, KOSPI Composite, All Ordinaries, Dow Jones Industrial Average, FTSE 100, DAX Performance Index, FTSE/JSE Top 40 and Tadawul All Share Index. 
The indices cover Asia Pacific, North America, Europe, Africa and the Middle East, so the sample is geographically diverse.

\begin{table}[t]
\centering
\small
\caption{Out-of-sample mean Wasserstein prediction errors (MWPEs) for the equity-index application. Lower values indicate more accurate distributional forecasts. Bold entries are included in the $90\%$ MCS. Significance markers $^{\ast\ast}$ and $^{\ast}$ on benchmark entries indicate that the DM test rejects equal expected squared Wasserstein loss relative to WES at the $1\%$ and $5\%$ levels, respectively.}
\label{tab:equity_oos_significance}
\begin{tabular}{lccccc}
\toprule
        & WAR                  & ATM                  & DAM                 & DoDRM                & WES \\
\midrule
N225    & $0.0317^{\ast\ast}$      & $\mathbf{0.0278}$    & $0.0397^{\ast\ast}$  & $0.0298^{\ast\ast}$  & $\mathbf{0.0257}$ \\
HSI     & $0.0336^{\ast\ast}$  & $0.0323^{\ast\ast}$  & $0.0399^{\ast\ast}$  & $0.0348^{\ast\ast}$  & $\mathbf{0.0291}$ \\
TWII    & $0.0318^{\ast\ast}$  & $0.0269^{\ast\ast}$  & $0.0413^{\ast\ast}$  & $0.0294^{\ast\ast}$  & $\mathbf{0.0241}$ \\
KS11    & $0.0284^{\ast\ast}$  & $0.0254^{\ast\ast}$  & $0.0380^{\ast\ast}$  & $0.0276^{\ast\ast}$  & $\mathbf{0.0232}$ \\
AORD    & $0.0227^{\ast}$      & $0.0228$             & $0.0258^{\ast}$      & $0.0236^{\ast\ast}$  & $\mathbf{0.0205}$ \\
DJI     & $\mathbf{0.0267}^{\ast}$ & $\mathbf{0.0255}$ & $\mathbf{0.0353}^{\ast\ast}$ & $0.0269^{\ast}$ & $\mathbf{0.0243}$ \\
FTSE    & $0.0190^{\ast\ast}$  & $0.0178^{\ast}$      & $0.0244^{\ast\ast}$  & $0.0189^{\ast\ast}$  & $\mathbf{0.0165}$ \\
GDAXI   & $0.0249^{\ast\ast}$  & $0.0228^{\ast}$      & $0.0347^{\ast\ast}$  & $0.0246^{\ast\ast}$  & $\mathbf{0.0213}$ \\
JTOPI   & $0.0230^{\ast\ast}$  & $0.0232^{\ast\ast}$  & $0.0258^{\ast}$      & $0.0247^{\ast\ast}$  & $\mathbf{0.0210}$ \\
TASI    & $0.0232^{\ast\ast}$  & $0.0233^{\ast\ast}$  & $0.0259^{\ast\ast}$  & $0.0240^{\ast\ast}$  & $\mathbf{0.0210}$ \\
\midrule
\multicolumn{6}{l}{\emph{Summary across 10 indices}} \\
In $90\%$ MCS                & $1$/10 & $2$/10 & $1$/10 & $0$/10 & $10$/10 \\
DM rejects vs.\ WES at $5\%$ & $10$/10 & $7$/10 & $10$/10 & $10$/10 & --- \\
\bottomrule
\end{tabular}

\end{table}
\begin{figure}[t]
    \centering
    \begin{subfigure}[b]{0.7\textwidth}
        \centering
        \includegraphics[width=\linewidth]{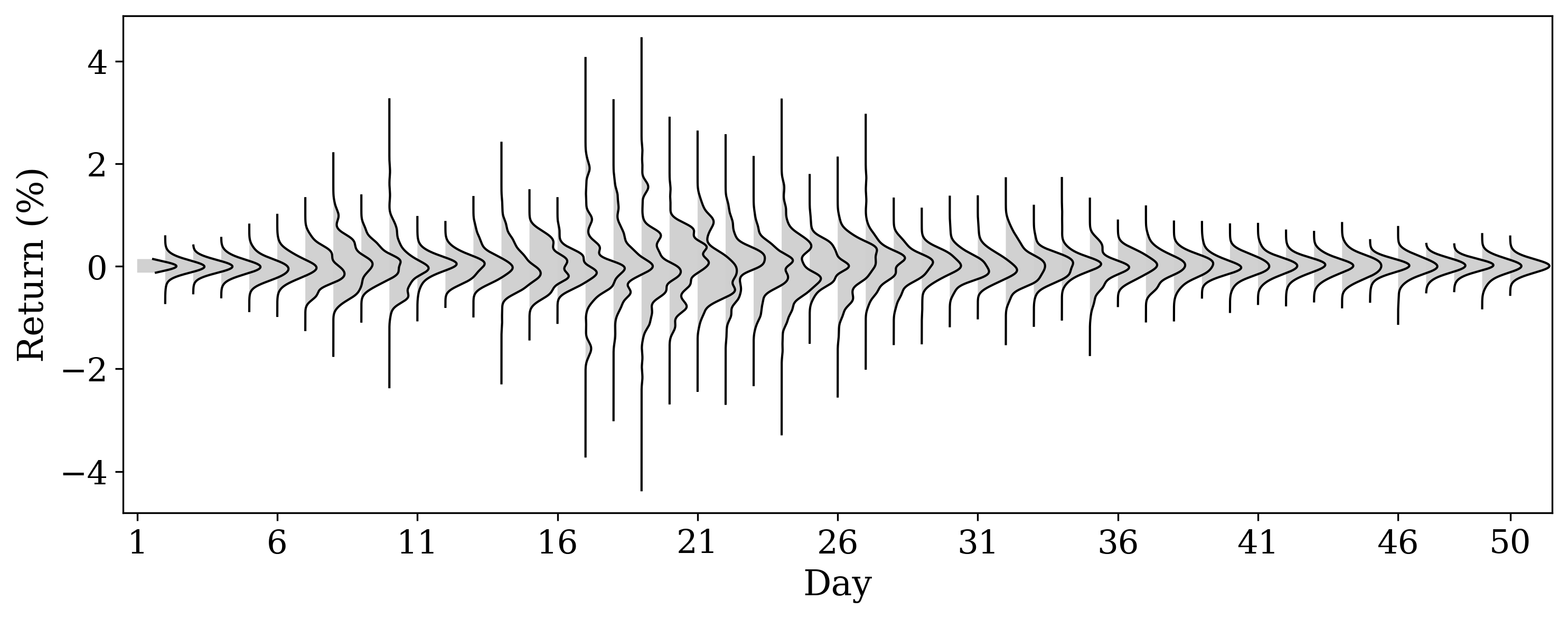}
        \label{fig:equity_observed}
    \end{subfigure}

    \begin{subfigure}[b]{0.7\textwidth}
        \centering
        \includegraphics[width=\linewidth]{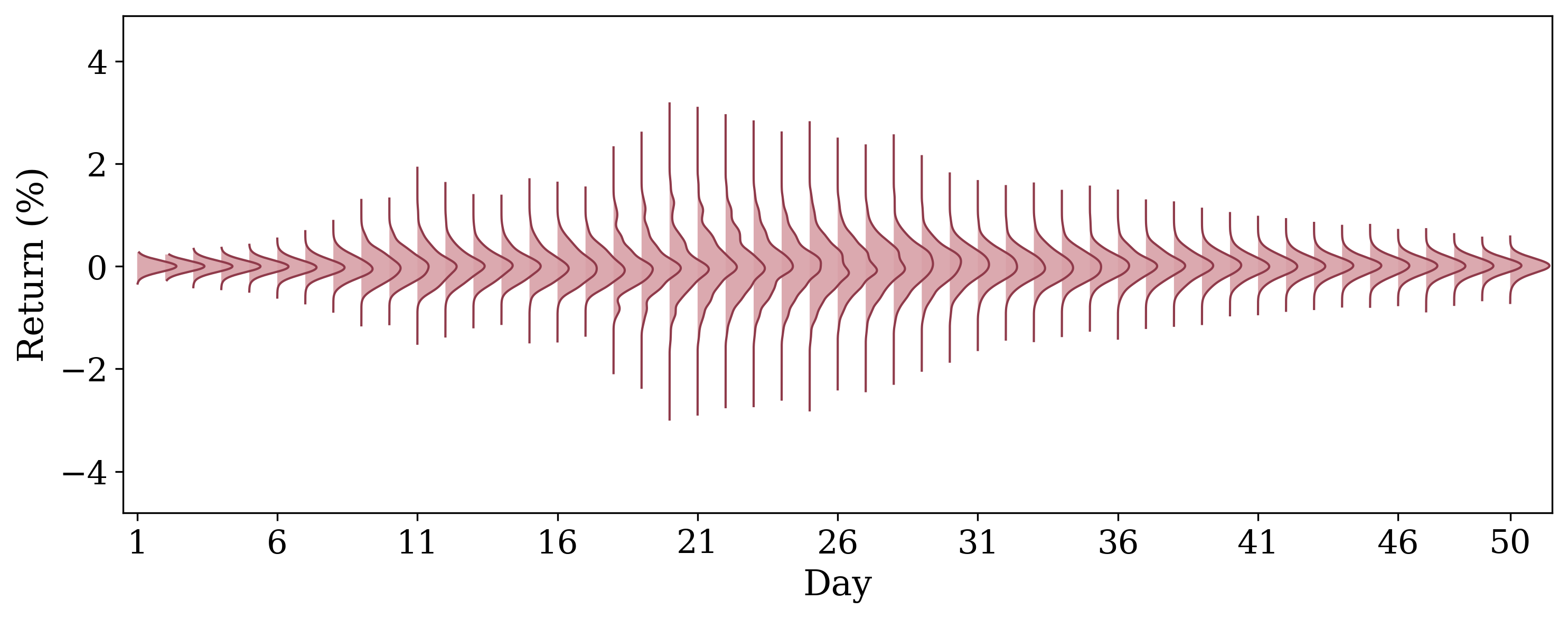}
        \label{fig:equity_wes}
    \end{subfigure}

    \caption{Observed (top) and WES-forecast (bottom) daily distributions of $5$-minute log returns for the Dow Jones Industrial Average index over a period containing the early COVID-19 market disruption.}
    \label{fig:equity_combined}
\end{figure}

The number of intraday observations varies across trading days. 
To construct comparable daily distributions, we retain only trading days whose number of $5$-minute returns is at least $50\%$ of the modal daily count for that index. 
This rule filters out shortened or non-comparable trading sessions while retaining ordinary full-session days.
The resulting sample sizes range from $T = 7000$ to $T = 7500$, except for JTOPI ($T = 5816$) and TASI ($T = 4112$).

\Cref{tab:equity_oos_significance} reports the out-of-sample forecasting results. 
WES attains the smallest MWPE for all $10$ indices, remaining in the $90\%$ MCS in every case and serving as the unique MCS member in $8$ cases. 
The DM tests reject equal predictive accuracy in favor of WES at the $5\%$ level for WAR in all $10$ indices, ATM in $7$, DAM in $10$, and DoDRM in $10$. 
The DM and MCS results indicate that these differences are statistically meaningful in most cases.

\Cref{fig:equity_combined} illustrates the WES forecast path for the Dow Jones Industrial Average index over the period from February 19, 2020, to April 29, 2020. The observed distributions capture the early phase of the COVID-19 market disruption: the return distributions widen sharply from around day $12$ and remain dispersed before gradually returning toward their pre-crisis shape. The WES forecast path follows this volatility burst and subsequent recovery, with a final estimated smoothing parameter of $\theta_T = 0.297$.

\subsection{Household Energy-Demand Forecasting}
\label{sec:smart_meter}

The second application concerns household energy-demand forecasting using the SmartMeter Energy Consumption Data in London Households dataset from the Low Carbon London smart-meter trial. 
This dataset yields a distributional forecasting problem in which each day provides repeated within-household readings, so that each daily observation is itself a distribution. 
The distributional object is the daily distribution of residual electricity demand within a household, after removing regular intraday and weekly consumption patterns.

The trial recorded household electricity consumption in kilowatt-hours every half hour between late 2011 and early 2014, yielding $48$ readings per household per day. 
We analyze $10$ household-level series with at least $500$ days of complete half-hourly readings.
The resulting sample sizes range from $T=524$ to $T=816$ days. 
Raw electricity consumption contains strong intraday and weekly periodicity. 
To focus the forecasting exercise on residual distributional dynamics rather than deterministic usage cycles, we remove a simple in-sample seasonal baseline. 
For each household, we estimate the mean consumption for each day-of-week and half-hour-slot combination over the initial $120$ days. 
Subtracting this baseline from the raw readings yields $48$ residual demands per day. 
Each observation $\nu_t$ is the empirical distribution of these $48$ residuals.

\Cref{tab:smart_meter_oos_significance} reports the out-of-sample forecasting results. 
As in the equity-index application, WES attains the smallest MWPE for all $10$ households and remains in the $90\%$ MCS in every case.
The DM tests reject equal predictive accuracy in favor of WES at the $5\%$ level for WAR in $5$ households, ATM in $8$, DAM in $9$, and DoDRM in $9$.
Among the benchmarks, WAR remains in the $90\%$ MCS for $4$ of the $10$ households, ATM and DAM for $1$ each, and DoDRM for none.
The application is nonetheless substantively different from the equity-index setting: whereas return distributions reflect volatility and tail-behavior persistence, residual-demand distributions reflect persistence in household-level consumption patterns after regular intraday and weekly cycles have been removed.

\Cref{fig:smart_meter_combined} shows, for household MAC001989, the observed and WES-forecast residual-demand distributions over $50$ consecutive days from December 24, 2013, to February 11, 2014.
The observed distributions display substantial day-to-day variation in shape, including a marked concentration around days $41$--$50$. 
The WES forecast path tracks these distributional changes closely, with a final estimated smoothing parameter of $\theta_T = 0.699$.

\begin{table}[t]
\centering
\small
\caption{Out-of-sample mean Wasserstein prediction errors (MWPEs) for the household energy-demand application. Lower values indicate more accurate distributional forecasts. Bold entries are included in the $90\%$ MCS. Significance markers $^{\ast\ast}$ and $^{\ast}$ on benchmark entries indicate that the DM test rejects equal expected squared Wasserstein loss relative to WES at the $1\%$ and $5\%$ levels, respectively.}
\label{tab:smart_meter_oos_significance}
\begin{tabular}{lccccc}
\toprule
        & WAR                  & ATM                  & DAM                 & DoDRM                & WES \\
\midrule
MAC001989 & $\mathbf{0.1914}$ & $0.1919^{\ast\ast}$ & $0.3336^{\ast\ast}$ & $0.1933$ & $\mathbf{0.1897}$ \\
MAC005116 & $\mathbf{0.0653}$ & $0.0758^{\ast\ast}$ & $\mathbf{0.0654}$ & $0.0776^{\ast\ast}$ & $\mathbf{0.0646}$ \\
MAC003031 & $0.0064^{\ast\ast}$ & $0.0061^{\ast\ast}$ & $0.0078^{\ast\ast}$ & $0.0067^{\ast\ast}$ & $\mathbf{0.0053}$ \\
MAC002150 & $0.1057^{\ast\ast}$ & $0.1033^{\ast\ast}$ & $0.1296^{\ast\ast}$ & $0.1146^{\ast\ast}$ & $\mathbf{0.0940}$ \\
MAC004208 & $0.1097^{\ast\ast}$ & $0.1108^{\ast\ast}$ & $0.1147^{\ast\ast}$ & $0.1267^{\ast\ast}$ & $\mathbf{0.0956}$ \\
MAC003212 & $0.0776^{\ast\ast}$ & $0.0716$ & $0.0849^{\ast\ast}$ & $0.0792^{\ast\ast}$ & $\mathbf{0.0675}$ \\
MAC002916 & $\mathbf{0.0162}$ & $\mathbf{0.0145}$ & $0.0188^{\ast\ast}$ & $0.0163^{\ast}$ & $\mathbf{0.0134}$ \\
MAC005260 & $\mathbf{0.0249}$ & $0.0270^{\ast\ast}$ & $0.0261^{\ast}$ & $0.0268^{\ast\ast}$ & $\mathbf{0.0241}$ \\
MAC000174 & $0.1042^{\ast\ast}$ & $0.1126^{\ast\ast}$ & $0.1068^{\ast\ast}$ & $0.1199^{\ast\ast}$ & $\mathbf{0.0977}$ \\
MAC001251 & $0.1266$ & $0.1433^{\ast\ast}$ & $0.1340^{\ast\ast}$ & $0.1484^{\ast\ast}$ & $\mathbf{0.1241}$ \\
\midrule
\multicolumn{6}{l}{\emph{Summary across 10 households}} \\
In $90\%$ MCS                & $4$/10 & $1$/10 & $1$/10 & $0$/10 & $10$/10 \\
DM rejects vs.\ WES at $5\%$ & $5$/10 & $8$/10 & $9$/10 & $9$/10 & --- \\
\bottomrule
\end{tabular}
\end{table}

\begin{figure}[t]
    \centering
    \begin{subfigure}[b]{0.7\textwidth}
        \centering
        \includegraphics[width=\linewidth]{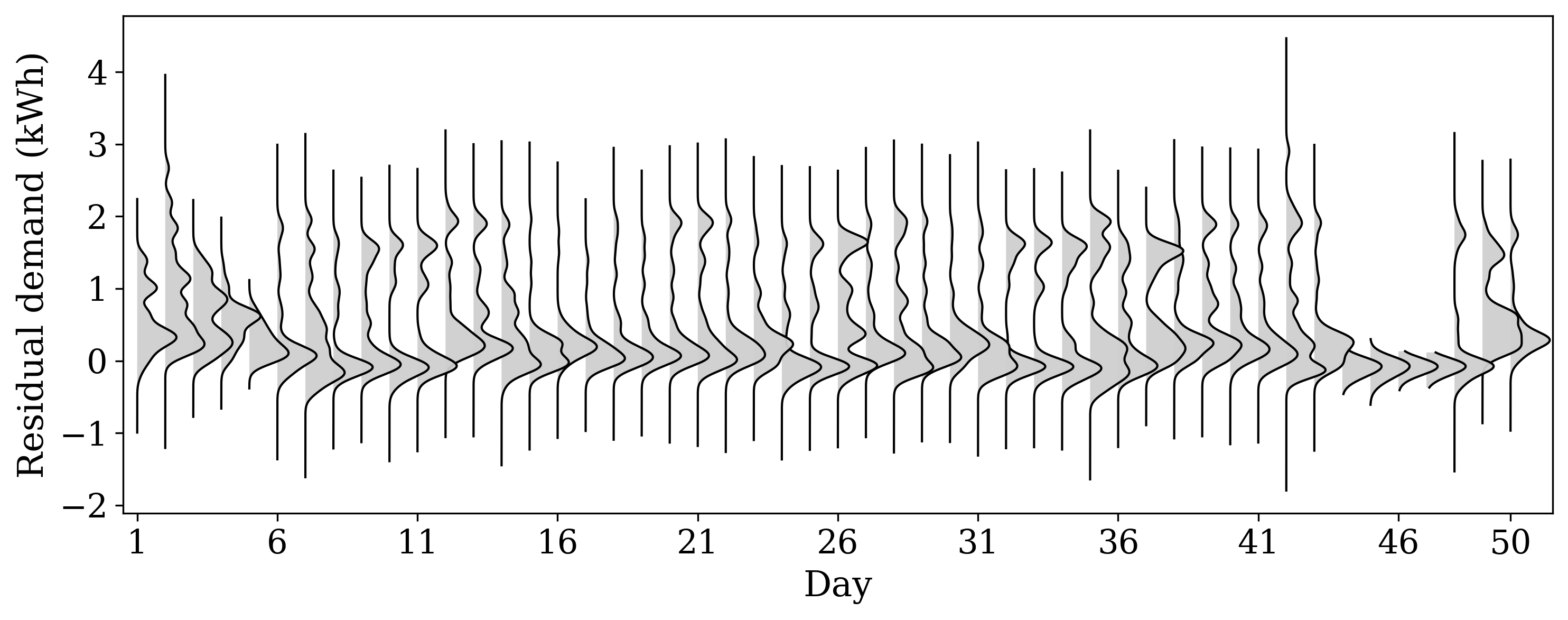}
        \label{fig:smart_meter_observed}
    \end{subfigure}

    \begin{subfigure}[b]{0.7\textwidth}
        \centering
        \includegraphics[width=\linewidth]{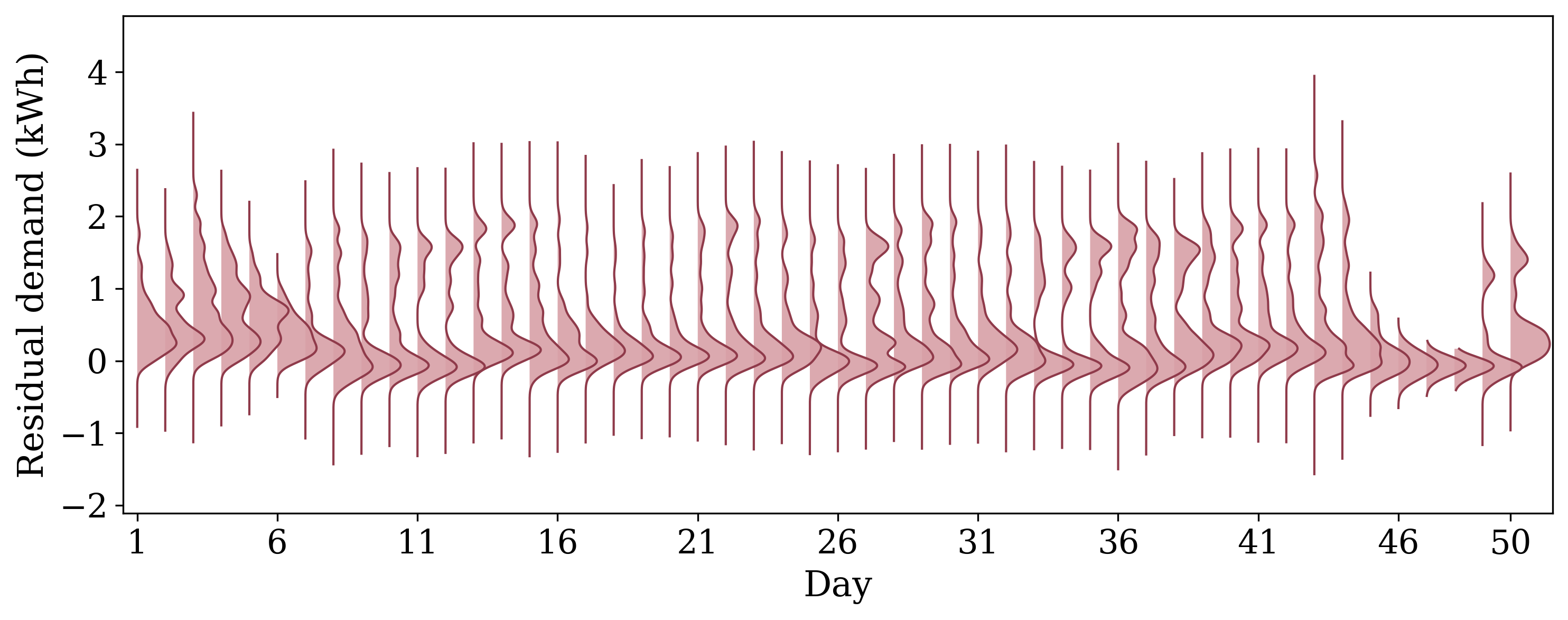}
        \label{fig:smart_meter_wes}
    \end{subfigure}

    \caption{Observed (top) and WES-forecast (bottom) daily residual-demand distributions for household MAC001989 over $50$ consecutive days.}
    \label{fig:smart_meter_combined}
\end{figure}


\section{Conclusion} \label{sec:conclusion}

Methodological work on distributional time series has concentrated largely on autoregressive structures. 
This leaves room for simpler recursive forecasting methods that play a role analogous to exponential smoothing in scalar time-series analysis. 
This paper proposed Wasserstein exponential smoothing (WES), a one-parameter forecasting method for distributional time series on $\R$. 
The method extends the practical logic of classical exponential smoothing to probability distributions by updating forecast distributions along Wasserstein geodesics.

The tractability of WES comes from the structure of the one-dimensional Wasserstein space. 
By using the isometry between the Wasserstein distance on $\R$ and the $L^2$ distance between quantile functions, the recursive update can be represented and implemented in quantile space. 
This provides a computationally simple filter while retaining a direct geometric interpretation in Wasserstein space. 
The theoretical analysis studied the sample-path properties of the resulting distributional dynamics and established consistency of the minimum Wasserstein estimator of the smoothing parameter under the stated regularity and identifiability conditions.

The empirical analysis shows that this parsimonious recursive structure is competitive with, and in these applications more accurate than, existing autoregressive and regression-based benchmarks for one-step-ahead distributional forecasting. 
In the equity-index application, WES produced lower out-of-sample MWPEs than the autoregressive distributional benchmarks for forecasting daily distributions of high-frequency returns. 
In the household energy-demand application, WES showed a similar pattern for daily residual-demand distributions constructed from smart-meter data. 
These two applications illustrate how the same recursive Wasserstein update can be applied in settings with different sources of distributional persistence.

The current formulation is restricted to one-dimensional probability distributions. 
This restriction is important because both the implementation and the theoretical analysis rely on the exact isometry between the Wasserstein metric and the $L^2$ distance between quantile functions, which does not generalize straightforwardly to multivariate settings. 
Extending the approach to distributions on $\R^d$ for $d>1$ would require computationally efficient approximations or alternative representations of Wasserstein geometry that preserve the simplicity of the recursive update. 
Another natural direction is to expand the current model to include an explicit trend, seasonal components, or covariates, bringing the approach closer to the broader exponential-smoothing toolkit used in scalar forecasting.

\bibliographystyle{abbrvnat}
\bibliography{ref} 

@book{villani2009optimal,
  title={Optimal transport: Old and New},
  author={Villani, C{\'e}dric},
  year={2009},
  publisher={Springer}
}

@book{Panaretos2020,
	author = {Victor M. Panaretos and Yoav Zemel},
	title = {An Invitation to Statistics in {W}asserstein Space},
	publisher = {Springer},
	year = {2020},
}

@incollection{Newey1994,
	title = {Large Sample Estimation and Hypothesis Testing},
	booktitle = {Handbook of Econometrics},
	chapter = {36},
	publisher = {Elsevier},
	volume = {4},
	pages = {2111--2245},
	year = {1994},
	author = {Whitney K. Newey and Daniel McFadden},
}

@book{Vaart1998, 
	place={Cambridge}, 
	series={Cambridge Series in Statistical and Probabilistic Mathematics}, 
	title={Asymptotic Statistics}, 
	publisher={Cambridge University Press}, 
	author={Vaart, A. W. van der}, 
	year={1998}, 
	collection={Cambridge Series in Statistical and Probabilistic Mathematics}
}

@article{Frechet1948,
    author = {Fr\'{e}chet, Maurice},
    journal = {Annales de l'institut Henri Poincaré},
    number = {4},
    pages = {215-310},
    publisher = {INSTITUT HENRI POINCARÉ ET GAUTHIER-VILLARS},
    title = {Les éléments aléatoires de nature quelconque dans un espace distancié},
    volume = {10},
    year = {1948},
}

@ARTICLE{Zhu2023,
  title   = "Autoregressive optimal transport models",
  author  = "Zhu, Changbo and Müller, Hans-Georg",
  journal = "Journal of the Royal Statistical Society Series B: Statistical
             Methodology",
  volume  =  85,
  number  =  3,
  pages   = "1012–-1033",
  year    =  2023
}

@book{Rudin1976, 
	title={Principles of Mathematical Analysis}, 
	publisher={McGraw-Hill}, 
	author={Walter Rudin}, 
	year={1976}, 
}

@article{chen2023wasserstein,
  title={Wasserstein regression},
  author={Chen, Yaqing and Lin, Zhenhua and M{\"u}ller, Hans-Georg},
  journal={Journal of the American Statistical Association},
  volume={118},
  number={542},
  pages={869--882},
  year={2023},
  publisher={Taylor \& Francis}
}

@article{ghodrati2022distribution,
  title={Distribution-on-distribution regression via optimal transport maps},
  author={Ghodrati, Laya and Panaretos, Victor M},
  journal={Biometrika},
  volume={109},
  number={4},
  pages={957--974},
  year={2022},
  publisher={Oxford University Press}
}

@article{zhang2022wasserstein,
  title={Wasserstein autoregressive models for density time series},
  author={Zhang, Chao and Kokoszka, Piotr and Petersen, Alexander},
  journal={Journal of Time Series Analysis},
  volume={43},
  number={1},
  pages={30--52},
  year={2022},
  publisher={Wiley Online Library}
}

@article{Gardner1985,
  author  = {Gardner, Everette S., Jr.},
  title   = {Exponential Smoothing: The State of the Art},
  journal = {Journal of Forecasting},
  year    = {1985},
  volume  = {4},
  number  = {1},
  pages   = {1--28},
}

@article{Gardner2006,
  author  = {Gardner, Everette S., Jr.},
  title   = {Exponential Smoothing: The State of the Art---Part {II}},
  journal = {International Journal of Forecasting},
  year    = {2006},
  volume  = {22},
  number  = {4},
  pages   = {637--666},
}

@article{Muth1960,
  author  = {Muth, John F.},
  title   = {Optimal Properties of Exponentially Weighted Forecasts},
  journal = {Journal of the American Statistical Association},
  year    = {1960},
  volume  = {55},
  number  = {290},
  pages   = {299--306},
}

@book{Hyndman2008,
  author    = {Hyndman, Rob J. and Koehler, Anne B. and Ord, J. Keith and Snyder, Ralph D.},
  title     = {Forecasting with Exponential Smoothing: The State Space Approach},
  publisher = {Springer},
  address   = {Berlin},
  year      = {2008}
}

@book{Ambrosio2024,
  author    = {Luigi Ambrosio and Elia Brué and Daniele Semola.},
  title     = {Lectures on Optimal Transport},
  publisher = {Springer},
  year      = {2024}
}

@article{ghodrati2024distributional,
  title={On distributional autoregression and iterated transportation},
  author={Ghodrati, Laya and Panaretos, Victor M},
  journal={Journal of Time Series Analysis},
  volume={45},
  number={5},
  pages={739--770},
  year={2024},
  publisher={Wiley Online Library}
}

@article{diebold1995comparing,
 ISSN = {07350015},
 URL = {http://www.jstor.org/stable/1392185},
 author = {Francis X. Diebold and Roberto S. Mariano},
 journal = {Journal of Business \& Economic Statistics},
 number = {3},
 pages = {253--263},
 publisher = {American Statistical Association, Taylor \& Francis, Ltd.},
 title = {Comparing Predictive Accuracy},
 urldate = {2026-04-25},
 volume = {13},
 year = {1995}
}

@article{hansen2011model,
 author = {Peter R. Hansen and Asger Lunde and James M. Nason},
 journal = {Econometrica},
 number = {2},
 pages = {453--497},
 publisher = {[Wiley, Econometric Society]},
 title = {THE MODEL CONFIDENCE SET},
 urldate = {2026-04-25},
 volume = {79},
 year = {2011}
}

@article{Diebold1998Density,
  author  = {Diebold, Francis X. and Gunther, Todd A. and Tay, Anthony S.},
  title   = {Evaluating Density Forecasts with Applications to Financial Risk Management},
  journal = {International Economic Review},
  year    = {1998},
  volume  = {39},
  number  = {4},
  pages   = {863--883},
  doi     = {10.2307/2527342}
}

@article{TayWallis2000,
  author  = {Tay, Anthony S. and Wallis, Kenneth F.},
  title   = {Density Forecasting: A Survey},
  journal = {Journal of Forecasting},
  year    = {2000},
  volume  = {19},
  number  = {4},
  pages   = {235--254},
  doi     = {10.1002/1099-131X(200007)19:4<235::AID-FOR772>3.0.CO;2-L}
}

@article{GneitingKatzfuss2014,
  author  = {Gneiting, Tilmann and Katzfuss, Matthias},
  title   = {Probabilistic Forecasting},
  journal = {Annual Review of Statistics and Its Application},
  year    = {2014},
  volume  = {1},
  pages   = {125--151},
  doi     = {10.1146/annurev-statistics-062713-085831}
}

@article{HyndmanUllah2007,
  author  = {Hyndman, Rob J. and Ullah, Md Shahid},
  title   = {Robust Forecasting of Mortality and Fertility Rates: A Functional Data Approach},
  journal = {Computational Statistics \& Data Analysis},
  year    = {2007},
  volume  = {51},
  number  = {10},
  pages   = {4942--4956},
  doi     = {10.1016/j.csda.2006.07.028}
}

@book{HorvathKokoszka2012,
  author    = {Horv{\'a}th, Lajos and Kokoszka, Piotr},
  title     = {Inference for Functional Data with Applications},
  series    = {Springer Series in Statistics},
  publisher = {Springer},
  address   = {New York},
  year      = {2012},
  doi       = {10.1007/978-1-4614-3655-3}
}

@article{Chen2022DQF,
  author  = {Chen, Wilson Ye and Peters, Gareth W. and Gerlach, Richard H. and Sisson, Scott A.},
  title   = {Dynamic Quantile Function Models},
  journal = {Quantitative Finance},
  year    = {2022},
  volume  = {22},
  number  = {9},
  pages   = {1665--1691},
  doi     = {10.1080/14697688.2022.2053193}
}

@article{HongFan2016,
  author  = {Hong, Tao and Fan, Shu},
  title   = {Probabilistic Electric Load Forecasting: A Tutorial Review},
  journal = {International Journal of Forecasting},
  year    = {2016},
  volume  = {32},
  number  = {3},
  pages   = {914--938},
  doi     = {10.1016/j.ijforecast.2015.11.011}
}

@article{BillardDiday2003,
  author  = {Billard, Lynne and Diday, Edwin},
  title   = {From the Statistics of Data to the Statistics of Knowledge: Symbolic Data Analysis},
  journal = {Journal of the American Statistical Association},
  year    = {2003},
  volume  = {98},
  number  = {462},
  pages   = {470--487},
  doi     = {10.1198/016214503000242}
}

@book{BillardDiday2006,
  author    = {Billard, Lynne and Diday, Edwin},
  title     = {Symbolic Data Analysis: Conceptual Statistics and Data Mining},
  publisher = {John Wiley \& Sons},
  address   = {Chichester},
  year      = {2006},
  doi       = {10.1002/9780470090183}
}

@article{ArroyoMate2009,
  author  = {Arroyo, Javier and Mat{\'e}, Carlos},
  title   = {Forecasting Histogram Time Series with k-Nearest Neighbours Methods},
  journal = {International Journal of Forecasting},
  year    = {2009},
  volume  = {25},
  number  = {1},
  pages   = {192--207},
  doi     = {10.1016/j.ijforecast.2008.07.003}
}

@article{ArroyoEspinolaMate2011,
  author  = {Arroyo, Javier and Esp{\'i}nola, Rosa and Mat{\'e}, Carlos},
  title   = {Different Approaches to Forecast Interval Time Series: A Comparison in Finance},
  journal = {Computational Economics},
  year    = {2011},
  volume  = {37},
  number  = {2},
  pages   = {169--191},
  doi     = {10.1007/s10614-010-9230-2}
}

@article{ArroyoGonzalezRiveraMateMunoz2011,
  author  = {Arroyo, Javier and Gonz{\'a}lez-Rivera, Gloria and Mat{\'e}, Carlos and Mu{\~n}oz San Roque, Antonio},
  title   = {Smoothing Methods for Histogram-Valued Time Series: An Application to Value-at-Risk},
  journal = {Statistical Analysis and Data Mining},
  year    = {2011},
  volume  = {4},
  number  = {2},
  pages   = {216--228},
  doi     = {10.1002/sam.10114}
}

@article{GonzalezRiveraArroyo2012,
  author  = {Gonz{\'a}lez-Rivera, Gloria and Arroyo, Javier},
  title   = {Time Series Modeling of Histogram-Valued Data: The Daily Histogram Time Series of S\&P500 Intradaily Returns},
  journal = {International Journal of Forecasting},
  year    = {2012},
  volume  = {28},
  number  = {1},
  pages   = {20--33},
  doi     = {10.1016/j.ijforecast.2011.02.007}
}

@article{MaiaCarvalho2011,
  author  = {Maia, Andr{\'e} Luis Santiago and de Carvalho, Francisco de A. T.},
  title   = {Holt's Exponential Smoothing and Neural Network Models for Forecasting Interval-Valued Time Series},
  journal = {International Journal of Forecasting},
  year    = {2011},
  volume  = {27},
  number  = {3},
  pages   = {740--759},
  doi     = {10.1016/j.ijforecast.2010.02.012}
}

\newpage
\appendix

\setcounter{figure}{0}
\setcounter{table}{0}
\setcounter{equation}{0}

\renewcommand{\thefigure}{S\arabic{figure}}
\renewcommand{\thetable}{S\arabic{table}}
\renewcommand{\theequation}{S\arabic{equation}}

\begin{center}
	\LARGE \textbf{Supplementary Material}
\end{center}

\vspace{30pt}

This appendix contains proofs of all the theoretical results presented in the main text and additional experimental details.
To establish the theoretical results, we first introduce the essential notation and lemmas upon which the proofs are constructed.
\Cref{apx:preparation} establishes this necessary prerequisites for the main proofs.
\Cref{apx:preparation_proofs} lists the proofs of the prerequisite lemmas introduced in \Cref{apx:preparation}.
Built on them, \Cref{apx:proofs} provides the main proofs of all the theoretical results presented in the main text.
\Cref{apx:experimental_detail} contains the additional details of the experiments.
Finally, \Cref{appendix:software} concisely introduces the basic usage of the Python package for WES.

For convenience, we summarize below notations adopted throughout this appendix.

\begin{table}[h]
	\caption{Common Notations in Appendix.} \label{table:1}
	\centering
	\begin{tabular}{ll} 
		\toprule
		\multicolumn{2}{l}{\underline{Function Space:}} \\
		$L^2$ & the Lebesgue space of square-integrable functions in the unit interval $(0, 1)$ \\
		$\| \cdot \|_2$ & the norm of the $L^2$ space, i.e., $\| f \|_2^2 = \int_{(0,1)} f(q)^2 dq$ \\
		$\langle \cdot, \cdot \rangle_2$ & the inner product of the $L^2$ space, i.e., $\langle f, g \rangle_2 = \int_{(0,1)} f(q) g(q) dq$ \\
		\multicolumn{2}{l}{\underline{Parameter:}} \\
		$\theta_*$ & the fixed parameter value of the WES process \\
		$\theta$ & the parameter of the WES filter to be estimated \\
		\multicolumn{2}{l}{\underline{Random Variable:}} \\
		$\nu_t$ & the observable variable of the WES process at $\theta_*$ \\
		$\mu_t$ & the predictor variable of the WES process at $\theta_*$ \\
		$\mu_t^\theta$ & the predictor variables of the WES filter at $\theta$ conditional on $\{ \nu_t \}_{t=1}^{T}$ \\
		$V_t$ & the quantile function of the observable variable $\nu_t$ \\
		$U_t$ & the quantile function of the predictor variable $\mu_t$ \\
		$U_t^\theta$ & the quantile function of the predictor variable $\mu_t^\theta$ \\
		$\T_t$ & the random pushforward map from $\R$ to $\R$ used in the WES process \\
		$F_t$ & the quantile residual to be defined in \Cref{def:quantile_residual} in \Cref{sec:WES_error_correction_form} \\
		$G_t^\theta$ & the quantile residual cumulative at $\theta$ to be defined in \Cref{def:residual_cumulative} in \Cref{apx:preparation} \\
		\multicolumn{2}{l}{\underline{Initial Value:}} \\
		$\mu_0$ & the fixed, initial predictor distribution of the WES process \\
		$\mu_0^\theta$ & the fixed, initial predictor distribution of the WES filter \\
        $U_0$ & the quantile function of the fixed initial condition $\mu_0$ \\
		$U_0^\theta$ & the quantile function of the fixed initial condition $\mu_0^\theta$ \\
		\bottomrule
	\end{tabular}
\end{table}

\section{Preliminaries for Main Proofs} \label{apx:preparation}

This section introduces essential notation and intermediate lemmas required for the main proofs.
\Cref{apx:residual} investigates the moment properties of the quantile residual $F_{t}$ introduced in the main text.
\Cref{apx:error_decomposition} further analyzes the moment properties of the temporal cumulative sum of the quantile residuals $F_{t}$, demonstrating that the model-quantile residual in \Cref{sec:theory} can be expressed via this sum.
Proofs of all the intermediate lemmas are contained in \Cref{apx:preparation_proofs}.

\subsection{Quantile Residual and Its Moment Properties} \label{apx:residual}

Since we consider univariate probability distributions in this work, every distribution admits a unique quantile function.
For any distribution $\mu$ in the 2-Wasserstein space, its quantile function $U$ has a finite $L^2$-norm. This follows from the change of variables $x = U(q)$ for $q \in (0, 1)$:
\begin{align}
	\| U \|_2^2 = \int_{(0,1)} U(q)^2 d q = \int_{\mathbb{R}} x^2 d \mu(x) < \infty ,
\end{align}
where the last inequality holds because $\mu$ in $\mathbb{W}_2(\R)$ admits the second moment.

Recall that the random observable $\nu_{t}$ is generated from the random predictor $\mu_{t-1}$ under the pushforward formula \eqref{eq:WES_process_formula 1}.
This pushforward formula can be expressed in terms of the quantile function:
\begin{align}
	V_{t}(q) = \mathcal{E}_{t}\big( U_{t-1}(q) \big) = \mathcal{E}_{t} \circ U_{t-1}(q) .
\end{align}
For convenience, we restate the definition of the quantile residual:
\begin{align}
	F_{t}(q) :=  \mathcal{E}_{t} \circ U_{t-1}(q) - U_{t-1}(q) \quad \text{at each} \quad q \in (0, 1).\label{eq:quantile residual redefined}
\end{align}
Observe that the predictor $\mu_{t-1}$ is a random element determined by the sequence of random transport maps $\T_1, \dots, \T_{t-1}$ via the iterations of \eqref{eq:WES_process_formula 1} and \eqref{eq:WES_process_formula 2}.
Consequently, the quantile residual $F_{t}$ given by \eqref{eq:quantile residual redefined} is a random element dependent on the maps $\T_1, \dots, \T_{t}$.

The following lemma establishes several fundamental properties of the quantile residual $F_{t}$.

\begin{lemma} \label{lem:residual_property}
	For each $t \ge 1$, the quantile residual $F_{t}$ is a random element defined by the sequence of random transport maps $\T_1, \dots, \T_{t}$.
	For all $t \ge 1$, the quantile residual $F_{t}$ satisfies:
	\begin{enumerate}
		\item $\E\big[ F_{t}(q) \mid \T_1, \dots, \T_{t-1} \big] = 0$ at every $q \in (0, 1)$;
		\item $\E\big[ \langle F_{t}, F_{s} \rangle_2 \mid \T_1, \dots, \T_{t-1} \big] = 0$ for any time $s$ such that $1 \le s < t$;
		\item $\E\big[ \| F_{t} \|_2^2 \mid \T_1, \dots, \T_{t-1} \big] = \sigma_{t}^2$;
		\item $\E\big[ \| F_{t} \|_2^4 \mid \T_1, \dots, \T_{t-1} \big] \le \kappa$ .
	\end{enumerate}
	Furthermore, these imply: $\E[ F_{t}(q) ] = 0$, $\E[ \langle F_{t}, F_{s} \rangle_2 ] = 0$, $\E[ \| F_{t} \|_2^2 ] = \sigma_{t}^2$, and $\E[ \| F_{t} \|_2^4 ] \le \kappa$.
\end{lemma}

The proof is provided in \Cref{apx:proof_residual_property}.
In the main text, \Cref{lem:quantile_process} demonstrated that the WES process can be formulated as autoregressive dynamics of the associated quantile functions.
Within these dynamics, the quantile residual $F_t$ serves as a noise term, although it is no longer i.i.d.
These moment properties of $F_t$ are essential for the main proofs presented in \Cref{apx:proofs}.

\subsection{Quantile Residual Cumulative and Its Moment Properties} \label{apx:error_decomposition}

Next, we introduce a quantity $G_t^\theta$, which we term the \emph{cumulative quantile residual}.
Recall that the Wasserstein metric between univariate distributions is equivalent to the $L_2$-distance between their quantile functions.
Thus, the Wasserstein error $L_T(\theta)$ in \Cref{sec:theory} can be expressed as
\begin{align}
	L_T(\theta) = \frac{1}{T} \sum_{t=0}^{T-1} \| U_{t}^\theta - V_{t+1} \|_2^2 . \label{eq:W2_loss}
\end{align}
As defined in \Cref{sec:theory}, the difference $U_{t}^\theta(\cdot) - V_{t+1}(\cdot)$ represents the model-quantile residual.
We demonstrate that the model-quantile residual can be expressed in terms of the cumulative quantile residual $G_t^\theta$.
This expression facilitates the subsequent concentration analysis of $L_T(\theta)$.

First, we formally define the quantile residual cumulative $G_t^\theta$ below.

\begin{definition}[Quantile Residual Cumulative] \label{def:residual_cumulative}
	For each $t \ge 1$ and $\theta \in (0, 1)$, let $G_t^\theta$ be a random element in $L^2$ defined by the sequence of quantile residuals $F_1, \dots, F_{t}$, given by
	\begin{align}
		G_t^\theta(q) := \sum_{i=1}^{t} (1 - \theta)^{t-i} F_i(q) ,
	\end{align}
    where we set $G_0^\theta(q) := 0$ for $t = 0$.
	We call $G_t^\theta$ the quantile residual cumulative at time $t$.
\end{definition}

The next lemma shows that, for each $t$, the model-quantile residual $U_{t}^\theta - V_{t+1}$ admits a closed-form expression in terms of the quantile residual cumulative $G_t^\theta$ and the quantile residual $F_t$.

\begin{lemma} \label{prop:model_quantile_decomposition}    
	For any $t \ge 0$ and $\theta \in (0, 1)$, we have for all $q \in (0, 1)$:
	\begin{align}
		U_{t}^\theta(q) - V_{t+1}(q) = (1 - \theta)^{t} \big( U_0^\theta(q) - U_0(q) \big) + ( \theta - \theta_* ) G_{t}^\theta(q) - F_{t+1}(q) .
	\end{align}
\end{lemma}

The proof is provided in \Cref{apx:proof_model_quantile_decomposition}.
Finally, we analyze the moment properties of the quantile residual cumulative $G_t^\theta$.
Observe that the quantile residual cumulative $G_t^\theta$ is a weighted sum of the quantile residuals $F_1, \dots, F_{t}$.
As a preliminary step for this analysis, we first show that the partial sums of the associated weights $\{ (1 - \theta)^{t-i} \}_{i=1}^{t}$ are bounded and convergent.
In the following lemma,  $\nearrow$ denotes monotone convergence from below as $t \to \infty$.

\begin{lemma} \label{lem:w_sum}
	For any $t \ge 1$ and $\theta \in (0, 1)$, we have
    \begin{enumerate}
		\item $\sum_{i=1}^{t} (1 - \theta)^{t-i} = \sum_{i=0}^{t-1} (1 - \theta)^{i} = \frac{ 1 - (1 - \theta)^{t} }{ \theta } \nearrow \frac{1}{\theta}$;
        \item $\sum_{i=1}^{t} (1 - \theta)^{2(t-i)} = \sum_{i=0}^{t-1} (1 - \theta)^{2 i} = \frac{ 1 - (1 - 2 \theta + \theta^2)^{t} }{ \theta ( 2 - \theta ) } \nearrow \frac{1}{\theta (2 - \theta)}$ .
	\end{enumerate}
\end{lemma}

\begin{proof}
	By reindexing the sum, $\sum_{i=1}^{t} (1 - \theta)^{t-i} = \sum_{i=0}^{t-1} (1 - \theta)^{i}$ and $\sum_{i=1}^{t} (1 - \theta)^{2(t-i)} = \sum_{i=0}^{t-1} (1 - \theta)^{2 i}$.
    The results then follow directly from the geometric series formula.
	The limits follow immediately.
\end{proof}

The moment properties of the quantile residual cumulative $G_t^\theta$ are established as below.

\begin{lemma} \label{lem:g_moment}
	At each time $t \ge 1$ and $\theta \in (0, 1)$, the quantile residual cumulative $G_t^\theta$ is a random element defined by the sequence of random transport maps $\T_1, \dots, \T_t$. 
    Define $m_t^\theta := \sum_{i=1}^{t} (1 - \theta)^{2(t-i)} \sigma_i^2$, setting $m_0^\theta := 0$ for $t = 0$.
    For any $0 \le s < t$, there exists a time-independent constant $c^\theta$, such that:
	\begin{enumerate}
        \item $\E[ G_t^\theta(q) ] = 0$;
        \item $\E[ \langle F_t, G_s^\theta \rangle ] = 0$;
		\item $\E[ \| G_t^\theta \|_2^2 ] = m_t^\theta \le \frac{\sigma^2}{\theta (2 - \theta)}$;
		\item $\E[ ( \| G_t^\theta \|_2^2 - m_t^\theta )^2 ] \le c^\theta$;
		\item $\E[ ( \| G_t^\theta \|_2^2 - m_t^\theta ) ( \| G_s^\theta \|_2^2 - m_s^\theta ) ] \le (1 - \theta)^{2(t-s)} c^\theta$ .
	\end{enumerate}
\end{lemma}

The proof is provided in \Cref{apx:proof_g_moment}.
Together with the moment properties of the quantile residual $F_t$, those of the quantile residual cumulative $G_t^\theta$ are essential for the subsequent main proofs.

\section{Proofs of \Cref{apx:preparation}} \label{apx:preparation_proofs}

This section contains proofs of the intermediate lemmas presented in \Cref{apx:preparation}.

\subsection{Proof of \Cref{lem:residual_property}} \label{apx:proof_residual_property}

\begin{proof} 
    Recall that the random predictor $\mu_{t}$ in \Cref{def:wesp}---hence its quantile function $U_{t}$ in \Cref{lem:quantile_process}---is determined by the sequence of random transport maps $\T_1, \dots, \T_{t}$ via its iterative formula.
    Accordingly, the quantile residual $F_{t}$---determined by $U_{t-1}$ and $\T_{t}$ as in \eqref{eq:quantile residual redefined}---is a random element defined by the sequence of random transport maps $\T_1, \dots, \T_{t}$.

    The above argument means that $F_t$ can be viewed as a functional of $\T_1, \dots, \T_{t}$ and its expectation $\E[ F_t(q) ]$ can be taken with respect to $\T_1, \dots, \T_{t}$.
	By assumption, the expectation of $\T_t$ equals the identity map.
    The first item in \Cref{lem:residual_property} then follows from the definition of $F_{t}$, yielding
	\begin{align}
		\E[ F_{t}(q) \mid \T_1, \dots, \T_{t-1} ] = \E[ \T_{t} \circ U_{t-1}(q) - U_{t-1}(q) \mid \T_1, \dots, \T_{t-1}] = U_{t-1}(q) - U_{t-1}(q) = 0 .
	\end{align}
    Here, $U_{t-1}(q)$ acts as a constant under conditioning on $\T_1, \dots, \T_{t-1}$, meaning the above expectation is taken solely with respect to $\T_t$.
    The second item holds by the same argument:
	\begin{align}
		\E[ \langle F_{t}, F_{s} \rangle_2 \mid \T_1, \dots, \T_{t-1} ] & = \big\langle \,\E[ F_{t} \mid \T_1, \dots, \T_{t-1} ] \, , \, F_{s} \, \big\rangle_2 = 0 ,
	\end{align}
    where $F_s$ acts as a constant under conditioning on $\T_1, \dots, \T_{t-1}$ since $s < t$, and the $L^2$ inner product and expectation can be interchanged via the Fubini-Tonelli theorem.
	Similarly, the third item follows from the definition of $F_{t}$ and the assumption on $\T_{t}$:
	\begin{align}
		\E[ \| F_{t} \|_2^2 \mid \T_1, \dots, \T_{t-1} ] & = \E\bigg[ \int_{(0,1)} \left( \T_{t} \circ U_{t-1}(q) - U_{t-1}(q) \right)^2 dq ~\bigg|~ \T_1, \dots, \T_{t-1} \bigg] \\
        & = \E\bigg[ \int_{\R} \left( \T_{t}(x) - x \right)^2 d \mu_{t-1}(x) ~\bigg|~ \T_1, \dots, \T_{t-1} \bigg] = \sigma_{t}^2 ,
	\end{align}
    where we applied the change of variables $x = U_t(q)$ to obtain the second equality.
	Finally, the fourth item is an immediate consequence of Jensen's inequality as follows:
	\begin{align}
		\E\left[ \| F_{t} \|_2^4 \mid \T_1, \dots, \T_{t-1} \right] & \le \E\bigg[ \int_{(0,1)} \left( \T_{t} \circ U_{t-1}(q) - U_{t-1}(q) \right)^4 dq ~\bigg|~ \T_1, \dots, \T_{t-1} \bigg] \\
        & = \E\bigg[ \int_{\R} \left( \T_{t}(x) - x \right)^4 d \mu_{t-1}(x) ~\bigg|~ \T_1, \dots, \T_{t-1} \bigg] \le \kappa
	\end{align}
	where the last inequality holds by assumption.
	
	Properties 1--4 immediately imply $\E[ F_{t}(q) ] = 0$, $\E[ \langle F_{t}, F_{s} \rangle_2 ] = 0$, $\E[ \| F_{t} \|_2^2 ] = \sigma_{t}^2$, and $\E[ \| F_{t} \|_2^4 ] \le \kappa$.
    To see this, applying the law of iterated expectations yields
    \begin{align}
        \E[ F_{t} ] = \E\big[\E[ F_{t} \mid \T_1, \dots, \T_{t-1}]\big] = 0 .
    \end{align}
    The rest can be shown similarly. 
    This completes the proof.
\end{proof}

\subsection{Proof of \Cref{prop:model_quantile_decomposition}} \label{apx:proof_model_quantile_decomposition}

\begin{proof}
	Recall that the WES-filter predictor $\mu_t^\theta$ is given by the pushforward formula
    \begin{align}
		\mu_{t}^{\theta} & = \left( (1 - \theta) \operatorname{Id} + \theta  T_{\mu_{t-1}^{\theta} \to {\nu}_{t}} \right)_{\#} \mu_{t-1}^{\theta} .
	\end{align}
    Equivalently, this formula can be expressed in terms of the corresponding quantile functions as
	\begin{align}
		U_t^\theta & = (1 - \theta) U_{t-1}^\theta + \theta V_{t} . \label{eq:u_t_theta_exp}
	\end{align}
    Consider $U_t^\theta - U_t$, where it follows from \eqref{eq:u_t_theta_exp} that $U_t^\theta - U_t = (1 - \theta) U_{t-1}^\theta + \theta V_{t} - U_{t}$.
    Applying the equalities $V_t = U_{t-1} + F_{t}$ and $U_t = U_{t-1} + \theta_* F_{t}$ from \Cref{lem:quantile_process} yields
	\begin{align}
		U_t^\theta - U_t & = (1 - \theta) U_{t-1}^\theta + \theta ( U_{t-1} + F_{t} ) - ( U_{t-1} + \theta_* F_{t} ) \\
		& = (1 - \theta) ( U_{t-1}^\theta -  U_{t-1} ) + (\theta - \theta_*) F_{t}.
	\end{align}
    This first-order non-homogeneous linear recursion admits the following explicit solution
	\begin{align}
		U_t^\theta - U_t & = (1 - \theta)^t ( U_0^\theta - U_0 ) + ( \theta - \theta_* ) \sum_{i=1}^{t} (1 - \theta)^{t-i} F_i \\ 
        &= (1 - \theta)^t ( U_0^\theta - U_0 ) + ( \theta - \theta_* ) G_t^\theta .
	\end{align}
    Consider $U_{t}^\theta - V_{t+1}$.
    Together with $V_{t+1} = U_{t} + F_{t+1}$, the above equality yields
	\begin{align}
		U_{t}^\theta - V_{t+1} = U_{t}^\theta - U_{t} - F_{t+1} = (1 - \theta)^{t} ( U_0^\theta - U_0 ) + ( \theta - \theta_* ) G_{t}^\theta - F_{t+1} ,
    \end{align}
    which complete the proof.
\end{proof}

\subsection{Proof of \Cref{lem:g_moment}} \label{apx:proof_g_moment}

\begin{proof}
	The following argument holds pointwise for any $\theta \in (0, 1)$.
	Hence, for brevity, we drop the superscript $\theta$ from the notation, writing $G_t^\theta$ as $G_t$ and setting $w_{t,i} := (1 - \theta)^{t-i}$.
	By \Cref{lem:w_sum}, the partial sums $\sum_{i=1}^{t} w_{t,i}$ and $\sum_{i=1}^{t} w_{t,i}^2$ satisfy the following bounds uniformly for all $t \ge 1$:
	\begin{align}
		\sum_{i=1}^{t} w_{t,i} \le \frac{1}{\theta} := c_1 \qquad \text{and} \qquad \sum_{i=1}^{t} w_{t,i}^2 \le \frac{1}{\theta (2 - \theta)} =: c_2 . \label{eq:w_bounds}
	\end{align}
    First, recall that $G_t$ is defined by the quantile residuals $F_1, \dots, F_t$.
    By \Cref{lem:residual_property}, $F_i$ is defined by the random maps $\T_1, \dots, \T_i$ for each $i \ge 1$.
    Consequently, $G_t$ is a random element entirely dependent on the random maps $\T_1, \dots, \T_t$.
	We now proceed to prove items 1–5 sequentially.
    
    \vspace{5pt}
	\noindent
	\textbf{Item 1:}
    By the definition of $G_t$, its expectation is given by
    \begin{align}
		\E[ G_t(q) ] & = \sum_{i=1}^{t} w_{t,i} \E[ F_i(q) ] = 0 ,
	\end{align}
    where the final equality holds immediately by \Cref{lem:residual_property}.
    
    \vspace{5pt}
	\noindent
	\textbf{Item 2:}
    As shown above, $G_s$ is a random element defined by $\T_1, \dots, \T_s$.
    Meanwhile, \Cref{lem:g_moment} states that $F_t$ is a random element defined by $\T_1, \dots, \T_t$.
    Since $s < t$, the law of iterated expectations yields
    \begin{align}
		\E\big[ \langle F_t, G_s \rangle \big] & = \E\Big[ \big\langle \E_{\T_t}\big[ F_t \mid \T_1, \dots, \T_{t-1} \big], G_s \big\rangle \Big] = 0
	\end{align}
    where the last equality follows from item (i) of \Cref{lem:residual_property}.
    
	\vspace{5pt}
	\noindent
	\textbf{Item 3:}
	By the definition of $G_t$, the squared norm $\| G_t \|_2^2$ expands as follows:
	\begin{align}
		\| G_t \|_2^2 & = \bigg\| \sum_{i=1}^{t} w_{t,i} F_i \bigg\|_2^2 = \sum_{i = 1}^{t} \sum_{j = 1}^{t} w_{t,i} w_{t,j} \left\langle F_i, F_j \right\rangle_2 = \sum_{i=1}^{t} w_{t,i}^2 \| F_i \|_2^2 + 2 \sum_{i = 1}^{t} \sum_{j = 1}^{i-1} w_{t,i} w_{t,j} \left\langle F_i, F_j \right\rangle_2 .
	\end{align}
	Taking the expectation of this squared norm yields
	\begin{align}
		\E\left[ \| G_t \|_2^2 \right] & = \sum_{i=1}^{t} w_{t,i}^2 \underbrace{ \E[ \| F_i \|_2^2 ] }_{ = \sigma_i^2 } + 2 \sum_{i = 1}^{t} \sum_{j = 1}^{i-1} w_{t,i} w_{t,j} \underbrace{ \E\left[ \left\langle F_i, F_j \right\rangle_2 \right] }_{ = 0 } = \sum_{i=1}^{t} w_{t,i}^2 \sigma_i^2 = m_t ,
	\end{align}
	where we applied the equalities $\E[ \| F_i \|_2^2 ] = \sigma_i^2$ and $\E\left[ \langle F_i, F_j \rangle_2 \right] = 0$ from \Cref{lem:residual_property}.
	Finally, since $\sigma_i \le \sigma$ by assumption, we obtain $m_t \le \sigma \sum_{i=1}^{t} w_{t,i}^2 \le \sigma c_2$, which establishes item 3.
	
	\vspace{5pt}
	\noindent
	\textbf{Item 4:}
	We now establish the fourth inequality. First, we have
	\begin{align}
		\E\left[ \left( \| G_t \|_2^2 - m_t \right)^2 \right] \le \E\left[ \| G_t \|_2^4 + \big( m_t \big)^2 \right] = \E\bigg[ \bigg\| \sum_{i=1}^{t} w_{t,i} F_i \bigg\|_2^4 \bigg] + \bigg( \sum_{i=1}^{t} w_{t,i}^2 \sigma_i^2 \bigg)^2 .
	\end{align}
    For the first term on the right-hand side, we normalize the weights so that Jensen's inequality can be applied.
    Let $W_t := \sum_{i = 1}^{t} w_{t,i}$ for notational convenience.
    By rearranging the first term, we obtain
	\begin{align}
		\E\bigg[ \bigg\| \sum_{i=1}^{t} w_{t,i} F_i \bigg\|_2^4 \bigg] = W_t^4 \E\bigg[ \bigg\| \sum_{i=1}^{t} \frac{w_{t,i}}{W_t} F_i \bigg\|_2^4 \bigg] \le W_t^4 \sum_{i=1}^{t} \frac{w_{t,i}}{W_t} \E[ \|  F_i \|_2^4 ] \le W_t^4 \kappa , 
	\end{align}
    where we used the inequality $\E[ \|  F_i \|_2^4 ] \le \kappa$ from \Cref{lem:residual_property}.
	Similarly, let $Z_t := \sum_{i = 1}^{t} w_{t,i}^2$ to simplify the notation.
	Rearranging the second term yields
	\begin{align}
		\bigg( \sum_{i=1}^{t} w_{t,i}^2 \sigma_i^2 \bigg)^2 = Z_t^2 \bigg( \sum_{i=1}^{t} \frac{w_{t,i}^2}{Z_t} \sigma_i^2 \bigg)^2 \le Z_t^2 \sum_{i=1}^{t} \frac{w_{t,i}^2}{Z_t} \sigma_i^4 \le Z_t^2 \sigma^4 
	\end{align}
    where we applied the assumption $\sigma_i \le \sigma$.
	Substituting these bounds into \eqref{eq:w_bounds} yields
	\begin{align}
		\E\left[ \left( \| G_t \|_2^2 - m_t \right)^2 \right] \le \bigg( \sum_{i=1}^{t} w_{t,i} \bigg)^4 \kappa + \bigg( \sum_{i=1}^{t} w_{t,i}^2 \bigg)^2 \sigma^4 \le \big( c_1 \big)^4 \kappa + \big( c_2 \big)^2 \sigma^4 .
	\end{align}
	Setting $c := ( c_1 )^4 \kappa + ( c_2 )^2 \sigma^4$ establishes item 4.
	
	\vspace{5pt}
	\noindent
	\textbf{Item 5:}
	For notational convenience, let $H_t := \| G_t \|_2^2 - m_t$, which expands to
	\begin{align}
		H_t = \| G_t \|_2^2 - m_t = \sum_{i=1}^{t} w_{t,i}^2 \left( \| F_i \|_2^2 - \sigma_i^2 \right) + 2 \sum_{i = 1}^{t} \sum_{j = 1}^{i-1} w_{t,i} w_{t,j} \left\langle F_i, F_j \right\rangle_2 .
	\end{align}
	For $0 \le s < t$, decomposing the summation yields $H_t = (*_1) + (*_2)$, where
	\begin{align}
		(*_1) & = \sum_{i=s+1}^{t} w_{t,i}^2 \left( \| F_i \|_2^2 - \sigma_i^2 \right) + 2 \sum_{i = s+1}^{t} \sum_{j = 1}^{i-1} w_{t,i} w_{t,j} \left\langle F_i, F_j \right\rangle_2 , \\
		(*_2) & = \sum_{i=1}^{s} w_{t,i}^2 \left( \| F_i \|_2^2 - \sigma_i^2 \right) + 2 \sum_{i = 1}^{s} \sum_{j = 1}^{i-1} w_{t,i} w_{t,j} \left\langle F_i, F_j \right\rangle_2 .
	\end{align}
	Under this decomposition, we have $\E[ H_t H_s ] = \E[ (*_1) H_s ] + \E[ (*_2) H_s ]$.
	First, we demonstrate that the first term $\E[ (*_1) H_s]$ is zero.
	Substituting the definition of $(*_1)$ into $\E[ (*_1) H_s]$ yields
	\begin{align}
		\E\left[ (*_1) H_s \right] & = \E\left[ \sum_{i=s+1}^{t} w_{t,i}^2 \left( \| F_i \|_2^2 - \sigma_i^2 \right) H_s + 2 \sum_{i = s+1}^{t} \sum_{j = 1}^{i-1} w_{t,i} w_{t,j} \left\langle F_i, F_j \right\rangle_2 H_s \right] \\
		& = \sum_{i=s+1}^{t} w_{t,i}^2 \underbrace{ \E\left[ \left( \| F_i \|_2^2 - \sigma_i^2 \right) H_s \right] }_{ =(*_i') } + 2 \sum_{i = s+1}^{t} \sum_{j = 1}^{i-1} w_{t,i} w_{t,j} \underbrace{ \E\left[ \left\langle F_i, F_j \right\rangle_2 H_s \right] }_{ =(*_i'') } .
	\end{align}
    Observe that $H_s$ is completely determined by $\T_1, \dots, \T_{s}$ via its dependence on $G_s$.
    Furthermore, the summation index $i$ is strictly greater than $s$.
    Thus, for any indices $i > s$ and $j \le s$, we can apply the law of iterated expectations to see
	\begin{align}
		(*_i') & = \E\Big[ \, \E_{\T_i}\big[ (\| F_i \|_2^2 - \sigma_i^2) \times H_s \, \mid \T_1, \dots, \T_{i-1} \big] \, \Big] = \E\Big[ \, \E_{\T_i}\big[ \| F_i \|_2^2 - \sigma_i^2 \mid \T_1, \dots, \T_{i-1} \big] \times H_s \, \Big] = 0 , \\
		(*_i'') & = \E\Big[ \, \E_{\T_i}\big[ \left\langle F_i, F_j \right\rangle_2 \times H_s \, \mid \T_1, \dots, \T_{i-1} \big] \, \Big] = \E\Big[ \, \E_{\T_i}\big[ \left\langle F_i, F_j \right\rangle_2  \mid \T_1, \dots, \T_{i-1} \big] \times H_s \, \Big] = 0 ,
	\end{align}
	where the final equalities follow from \Cref{lem:residual_property}.
	This verifies that $\E[ (*_1) H_s] = 0$.
	To conclude item 5, we evaluate the second term $\E[ (*_2) H_s]$.
	Observe the algebraic identity:
	\begin{align}
		w_{t,i} = (1 - \theta)^{t-i} = (1 - \theta)^{t-s} (1 - \theta)^{s-i} = (1 - \theta)^{t-s} w_{s,i} .
	\end{align}
    Substituting this expression for $w_{t,i}$ into $(*_2)$ yields
	\begin{align}
		(*_2) & = (1 - \theta)^{2(t-s)} \bigg( \sum_{i=1}^{s}{ w_{s,i}^2} \left( \| F_i \|_2^2 - \sigma_i^2 \right) + 2 \sum_{i = 1}^{s} \sum_{j = 1}^{i-1}{w_{s,i} w_{s,j}} \left\langle F_i, F_j \right\rangle_2 \bigg) = (1 - \theta)^{2(t-s)} H_s .
	\end{align}
    This, in turn, implies that $\E[ (*_2) H_s] = (1 - \theta)^{2(t-s)} \E[ H_s^2 ]$.
    Applying the upper bound for $\E[ H_s^2 ]$ established in item 4 yields the desired upper bound for $\E[ H_t H_s]$ in item 5.
\end{proof}

\section{Main Proofs} \label{apx:proofs}

This section contains the proofs of the theoretical results presented in the main text.

\subsection{Proof of \Cref{lem:quantile_process}} \label{apx:proof_quantile_process}

\begin{proof}
    The expression for the quantile function $V_t$ follows immediately from algebraic manipulation
	\begin{align}
		V_{t}(q) = \mathcal{E}_{t} \circ U_{t-1}(q) = U_{t-1}(q) + \mathcal{E}_{t} \circ U_{t-1}(q) - U_{t-1}(q) = U_{t-1}(q) + F_{t}(q) .
	\end{align}
    Recall that the random predictor $\mu_t$ is defined via the pushforward formula in \eqref{eq:WES_process_formula 2}, given by
    \begin{align}
    \mu_t = \left( (1 - \theta_*) \cdot \operatorname{Id} + \theta_* \cdot T_{\mu_{t-1} \to \nu_{t}} \right)_{\#} \mu_{t-1} .
    \end{align}
    In the univariate case, the pushforward map $T_{\mu_{t-1} \to \nu_{t}}$ for optimal transport from a distribution $\mu$ to another distribution $\nu$ is given by $T_{\mu_{t-1} \to \nu_{t}}(\cdot) = V \circ U^{-1}(\cdot)$, where $U$ and $V$ denote their quantile functions.
    Applying this identity, we obtain the following equation for each $x \in \R$:
	\begin{align}
		U_t \circ U_{t-1}^{-1} (x) & = (1 - \theta_*) \cdot \operatorname{Id}(x) + \theta_* \cdot V_{t} \circ U_{t-1}^{-1} (x) .
	\end{align}
    Applying the change of variables $x = U_{t-1}(q)$ for the quantile variable $q \in (0, 1)$ yields
    \begin{align}
		U_t(q) & = (1 - \theta_*) U_{t-1}(q) + \theta_* V_{t}(q) = U_{t-1}(q) + \theta_* F_{t}(q) ,
	\end{align}
    which completes the proof.
\end{proof}

\subsection{Proof of \Cref{prop:stationarity}} \label{apx:proof_stationarity}

\begin{proof}
    First, by solving the recursive formula in \Cref{lem:quantile_process}, we express the quantile function $V_t$ of the random observable $\nu_t$ as the following linear combination:
	\begin{align}
		V_t(q) & = U_0(q) +  F_t(q) + \theta_* \sum_{i=1}^{t-1} F_i(q) .
	\end{align}
	Applying \Cref{lem:residual_property} in \Cref{apx:preparation} then yields
	\begin{align}
		\E[ V_t(q) ] & = U_0(q) + \E[ F_t(q) ] + \theta_* \sum_{i=1}^{t-1} \E[ F_i(q) ] = U_0(q) . 
	\end{align}
    For a random measure $\mu$ with corresponding random quantile function $U$, it is well known that its Fr\'{e}chet mean $E(\mu)$ is a distribution whose quantile function equals $\E[ U(q) ]$ pointwise \citep[see, e.g.,][Theorem 3.2.11]{Panaretos2020}.
    Accordingly, the Fr\'{e}chet mean $E(\nu_t)$ of the random observable $\nu_t$ is a distribution whose quantile function equals $\E[ V_t(q) ] = U_0(q)$ pointwise.
	In turn, this implies that the Fr\'{e}chet mean $E(\nu_t)$ equals the initial fixed distribution $\mu_0$.
	Finally, the existence of the Fr\'{e}chet mean $E(\nu_t)$ implies $V(\nu_t) < \infty$ by definition, which completes the proof.
\end{proof}

\subsection{Proof of \Cref{prop:residual_autocovariance}} \label{apx:proof_residual_autocovariance}

\begin{proof}
Let $\Delta_0 := U_0^\theta - U_0$, noting that $\Delta_0$ is deterministic and independent of $\theta$, since $U_0^\theta$ is the initial condition of the WES filter.
\Cref{prop:model_quantile_decomposition} (see \Cref{apx:preparation}) establishes that
\begin{align}
    U_{t}^\theta - V_{t+1} = (1 - \theta)^t \Delta_0 + ( \theta - \theta_* ) G_t^\theta - F_{t+1} .
\end{align}
Consequently, the autocovariance of interest $\E[ \langle U_{t}^\theta - V_{t+1},U_{s}^\theta - V_{s+1} \rangle ]$ expands to
\begin{align}
    (*) & := \E[ \langle \, (1 - \theta)^t \Delta_0 + ( \theta - \theta_* ) G_t^\theta - F_{t+1} \, , \, (1 - \theta)^s \Delta_0 + ( \theta - \theta_* ) G_s^\theta - F_{s+1} \, \rangle ] \\
    & = \E\left[ \langle (1-\theta)^t \Delta_0, (1-\theta)^{s} \Delta_0 \rangle \right] + \E\left[ \langle (1-\theta)^t \Delta_0, (\theta - \theta_*) G_s^\theta \rangle \right] + \E\left[ \langle (1-\theta)^t \Delta_0, F_{s+1} \rangle \right] \\
    & \hspace{25pt} + \E\left[ \langle (\theta - \theta_*) G_t^\theta, (1-\theta)^{s} \Delta_0 \rangle \right] + \E\left[ \langle (\theta - \theta_*) G_t^\theta, (\theta - \theta_*) G_s^\theta \rangle \right] + \E\left[ \langle (\theta - \theta_*) G_t^\theta, F_{s+1} \rangle \right] \\
    & \hspace{50pt} + \E\left[ \langle F_{t+1}, (1-\theta)^{s} \Delta_0 \rangle \right] + \E\left[ \langle F_{t+1}, (\theta - \theta_*) G_s^\theta \rangle \right] + \E\left[ \langle F_{t+1}, F_{s+1} \rangle \right] \\
    & = (1-\theta)^{t+s}\|\Delta_0\|^2 + (1-\theta)^t (\theta - \theta_*) \left\langle \Delta_0, \E\big[ G_s^\theta \big] \right\rangle + (1-\theta)^t \left\langle \Delta_0, \E\big[ F_{s+1} \big] \right\rangle \\
    & \hspace{25pt} + (1-\theta)^s (\theta - \theta_*) \left\langle \E\big[ G_t^\theta \big], \Delta_0 \right\rangle + (\theta - \theta_*)^2 \E\left[ \langle G_t^\theta, G_s^\theta \rangle \right] + (\theta - \theta_*) \E\left[ \langle G_t^\theta, F_{s+1} \rangle \right] \\
    & \hspace{50pt} + (1-\theta)^{s} \left\langle \E\big[ F_{t+1} \big], \Delta_0 \right \rangle + (\theta - \theta_*) \E\left[ \langle F_{t+1}, G_s^\theta \rangle \right] + \E\left[ \langle F_{t+1}, F_{s+1} \rangle \right] .
\end{align}
By \Cref{lem:residual_property} and \Cref{lem:g_moment}, we have $\E[ G_s^\theta ] = 0$, $\E[ F_{s+1} ] = 0$, $\E[ G_t^\theta ] = 0$, $\E[ F_{t+1} ] = 0$, $\E[ \langle F_{t+1}, G_s^\theta \rangle ] = 0$, and $\E[ \langle F_{t+1}, F_{s+1} \rangle ] = 0$.
Substituting these identities simplifies $(*)$ to
\begin{align}
    (*) & = (1-\theta)^{t+s}\|\Delta_0\|^2 + (\theta - \theta_*)^2 \underbrace{ \E\left[ \langle G_t^\theta, G_s^\theta \rangle \right] }_{ (*_1) } + (\theta - \theta_*) \underbrace{ \E\left[ \langle G_t^\theta, F_{s+1} \rangle \right] }_{ (*_2) } . \label{eq:autocovariance_decomp}
\end{align}
We evaluate the terms $(*_1)$ and $(*_2)$ in the remainder.

\vspace{5pt}
\noindent
\textbf{Term $(*_1)$:}
Recall the definition of $G_t^\theta$ in \Cref{def:residual_cumulative}.
The following decomposition holds:
\begin{align}
    G_t^\theta = \sum_{i=1}^{s} (1-\theta)^{t-i} F_i + \sum_{i=s+1}^{t} (1-\theta)^{t-i}F_i = (1-\theta)^{t-s} \underbrace{ \sum_{i=1}^s (1-\theta)^{s-i} F_i }_{ = G_s^\theta } + \sum_{i=s+1}^t (1-\theta)^{t-i} F_i
\end{align}
It then follows from this decomposition and \Cref{lem:g_moment} that
\begin{align}
    (*_1) = \E\left[ \langle G_t^\theta, G_s^\theta \rangle \right] = (1 - \theta)^{t-s} \E\big[ \| G_s^\theta \|^2 \big] + \sum_{i=s+1}^t (1 - \theta)^{t-i} \E\big[ \langle G_s^\theta, F_i \rangle \big] = (1-\theta)^{t-s} m_s^\theta ,
\end{align}
where $m_s^\theta = \sum_{i=1}^{t} (1 - \theta)^{2(t-i)} \sigma_i^2$ as defined in \Cref{lem:g_moment}.

\vspace{5pt}
\noindent
\textbf{Term $(*_2)$:}
Returning to the definition of $G_t^\theta$, we obtain
\begin{align}
    (*_2) = \E\left[ \langle G_t^\theta, F_{s+1} \rangle \right] = \sum_{i=1}^{t} (1-\theta)^{t-i} \E\left[ \langle F_i, F_{s+1} \rangle \right] .
\end{align}
\Cref{lem:residual_property} guarantees that $\E\left[ \langle F_t, F_{s+1} \rangle \right] = 0$ for all $i \ne s + 1$, meaning that only $i = s + 1$ remains:
\begin{align}
    (*_2) = (1 - \theta)^{t-s-1} \E\left[ \| F_{s+1} \|_2^2 \right] = (1 - \theta)^{t-s-1} \sigma_{s+1}^2 .
\end{align}
where the second equality follows from $\E[ \| F_{s+1} \|_2^2 ] = \sigma_{s+1}^2$ by \Cref{lem:residual_property}.

\vspace{5pt}
\noindent
Substituting the equalities of $(*_1)$ and $(*_2)$ into \eqref{eq:autocovariance_decomp} completes the proof.
\end{proof}

\subsection{Proof of \Cref{thm:pointwise_convergence}} \label{apx:proof_pointwise_convergence}

\begin{proof}
    Let $\Delta_0 := U_0^\theta - U_0$, noting that $\Delta_0$ is deterministic and independent of $\theta$, since $U_0^\theta$ is the initial condition of the WES filter.
    Recall from \eqref{eq:W2_loss} that $L_T(\theta) = (1 / T) \sum_{t=0}^{T-1} \| U_{t}^\theta - V_{t+1} \|_2^2$.
    Applying \Cref{prop:model_quantile_decomposition} (see \Cref{apx:preparation}) and expanding the square yield
    \begin{align}
        \| U_{t}^\theta - V_{t+1} \|_2^2 & = \| (1 - \theta)^{t} \Delta_0 + ( \theta - \theta_* ) G_{t}^\theta - F_{t+1} \|_2^2 \\
        & = ( \theta - \theta_* )^2 \| G_{t}^\theta\|_2^2 + 2 ( \theta - \theta_* ) \big\langle G_{t}^\theta, (1 - \theta)^{t} \Delta_0 - F_{t+1} \big\rangle_2 + \| (1 - \theta)^{t} \Delta_0 - F_{t+1} \|_2^2 .
    \end{align}
    Substituting this into $L_T$ yields the following decomposition
	\begin{multline}
		L_T(\theta) = \bigg( \underbrace{ \frac{1}{T} \sum_{t=0}^{T-1} \| G_{t}^\theta\|_2^2 }_{ =: A_T(\theta) } \bigg) \times (\theta - \theta_*)^2  + 2 \bigg( \underbrace{ \frac{1}{T} \sum_{t=0}^{T-1} \big\langle G_{t}^\theta, (1 - \theta)^{t} \Delta_0 - F_{t+1} \big\rangle_2 }_{ =: B_T(\theta) } \bigg) \times (\theta - \theta_*) \\
        + \underbrace{ \frac{1}{T} \sum_{t=0}^{T-1} \| (1 - \theta)^{t} \Delta_0 - F_{t+1} \|_2^2 }_{ =: C_T(\theta) } .
	\end{multline}
    In the remainder of \Cref{apx:proof_pointwise_convergence}, we show that there exist a strictly-positive deterministic function $A_\infty(\theta)$ and a non-negative constant $C_\infty$ such that the following convergences hold pointwise:
	\begin{align}
		\text{(i)   } A_T(\theta) \overset{p}{\longrightarrow} A_\infty(\theta); \qquad \text{(ii)   } B_T(\theta) \overset{p}{\longrightarrow} 0; \qquad \text{(iii)   } C_T(\theta) \overset{p}{\longrightarrow} C_\infty .
	\end{align}
    We establish each convergence in the subsequent subsections.
    Specifically, 
    \begin{itemize}
        \item \Cref{apx:proof_A_T_convergence} shows item (i);
        \item \Cref{apx:proof_B_T_convergence} shows item (ii);
        \item \Cref{apx:proof_C_T_convergence} shows item (iii);
    \end{itemize}
    Given these intermediate results, the pointwise convergence of $L_T$ follows as
    \begin{align}
        L_T(\theta) \overset{p}{\longrightarrow} L_\infty(\theta) = A_\infty(\theta) \times (\theta - \theta_*)^2  + C_\infty .
    \end{align}
    This completes the main proof.
\end{proof}

\subsubsection{Proof of Item (i): $A_T(\theta) \overset{p}{\longrightarrow} A_\infty(\theta)$} \label{apx:proof_A_T_convergence}

\begin{proof}
    The following argument holds pointwise for any $\theta \in (0, 1)$.
    Thus, for brevity, we drop the superscript $\theta$ in the notation whenever possible.
    Specifically, we denote $G_t^\theta$, $m_t^\theta$, and $c^\theta$ from \Cref{lem:g_moment} by $G_t$, $m_t$, and $c$, respectively.
    Similarly, we denote $A_T(\theta)$ by $A_T$.
    
    First, recall from \Cref{lem:g_moment} that $\E[ \| G_t \|_2^2 ] = m_t$.
    Define $H_t := \| G_t \|_2^2 - m_t$ for each time $t \ge 0$.
    Our aim is to show the concentration of a centered quantity $(*_T)$ defined as
	\begin{align}
		(*_T) := A_T - \E[ A_T ] & = \frac{1}{T} \sum_{t = 0}^{T-1} \| G_t \|_2^2 - \frac{1}{T} \sum_{t = 0}^{T-1} m_t = \frac{1}{T} \sum_{t = 0}^{T-1} H_t .
	\end{align}
    By Chebyshev's inequality, we have
	\begin{align}\label{eq:Chebyshev's inequality}
		\mathbb{P}\left( \left| (*_T) \right| > \epsilon \right) \le \frac{ \E\left[ (*_T)^2 \right] }{ \epsilon^2 } .
	\end{align}
    It suffices to show that the expectation term $\E[ (*_T)^2 ]$ converges to $0$ in the limit of $T \to \infty$, since this implies that $(*_T)$ converges to $0$ in probability as above.
	By expanding the square, we have
	\begin{align}
		\E\left[ (*_T)^2 \right] & = \E\left[ \frac{1}{T^2} \sum_{t = 0}^{T-1} H_t^2 + \frac{2}{T^2} \sum_{t=1}^{T-1} \sum_{s=0}^{t-1} H_t H_s \right] = \frac{1}{T} \bigg( \frac{1}{T} \sum_{t=0}^{T-1} \underbrace{ \vphantom{\sum_{s=1}^{t-1}} \E\left[ H_t^2 \right] }_{ =: (*_1) } + \frac{2}{T} \sum_{t=1}^{T-1} \underbrace{ \sum_{s=0}^{t-1} \E\left[ H_t H_s \right] }_{ =: (*_2) } \bigg) . 
	\end{align}
    By \Cref{lem:g_moment}, the first term $(*_1)$ is bounded as $(*_1) = \E[ H_t^2 ] \le c$ .
    \Cref{lem:g_moment} further implies that the second term $(*_2)$ satisfies $(*_2) = \sum_{s=1}^{t-1} \E\left[ H_t H_s \right] \le \sum_{s=0}^{t-1} (1 - \theta)^{2(t-s)} c$.
    For this upper bound of the second term $(*_2)$, reindexing the sum first and applying \Cref{lem:w_sum} next yield
	\begin{align}
		(*_2) \le \sum_{s=0}^{t-1} (1 - \theta)^{2(t-s)} c = \sum_{s=1}^{t} (1 - \theta)^{2(t-s)} c \le \frac{c}{\theta (2 - \theta)} .
	\end{align}
	Substituting these two upper bounds into $\E[ (*_T)^2 ]$ yields
    \begin{align}
		\E\left[ (*_T)^2 \right] & \le \frac{1}{T} \left( c + \frac{2 c}{\theta (2 - \theta)} \right) \longrightarrow 0 \qquad \text{as} \qquad T \longrightarrow \infty .
	\end{align}
    Therefore, Chebyshev's inequality \eqref{eq:Chebyshev's inequality} implies that $(*_T) \to 0$ in probability.
    
	Finally, by the definition of $(*_T)$, this convergence $(*_T) \to 0$ is equivalent to
	\begin{align}
		A_T = \frac{1}{T} \sum_{t = 0}^{T-1} \| G_{t} \|_2^2 \quad \overset{p}{\longrightarrow} \quad \lim_{T \to \infty} \frac{1}{T} \sum_{t = 0}^{T-1} m_{t} = A_\infty .
	\end{align}
	By \Cref{lem:g_moment}, $m_t$ is bounded uniformly for all $t \ge 1$.
	Since $\{ m_t \}_{t=1}^{\infty}$ is a uniformly bounded sequence, its average limit $A_\infty$ exists.
	Furthermore, since $\{ m_t \}_{t=0}^{\infty}$ is an increasing sequence of positive terms, its average limit $A_\infty$ is positive.
    Making $\theta$ explicit completes the proof.
\end{proof}

\subsubsection{Proof of Item (ii): $B_T(\theta) \overset{p}{\longrightarrow} 0$} \label{apx:proof_B_T_convergence}

\begin{proof}
    As in \Cref{apx:proof_A_T_convergence}, the following argument holds pointwise for each $\theta \in (0, 1)$.
	Hence, for brevity, we drop the superscript $\theta$ in the notation.
    Specifically, we denote $G_t^\theta$, $m_t^\theta$, and $c^\theta$ from \Cref{lem:g_moment} by $G_t$, $m_t$, and $c$, respectively.
    Similarly, we denote $B_T^\theta$ by $B_T$.
    
    First, we decompose $B_T$ into two terms as follows:
	\begin{align}
		B_T = \underbrace{ \frac{1}{T} \sum_{t=0}^{T-1} (1-\theta)^{t} \big\langle G_t, U_0^\theta - U_0 \big\rangle_2 }_{ =: B_T^{(1)} } - \underbrace{ \frac{1}{T} \sum_{t=0}^{T-1} \big\langle G_t, F_{t+1} \big\rangle_2 }_{ =: B_T^{(2)} } .
	\end{align}
    We show that both $B_T^{(1)}$ and $B_T^{(2)}$ converge to 0 in probability.
    By Chebyshev's inequality, we have
    \begin{align}
		\mathbb{P}\big( \big| B_T^{(k)} \big| \ge \epsilon \big) \le \frac{ \E\big[ \big( B_T^{(k)} \big)^2 \big] }{ \epsilon^2 }
	\end{align}
    for each $k = 1, 2$.
    Thus, it suffices to show that $\E[ ( B_T^{(k)} )^2 ] \to 0$ as $T \to \infty$ for each $k = 1, 2$.

    \vspace{5pt}
    \noindent
    \textbf{The First Term:}    
    Applying Jensen's inequality for the first term $B_T^{(1)}$ yields
    \begin{align}
		\E\big[ \big( B_T^{(1)} \big)^2 \big] & \le \E\bigg[ \frac{1}{T} \sum_{t=0}^{T-1} (1-\theta)^{2t} \langle G_t, U_0^\theta - U_0 \rangle_2^2 \bigg] = \frac{1}{T} \sum_{t=0}^{T-1} (1-\theta)^{2t} \E\big[ \langle G_t, U_0^\theta - U_0 \rangle_2^2 \big] .
    \end{align}
    Note that $U_0^\theta$ and $U_0$ are fixed initial conditions.
    The Cauchy-Schwarz inequality and \Cref{lem:g_moment} imply
    \begin{align}
        \E\big[ \big( B_T^{(1)} \big)^2 \big] & \le \frac{1}{T} \sum_{t=0}^{T-1} (1-\theta)^{2t} \, \E\big[ \|G_t\|^2 \big] \, \| U_0^\theta - U_0 \|^2 = \frac{1}{T} \sum_{t=0}^{T-1} (1-\theta)^{2t} m_t^\theta \| U_0^\theta - U_0 \|^2 .
    \end{align}
    Applying the bound $m_t^\theta \le \sigma^2 \theta^{-1} (2 - \theta)^{-1}$ in \Cref{lem:g_moment} and subsequently \Cref{lem:w_sum} yields
    \begin{align}
        \E\big[ \big( B_T^{(1)} \big)^2 \big] & = \frac{1}{T} \frac{\| U_0^\theta - U_0 \|^2}{\theta (2 - \theta)} \frac{\sigma^2}{\theta(2-\theta)} \longrightarrow 0 \qquad \text{as} \qquad T \to \infty .
	\end{align}
    By Chebyshev's inequality, we conclude that $B_T^{(1)} \to 0$ in probability.

    \vspace{5pt}
    \noindent
    \textbf{The Second Term:}
    Expanding the square for the second term $B_T^{(2)}$ yields
    \begin{align}
		\E\big[ \big( B_T^{(2)} \big)^2 \big] & = \E\left[ \frac{1}{T^2} \bigg( \sum_{t=0}^{T-1} \left\langle G_t, F_{t+1} \right\rangle_2^2 + 2 \sum_{t=1}^{T-1} \sum_{s=0}^{t-1} \left\langle G_t, F_{t+1} \right\rangle_2 \left\langle G_s, F_{s+1} \right\rangle_2 \bigg) \right] \\
		& = \frac{1}{T} \bigg( \frac{1}{T} \sum_{t=0}^{T-1} \underbrace{ \E\left[ \left\langle G_t, F_{t+1} \right\rangle_2^2 \right] \vphantom{\sum_{j=0}^{t-1}} }_{ = (*_t) } + \frac{2}{T} \sum_{t=1}^{T-1} \sum_{s=0}^{t-1} \underbrace{ \E\bigg[ \left\langle G_t, F_{t+1} \right\rangle_2 \left\langle G_s, F_{s+1} \right\rangle_2 \bigg] \vphantom{\sum_{j=0}^{t-1}} }_{ = (*_{t,s}) }  \bigg) . \label{eq:B_T_bound}
	\end{align}
    We now show that the terms $(*_t)$ is uniformly bounded in $t$ and that the term $(*_{t,s})$ is zero for any time $t > s$.
	By \Cref{lem:g_moment}, the random element $G_t$ is defined by $\T_1, \dots, \T_{t}$.
    Meanwhile, by \Cref{lem:residual_property}, the random element $F_{t+1}$ is defined by $\T_1, \dots, \T_{t+1}$.
	Consequently, applying the Cauchy-Schwarz inequality and the law of iterated expectations bound the first term $(*_t)$ as follows:
	\begin{align}       
        (*_t) & \le \E\left[ \| G_t \|_2^2 \| F_{t+1} \|_2^2 \right] = \E\Big[ \| G_t \|_2^2 \times \E_{\T_{t+1}}\left[ \| F_{t+1} \|_2^2 \mid \T_1, \dots, \T_{t} \right] \Big] = \E\left[ \| G_t \|_2^2 \right] \times \sigma_{t+1}^2 ,
	\end{align}
	where the last equality follows from \Cref{lem:residual_property}.
	Then, applying $\E[ \| G_t \|_2^2 ] = m_t \le \sigma^2 \theta^{-1} (2 - \theta)^{-1}$ from \Cref{lem:g_moment} and the assumption $\sigma_{t+1} \le \sigma$ yields
	\begin{align}
		(*_t) & \le m_t \times \sigma_t^2 \le \frac{\sigma^4}{\theta (2 - \theta)} .
	\end{align}
	Similarly, for the second term $(*_{t,s})$, the law of iterated expectations implies
	\begin{align}
		(*_{t,s}) & = \E\Big[ \E\big[ \left\langle G_t, F_{t+1} \right\rangle_2 \left\langle G_s, F_{s+1} \right\rangle_2 \mid \T_1, \dots, \T_{t} \big] \Big] = \E\Big[ \big\langle G_t , \, \E_{\T_{t+1}}\left[ F_{t+1} \mid \T_1, \dots, \T_{t} \right] \, \big\rangle_2 \big\langle G_s, F_{s+1} \big\rangle_2 \Big] ,
	\end{align}
	where the last equality holds because $s < t$.
    Thus, we have $(*_{t,s}) = 0$ because $\E[ F_{t+1} \mid \T_1, \dots, \T_{t} ] = 0$ according to \Cref{lem:residual_property}.
	Substituting the bound for $(*_t)$ and the fact $(*_{t,s}) = 0$ into \eqref{eq:B_T_bound} gives
	\begin{align}
		\E\big[ \big( B_T^{(2)} \big)^2 \big] & \le \frac{1}{T} \left( \frac{\sigma^4}{\theta (2 - \theta)} + 0 \right) \longrightarrow 0 \qquad \text{as} \qquad T \longrightarrow \infty ,
	\end{align}
	Therefore, Chebyshev's inequality guarantees that $B_T^{(2)} \to 0$ in probability.
\end{proof}

\subsubsection{Proof of Item (iii): $C_T(\theta) \overset{p}{\longrightarrow} C_\infty$} \label{apx:proof_C_T_convergence}

\begin{proof}
    As in \Cref{apx:proof_A_T_convergence} and \Cref{apx:proof_B_T_convergence}, since the following argument holds pointwise, we denote $C_T(\theta)$ by $C_T$ for brevity.
    First, we decompose $C_T$ into three terms:
    \begin{align}
      C_T & := \underbrace{ \frac{1}{T} \sum_{t=0}^{T-1} (1-\theta)^{2t} \| U_0^\theta - U_0 \|^2_2 }_{ =: C_T^{(1)} } - 2 \underbrace{ \frac{1}{T} \sum_{t=0}^{T-1}(1-\theta)^t \big\langle U_{0}^\theta - U_0,  F_{t+1}\big\rangle_2 }_{ =: C_T^{(2)} } + \underbrace{ \frac{1}{T} \sum_{t=0}^{T-1} \|F_{t+1}\|^2_2 }_{ =: C_T^{(3)} }
    \end{align}
    We show convergence of each term in order.

    \vspace{5pt}
    \noindent
    \textbf{The First Term:}
    Since $U_0^\theta$ and $U_0$ are fixed initial conditions, the first term $C_T^{(1)}$ is deterministic.
    Rearranging the summation index and applying \Cref{lem:w_sum} yields
    \begin{align}
		C_T^{(1)} = \frac{1}{T} \| U_0^\theta - U_0 \|^2_2 \sum_{t=1}^{T} (1-\theta)^{2(T-t)} \le \frac{1}{T} \frac{\| U_0^\theta - U_0 \|^2_2}{\theta (2 - \theta)} ,
	\end{align}
    which implies that $C_T^{(1)} \to 0$ in the limit $T \to \infty$.

    \vspace{5pt}
    \noindent
    \textbf{The Second Term:}    
    We show that $C_T^{(2)} \to 0$ in probability.
    By Chebyshev’s inequality, we have
    \begin{align}
        \mathbb{P}\big( \big| C_T^{(2)} \big| \ge \epsilon \big) \le \frac{1}{\epsilon^2} \E\big[ \big( C_T^{(2)} \big)^2 \big] .
	\end{align}
    First, applying Jensen's inequality yields
    \begin{align}
		\E\big[ \big( C_T^{(2)} \big)^2 \big] & \leq \frac{1}{T} \E\bigg[ \sum_{t=0}^{T-1} (1-\theta)^{2t} \langle U_0^\theta - U_0, F_{t+1} \big\rangle_2^2 \bigg] = \frac{1}{T} \sum_{t=0}^{T-1} (1-\theta)^{2t} \E\big[ \langle U_0^\theta - U_0, F_{t+1} \big\rangle_2^2 \big] .
	\end{align}
    Note that $U_0^\theta$ and $U_0$ are deterministic.
    Applying the Cauchy-Schwarz inequality and \Cref{lem:residual_property} yields
    \begin{align}
		\E\big[ \big( C_T^{(2)} \big)^2 \big] & \le \frac{1}{T} \sum_{t=0}^{T-1} (1-\theta)^{2t} \, \| U_0^\theta - U_0 \|_2^2 \, \E\big[ \| F_{t+1} \|_2^2 \big] = \frac{1}{T} \| U_0^\theta - U_0 \|_2^2 \sum_{t=0}^{T-1} (1-\theta)^{2t} \sigma_{t+1}^2 .
	\end{align}
    We apply the assumption $\sigma_{t+1} \le \sigma$ and then \Cref{lem:w_sum} to obtain
    \begin{align}
		\E\big[ \big( C_T^{(2)} \big)^2 \big] & \le \frac{1}{T} \frac{\| U_0^\theta - U_0 \|_2^2 \, \sigma^2}{\theta (2 - \theta)} \longrightarrow 0 \qquad \text{as} \qquad T \longrightarrow \infty .
	\end{align}
    Consequently, Chebyshev's inequality guarantees that $C_T^{(2)} \to 0$ in probability. 
    
    \vspace{5pt}
    \noindent
    \textbf{The Third Term:} 
    For notational convenience, we shift the summation index to obtain
    \begin{align}
		C_T^{(3)} = \frac{1}{T} \sum_{t=0}^{T-1} \|F_{t+1}\|^2_2 = \frac{1}{T} \sum_{t=1}^{T} \|F_{t}\|^2_2 .
	\end{align}
    Define $H_{t} := \| F_{t} \|_2^2 - \sigma_{t}^2$ for each $t \ge 1$.
    We show the concentration of the centered quantity
	\begin{align}
		(*_T) := \frac{1}{T} \sum_{t=1}^{T} \| F_{t} \|_2^2 - \E\left[ \frac{1}{T} \sum_{t=1}^{T} \| F_{t} \|_2^2  \right] = \frac{1}{T} \sum_{t=1}^{T} \| F_{t} \|_2^2 - \frac{1}{T} \sum_{t=1}^{T} \sigma_{t}^2 = \frac{1}{T} \sum_{t=1}^{T} H_{t} ,
	\end{align}
    where the second equality follows from \Cref{lem:residual_property}.
	By Chebyshev's inequality, we have
	\begin{align}
		\mathbb{P}\left( \left| (*_T) \right| > \epsilon \right) \le \frac{1}{\epsilon^2} \E\left[ (*_T)^2 \right] .
	\end{align}
	By expanding the square, we further have
	\begin{align}
		\E\left[ (*_T)^2 \right] & = \E\left[ \frac{1}{T^2} \sum_{t=1}^{T} H_{t}^2 + \frac{2}{T^2} \sum_{t=2}^{T} \sum_{s=1}^{t-1} H_{t} H_{s} \right] = \frac{1}{T} \bigg( \frac{1}{T} \sum_{t=1}^{T} \underbrace{ \E\left[ H_{t}^2 \right] \vphantom{\sum_{t=1}^{T} } }_{ = (*_t^a) } + \frac{2}{T} \sum_{t=2}^{T} \sum_{s=1}^{t-1} \underbrace{ \E\left[ H_{t} H_{s} \right] \vphantom{\sum_{t=1}^{T} } }_{ = (*_t^b) } \bigg) .
	\end{align}
	For the first term $(*_t^a)$, applying \Cref{lem:residual_property} and the assumption $\sigma_t^2 \le \sigma^2$ yields
	\begin{align}
		(*_t^a) \le \E\left[ \| F_{t+1} \|_2^4 \right] + \sigma_{t+1}^4 \le \kappa + \sigma^4 .
	\end{align}
	For the second term $(*_t^b)$, the law of iterated expectations and \Cref{lem:residual_property} imply
	\begin{align}
		(*_t^b) & = \E\Big[ H_{s} \times \E_{\T_{t}}\big[ H_{t} | \T_1, \dots, \T_{t-1} \big] \Big] \\
		& = \E\Big[ H_{s} \times \Big( \underbrace{ \E_{\T_{t}}\big[ \| F_{t} \|_2^2 \mid \T_1, \dots, \T_{t-1} \big] }_{ = \sigma_{t}^2} - \sigma_{t}^2 \Big) \Big] = 0 .
	\end{align}
	Substituting the bound of $(*_t^a)$ and the fact $(*_t^b) = 0$ into $\E[ (*_T)^2 ]$ yields
	\begin{align}
		\E\left[ (*_T)^2 \right] \le \frac{ 1 }{ T } \left( \kappa + \sigma^4 \right) \longrightarrow 0 \qquad \text{as} \qquad T \to \infty .
	\end{align}
    By Chebyshev's inequality, we conclude that $(*_T) \to 0$ in probability.
    In turn, by the definition of $(*_T)$, this convergence in probability is equivalent to
	\begin{align}
		C_T^{(3)} = \sum_{t=1}^{T} \| F_t \|_2^2 \quad \overset{p}{\longrightarrow} \quad C_\infty := \lim_{T \to \infty} \frac{1}{T} \sum_{t=1}^{T} \sigma_t^2 .
	\end{align}
	The sequence $\{ \sigma_t^2 \}_{t=1}^{\infty}$ is uniformly bounded due to the assumption $\sigma_t^2 \le \sigma^2$, and is trivially non-negative.
	Therefore, the average limit $C_\infty$ exists and is non-negative, which completes the proof.
\end{proof}

\subsection{Proof of \Cref{thm:uniform_convergence}} \label{apx:proof_uniform_convergence}

\begin{proof}
    Let $\Omega$ denote an arbitrary compact set in $(0, 1)$, with its minimum and maximum values denoted $\Omega_{\text{min}} > 0$ and $\Omega_{\text{max}} < 1$, respectively.
    By \Cref{thm:pointwise_convergence}, we have $L_T(\theta) \to L_\infty(\theta) := A_\infty(\theta) (\theta - \theta_*)^2 + C_\infty$ in probability pointwise for each $\theta \in \Omega$.
    Our aim is to apply Lemma 2.9 of \cite{Newey1994}, which states that the uniform convergence $L_T \to L_\infty$ in probability holds on $\Omega$ if
    \begin{enumerate}
        \item the function $L_T$ is stochastically equicontinuous;
        \item the limiting function $L_\infty$ is continuous on $\Omega$.
    \end{enumerate}
    The first condition is satisfied if $| L_T(\theta) - L_T(\theta') | \le C_T | \theta - \theta' |$ for some $C_T$ that is \emph{bounded in probability}.
    Here, the non-negative real-valued random variable $C_T$ is said to be bounded in probability if, for each $\epsilon > 0$, there exists a constant $M$ such that $\mathbb{P}( C_T \le M ) \ge 1 - \epsilon$ for all $T$ sufficiently large.
    By Markov's inequality, it suffices to show that $\E[ C_T ]$ is bounded uniformly for all $T$ sufficiently large.
    
    For the first condition, recall the quantile representation $L_T(\theta) = (1 / T) \sum_{i=0}^{T-1} \| U_t^{\theta} - V_{t+1} \|_2^2$ from \Cref{apx:preparation}.
    To establish stochastic equicontinuity, consider the absolute difference 
    \begin{align}
        \big| L_T(\theta) - L_T(\theta') \big| = \bigg| \frac{1}{T} \sum_{i=0}^{T-1} \| U_t^{\theta} - V_{t+1} \|_2^2 - \| U_t^{\theta'} - V_{t+1} \|_2^2 \bigg| 
    \end{align}
    where $\theta, \theta'$ are two arbitrary points in $\Omega$.
    Recall that $| \| u \|_2^2 - \| v \|_2^2 | \leq ( \|u\|_2 + \|v\|_2 ) \| u - v \|_2$ holds for any $u, v \in L^2$.
    Applying this inequality yields the following upper bound:
    \begin{align}
        \big| L_T(\theta) - L_T(\theta') \big| & \le \frac{1}{T} \sum_{i=0}^{T-1} \left( \| U_t^{\theta} - V_{t+1} \|_2 + \| U_t^{\theta'} - V_{t+1} \|_2 \right) \| U_t^{\theta}  - V_{t+1} - ( U_t^{\theta'} - V_{t+1} ) \|_2 \\
        & \le \frac{1}{T} \sum_{i=0}^{T-1} \left( 2 \sup_{\theta \in \Omega}  \| U_t^{\theta} - V_{t+1} \|_2 \right) \| U_t^{\theta}  - V_{t+1} - ( U_t^{\theta'} - V_{t+1} ) \|_2 \\
        & \le 2 \underbrace{ \bigg( \frac{1}{T} \sum_{i=0}^{T-1} \sup_{\theta \in \Omega} \| U_t^{\theta} - V_{t+1} \|_2^2 \bigg)^{\frac{1}{2}} }_{ = (*_1) } \underbrace{ \bigg( \frac{1}{T} \sum_{i=0}^{T-1} \| U_t^{\theta}  - V_{t+1} - ( U_t^{\theta'} - V_{t+1} ) \|_2^2 \bigg)^{\frac{1}{2}} }_{ = (*_2) } 
    \end{align}
    where the last inequality follows from the Cauchy-Schwarz inequality.
    Thus, the function $L_T$ is stochastically equicontinuous, if (i) the first term $(*_1)$ is bounded in probability and (ii) the second term $(*_2)$ satisfies the bound $(*_2) \le C_T | \theta - \theta' |$ for some non-negative $C_T$ bounded in probability.
    
    To structure the proof, we establish conditions (i) and (ii) in the subsequent subsections.
    We then verify the continuity of the limiting function $L_\infty(\theta)$ in the final subsection.
    Specifically,
    \begin{itemize}
        \item \Cref{apx:stochastic_eqcon_1} shows condition (i);
        \item \Cref{apx:stochastic_eqcon_2} shows condition (ii);
        \item \Cref{apx:limiting_con} shows condition 2 (continuity of $L_\infty(\theta)$).
    \end{itemize}
    Given these results, the main claim follows from Lemma 2.9 of \cite{Newey1994}.
\end{proof}

\subsubsection{Proof of Condition (i): The Term $(*_1)$ is Bounded in Probability} \label{apx:stochastic_eqcon_1}

Let $\Delta_0 := U_0^\theta - U_0$, noting that $\Delta_0$ is deterministic and independent of $\theta$, since $U_0^\theta$ is the initial condition of the WES filter.
It follows from \Cref{prop:model_quantile_decomposition} in \Cref{apx:preparation} that
\begin{align}
    \| U_t^{\theta} - V_{t+1} \|_2 & = \| (1 - \theta)^t \Delta_0 + (\theta - \theta_*) G_t^\theta - F_{t+1} \big\|_2 \\
    & \le (1 - \theta)^t \big\| \Delta_0 \big\|_2 + \big| \theta - \theta_* \big| \big\| G_t^\theta \big\|_2 + \big\| F_{t+1} \big\|_2 \\
    & \le \big\| \Delta_0 \big\|_2 + \big\| G_t^\theta \big\|_2 + \big\| F_{t+1} \big\|_2 ,
\end{align}
where we use the bounds $(1 - \theta)^t \le 1$ and $| \theta - \theta_* | \le 1$ since $\Omega \subset (0, 1)$.
This leads to
\begin{align}
    \| U_t^{\theta} - V_{t+1} \|_2^2 & \le 3 \Big( \big\| \Delta_0 \big\|_2^2 + \big\| G_t^\theta \big\|_2^2 + \big\| F_{t+1} \big\|_2^2 \Big) ,
\end{align}
using the elementary inequality $(a + b + c)^2 \le 3 ( a^2 + b^2 + c^2 )$ for any $a, b, c \in \R$.
Consequently, 
\begin{align}
    (*_1) & \le 3^{\frac{1}{2}} \bigg( \big\| \Delta_0 \big\|_2^2 + \underbrace{ \frac{1}{T} \sum_{t=0}^{T-1} \sup_{\theta \in \Omega} \big\| G_t^\theta \big\|_2^2 }_{ = (*_{11}) } + \underbrace{ \frac{1}{T} \sum_{t=0}^{T-1} \big\| F_{t+1} \big\|_2^2 }_{ = (*_{12}) } \bigg)^{\frac{1}{2}} .
\end{align}
To prove that $(*_1)$ is bounded in probability, it suffices to show that both terms $(*_{11})$ and $(*_{12})$ are bounded in probability.

We begin by verifying that for $(*_{11})$.
By the definition of $G_t^\theta$ and the triangle inequality,
\begin{align}
    \big\| G_t^\theta \big\|_2 = \bigg\| \sum_{i=1}^t (1-\theta)^{t-i} F_i \bigg\|_2 \le \sum_{i=1}^t (1 - \theta)^{t-i} \| F_i \|_2 \le \sum_{i=1}^t (1 - \Omega_{\text{min}})^{t-i} \| F_i \|_2 .
\end{align}
Let $w_{t,i} = (1 - \Omega_{\text{min}})^{t-i}$ and define $W_t := \sum_{i=1}^{t} w_{t,i}$.
Applying Jensen's inequality yields
\begin{align}
	\big\| G_t^\theta \big\|_2^2 \le \bigg( \sum_{i=1}^{t} w_{t,i} \| F_i \|_2 \bigg)^2 & = W_t^2 \bigg( \sum_{i=1}^{t} \frac{w_{t,i}}{W_t} \| F_i \|_2 \bigg)^2 \le W_t \sum_{i=1}^{t} w_{t,i} \| F_i \|_2^2 .
\end{align}
Observe that this upper bound is independent of $\theta \in \Omega$.
Hence, we have
\begin{align}
	(*_{11}) \le \frac{1}{T} \sum_{t=0}^{T-1} W_t \sum_{i=1}^{t} w_{t,i} \| F_i \|_2^2 .
\end{align}
As noted at the beginning of \Cref{apx:proof_uniform_convergence}, Markov's inequality ensures that $(*_{11})$ is bounded in probability provided $\E[ (*_{11}) ]$ is uniformly bounded for all sufficiently large $T$.
We have
\begin{align}
	\E[ (*_{11}) ] & \le \frac{1}{T} \sum_{t=0}^{T-1} W_t \sum_{i=1}^{t} w_{t,i} \E\big[ \| F_i \|_2^2 \big] .
\end{align}
From \Cref{lem:residual_property} and our standing assumptions, we know that $\E\big[ \| F_i \|_2^2 \big] = \sigma_i^2 \le \sigma^2$. 
Furthermore, \Cref{lem:w_sum} establishes that $W_t \le 1 / \Omega_{\text{min}}$.
Substituting these bounds into the expectation yields
\begin{align}
	\E[ (*_{11}) ] & \le \frac{1}{T} \sum_{t=0}^{T-1} W_t \sum_{i=1}^{t} w_{t,i} \sigma^2 = \frac{\sigma^2}{T} \sum_{t=0}^{T-1} W_t^2 \le \frac{\sigma^2}{\Omega_{\text{min}}^2} ,
\end{align}
which holds uniformly for all $T$.
Thus, we conclude that the term $(*_{11})$ is bounded in probability.

Next, we verify that $(*_{12})$ is bounded in probability.
As before, it suffices to show that $\E[ (*_{12}) ]$ is uniformly bounded for all $T$.
Since $\E[ \| F_t \|_2^2 ] = \sigma_t^2 \le \sigma^2$ by \Cref{lem:residual_property} and assumption, we have
\begin{align}
	\E[ (*_{12}) ] & = \frac{1}{T} \sum_{t=0}^{T-1} \E[ \| F_{t+1} \|_2^2 ] = \frac{1}{T} \sum_{t=1}^{T} \E[ \| F_{t} \|_2^2 ] \le \frac{1}{T} \sum_{t=1}^{T} \sigma^2 = \sigma^2
\end{align}
which holds uniformly for all $T \ge 1$.
Thus, we conclude that the term $(*_{12})$ is bounded in probability, which in turn completes the entire proof of condition (i).

\subsubsection{Proof of Condition (ii): $(*_2) \le C_T | \theta - \theta' |$ for Some $C_T$ Bounded in Probability} \label{apx:stochastic_eqcon_2}

As in \Cref{apx:stochastic_eqcon_1}, we define $\Delta_0 := U_0^\theta - U_0$ which is independent of $\theta$ since $U_0^\theta$ is the fixed initial condition.
It follows from \Cref{prop:model_quantile_decomposition} in \Cref{apx:preparation} and the triangle inequality that
\begin{align}
    & \| U_t^{\theta}  - V_{t+1} - ( U_t^{\theta'} - V_{t+1} ) \|_2 = \big\| (1 - \theta)^t \Delta_0 + (\theta - \theta_*) G_t^\theta - (1 - \theta')^t \Delta_0 - (\theta' - \theta_*) G_t^{\theta'} \big\|_2 \\
    & \hspace{25pt} \le \big\| (1 - \theta)^t \Delta_0 - (1 - \theta')^t \Delta_0 \big\|_2 + \big\| (\theta - \theta_*) G_t^\theta - (\theta' - \theta_*) G_t^{\theta'} \big\|_2 \\
    & \hspace{25pt} \le \big\| (1 - \theta)^t \Delta_0 - (1 - \theta')^t \Delta_0 \big\|_2 + \big\| (\theta - \theta_*) G_t^\theta - (\theta - \theta_*) G_t^{\theta'} + (\theta - \theta_*) G_t^{\theta'} - (\theta' - \theta_*) G_t^{\theta'} \big\|_2 \\
    & \hspace{25pt} \le \big| (1 - \theta)^t - (1 - \theta')^t \big| \big\| \Delta_0 \big\|_2 + \big| \theta - \theta_* \big| \big\| G_t^\theta - G_t^{\theta'} \big\|_2 + \big| \theta - \theta' \big| \big\| G_t^{\theta'} \big\|_2 .
\end{align}
Standard calculus shows that the function $\theta \mapsto (1 - \theta)^{t}$ is Lipschitz continuous in $\Omega$ with the constant $t w^{t-1}$, where $w := (1 - \Omega_{\text{min}})$.
Since $| \theta - \theta_* | \le 1$ in $\Omega \subset (0, 1)$, this Lipschitz condition yields
\begin{align}
    \| U_t^{\theta}  - V_{t+1} - ( U_t^{\theta'} - V_{t+1} ) \|_2 & \le \frac{ t w^{t} }{w} \big| \theta - \theta' \big| \big\| \Delta_0 \big\|_2 + \big\| G_t^\theta - G_t^{\theta'} \big\|_2 + \big| \theta - \theta' \big| \big\| G_t^{\theta} \big\|_2 .
\end{align}
Applying the elementary inequality $(a + b + c)^2 \le 3 (a^2 + b^2 + c^2)$ for all $a, b, c \in \R$ gives 
\begin{align}
    (*_2) & = 3^{\frac{1}{2}} \Bigg( \big| \theta - \theta' \big| \frac{ \big\| \Delta_0 \big\|_2 }{ w } \bigg( \frac{1}{T} \sum_{t=0}^{T-1} t^2 w^{2 t} \bigg)^{\frac{1}{2}} + \bigg( \frac{1}{T} \sum_{t=0}^{T-1} \big\| G_t^\theta - G_t^{\theta'} \big\|_2^2 \bigg)^{\frac{1}{2}} + \big| \theta - \theta' \big| \bigg( \frac{1}{T} \sum_{t=0}^{T-1} \big\| G_t^{\theta} \big\|_2^2 \bigg)^{1/2} \Bigg) . 
\end{align}
Standard calculus shows that $t^2 w^{2 t}$ is uniformly bounded by $1 / (e^2 (\log w)^2 )$. 
Because $e > 1$, we can loosen and simplify this bound to $t^2 w^{2 t} < 1 / (\log w)^2$, which leads to
\begin{align}
    (*_2) & \le 3^{\frac{1}{2}} \Bigg( \big| \theta - \theta' \big| \frac{ \big\| \Delta_0 \big\|_2 }{ w | \log w | } + \bigg( \underbrace{ \frac{1}{T} \sum_{t=0}^{T-1} \big\| G_t^\theta - G_t^{\theta'} \big\|_2^2 }_{ = (*_{21}) } \bigg)^{\frac{1}{2}} + \big| \theta - \theta' \big| \bigg( \underbrace{ \frac{1}{T} \sum_{t=0}^{T-1} \sup_{\theta \in \Omega} \big\| G_t^{\theta} \big\|_2^2 }_{ = (*_{22}) } \bigg)^{1/2} \Bigg) . \label{eq:item_2_Lipschitz}
\end{align}
Observe that $(*_{22})$ is identical to $(*_{11})$ in \Cref{apx:stochastic_eqcon_1}, which we previously showed to be bounded in probability.
Therefore, it remains only to bound the term $(*_{21})$.

It follows from the definition of $G_t^\theta$ and the triangle inequality that
\begin{align}
    \big\| G_t^\theta - G_t^{\theta'} \big\|_2 & = \bigg\| \sum_{i=1}^{t} \left( ( 1 - \theta )^{t-i} - ( 1 - \theta' )^{t-i} \right) F_i \bigg\|_2 \le \sum_{i=1}^{t} \left| ( 1 - \theta )^{t-i} - ( 1 - \theta' )^{t-i} \right| \| F_i \|_2. 
\end{align}
Standard calculus verifies that, for each $1 \le i \le t$, the function $\theta \mapsto (1 - \theta)^{t-i}$ is Lipschitz continuous in $\Omega$ with constant $(t-i) w^{t-i-1}$, where $w := (1 - \Omega_{\text{min}})$.
This yields
\begin{align}
    \big\| G_t^\theta - G_t^{\theta'} \big\|_2 & \le \sum_{i=1}^{t} (t-i) w^{t-i-1} \left| \theta - \theta' \right| \| F_i \|_2 = \frac{\left| \theta - \theta' \right|}{w} \sum_{i=1}^{t} (t-i) w^{t-i} \| F_i \|_2 .
\end{align}
Let $\tilde{w}_{t,i} := (t-i) w^{t-i}$, and define $\tilde{W}_t := \sum_{i=1}^{t} \tilde{w}_{t,i}$.
Applying Jensen's inequality yields
\begin{align}
    \big\| G_t^\theta - G_t^{\theta'} \big\|_2^2 & \le \frac{\left| \theta - \theta' \right|^2}{w^2} \bigg( \tilde{W}_t \sum_{i=1}^{t} \frac{ \tilde{w}_{t,i} }{ \tilde{W}_t } \| F_i \|_2 \bigg)^2 \le \frac{ \left| \theta - \theta' \right|^2 }{ w^2 } \tilde{W}_t \sum_{i=1}^{t} \tilde{w}_{t,i} \| F_i \|_2^2 . \label{eq:G_t_theta_Lipschitz}
\end{align}
Recall that a sequence $\{i r^i \}_{i=0}^{\infty}$ for some constant $r > 0$ is called an arithmetic-geometric sequence.
Its partial sum $s_t := \sum_{i=0}^{t} ir^i$ converges monotonically to $r / (1 - r)^2$ for any $0 < r < 1$.
Due to the monotonicity, it immediately follows that $s_t \le r / (1 - r)^2$ for all $t \ge 0$.
Applying this bound gives
\begin{align}
	\tilde{W}_t = \sum_{i=1}^{t} (t-i) w^{t-i} = \sum_{i=0}^{t-1} i w^i \le \frac{w}{(1 - w)^2} , \label{eq:a_t_z_t_bound}
\end{align}
which holds uniformly for all $t \ge 0$.
Substituting \eqref{eq:a_t_z_t_bound} into \eqref{eq:G_t_theta_Lipschitz} yields
\begin{align}
    (*_{21}) = \frac{1}{T} \sum_{t=0}^{T-1} \big\| G_t^\theta - G_t^{\theta'} \big\|_2^2 \le \frac{ \left| \theta - \theta' \right|^2 }{ w (1 - w)^2 } \underbrace{ \frac{1}{T} \sum_{t=0}^{T-1} \sum_{i=1}^{t} \tilde{w}_{t,i} \| F_i \|_2^2 }_{ = (*_{23}) } . \label{eq:lipchitz A_theta}
\end{align}
Furthermore, substituting \eqref{eq:lipchitz A_theta} into the main bound \eqref{eq:item_2_Lipschitz} gives
\begin{align}
    (*_2) & \le 3^{\frac{1}{2}} \Bigg( \big| \theta - \theta' \big| \frac{ \big\| \Delta_0 \big\|_2 }{ w | \log w | } + \left| \theta - \theta' \right| \frac{ 1 }{ w^{\frac{1}{2}} (1 - w) } (*_{23})^{\frac{1}{2}} + \big| \theta - \theta' \big| (*_{22})^{\frac{1}{2}} \Bigg) \\
    & = 3^{\frac{1}{2}}  \Bigg( \frac{ \big\| \Delta_0 \big\|_2 }{ w | \log w | } + \frac{ 1 }{ w^{\frac{1}{2}} (1 - w) } (*_{23})^{\frac{1}{2}} + (*_{22})^{\frac{1}{2}} \Bigg) \big| \theta - \theta' \big| .
\end{align}
Recall that the term $(*_{22})$ is bounded in probability, as established previously.
Therefore, verifying that $(*_{23})$ is bounded in probability completes the proof of condition (ii).

As noted at the beginning of \Cref{apx:proof_uniform_convergence}, Markov's inequality ensures that $(*_{23})$ is bounded in probability provided $\E[ (*_{23}) ]$ is uniformly bounded for all sufficiently large $T$.
We have
\begin{align}
    \E[ (*_{23}) ] = \frac{1}{T} \sum_{t=0}^{T-1} \sum_{i=1}^{t} \tilde{w}_{t,i} \E\big[ \| F_i \|_2^2 \big] = \frac{1}{T} \sum_{t=0}^{T-1} \sum_{i=1}^{t} \tilde{w}_{t,i} \sigma_i^2 \le \frac{1}{T} \sum_{t=0}^{T-1} \sum_{i=1}^{t} \tilde{w}_{t,i} \sigma^2 = \frac{\sigma^2}{T} \sum_{t=0}^{T-1} \tilde{W}_{t}
\end{align}
where we use \Cref{lem:residual_property} and the assumption $\sigma_i \le \sigma$.
It then follows from \eqref{eq:a_t_z_t_bound} that
\begin{align}
    \E[ (*_{23}) ] \le \frac{\sigma^2 w}{(1 - w)^2} ,
\end{align}
which holds uniformly for all $T \ge 1$.
Therefore, the term $(*_{23})$ is bounded in probability, which in turn completes the entire proof of condition (ii).

\subsubsection{Proof of Condition 2: Continuity of Limiting Function $L_\infty(\theta)$} \label{apx:limiting_con}
    
By the definition of $L_\infty(\theta)$, it suffices to show that $A_\infty(\theta)$ is continuous in $\theta$.
Recall that \Cref{apx:proof_A_T_convergence} within the proof of \Cref{thm:pointwise_convergence} established
\begin{align}
    A_\infty(\theta) = \lim_{T \to \infty} \frac{1}{T} \sum_{t=0}^{T-1} m_t^\theta = \lim_{T \to \infty} \frac{1}{T} \sum_{t=0}^{T-1} \sum_{i=1}^{t} (1 - \theta)^{t-i} \sigma_i^2 
\end{align}
where $m_t^\theta$ is defined as in \Cref{lem:g_moment}.
For notational convenience, we express $A_\infty(\theta)$ as
\begin{align}
	A_\infty(\theta) = \lim_{T \to \infty} f_T(\theta) \quad \text{where} \quad f_T(\theta) := \frac{1}{T} \sum_{t=0}^{T-1} \sum_{i=1}^{t} (1 - \theta)^{t-i} \sigma_i^2 .
\end{align}
Note that $f_T$ is a continuous function for each $T$, and the sequence of $f_T(\theta)$ converges to $A_\infty(\theta)$ pointwise by definition.
We will show that the sequence of $f_T$ is uniformly bounded and equicontinuous. 
Under these conditions, the Ascoli-Arzel\`{a} theorem \citep[e.g.,][]{Rudin1976} guarantees that the sequence of $f_T$ admits a uniformly convergent subsequence to $A_\infty(\theta)$.
Then, the sequence must converge uniformly to this same limit $A_\infty(\theta)$, because it is equicontinuous and pointwise convergent.
Finally, by the uniform limit theorem \citep[e.g.,][]{Rudin1976}, the uniform limit of continuous functions remains continuous, which establishes the continuity of $A_\infty(\theta)$.

We first verify that the sequence of $f_T$ is uniformly bounded for all $T \ge 1$. 
Applying \Cref{lem:w_sum} and the standing assumption $\sigma_i \le \sigma$ yields
\begin{align}
    f_T(\theta) \le \frac{1}{T} \sum_{t=0}^{T-1} \sum_{i=1}^{t} (1 - \theta)^{t-i} \sigma^2 = \frac{\sigma^2}{\theta} \le \frac{\sigma^2}{\Omega_{\text{min}}}  .
\end{align}
Next, we establish that the sequence $\{f_T\}$ is equicontinuous. 
By the mean value theorem, we have
\begin{align}
    | f_T(\theta) - f_T(\theta') | \le \sup_{\theta \in \Omega} | f_T'(\theta) | | \theta - \theta' | 
\end{align}
where $f_T'(\theta)$ denotes the derivative of $f_T$ with respect to $\theta$, given by
\begin{align}
    f_T'(\theta) = \frac{1}{T} \sum_{t=0}^{T-1} \sum_{i=1}^{t} (t-i) (1 - \theta)^{t-i-1} \sigma_i^2 .
\end{align}
We apply the assumption $\sigma_i \le \sigma$ and rearrange the terms to obtain
\begin{align}
    | f_T'(\theta) | \le \frac{\sigma^2}{(1 - \theta)} \frac{1}{T} \sum_{t=0}^{T-1} \sum_{i=1}^{t} (t-i) (1 - \theta)^{t-i} = \frac{\sigma^2}{(1 - \theta)} \frac{1}{T} \sum_{t=0}^{T-1} \sum_{i=0}^{t-1} i (1 - \theta)^{i} .
\end{align}
Recall from our earlier derivation that the partial sum $s_t = \sum_{i=0}^{t-1} i (1 - \theta)^i$ of the arithmetico-geometric sequence $\{i (1 - \theta)^i \}_{i=0}^{\infty}$ is uniformly bounded by $r / (1 - r)^2$ with $r = 1 - \theta$.
This gives
\begin{align}
    | f_T'(\theta) | \le \frac{\sigma^2}{(1 - \theta)} \frac{1}{T} \sum_{t=0}^{T-1} \frac{(1 - \theta)}{ \theta^2 } \le \frac{\sigma^2}{\theta^2} \le \frac{\sigma^2}{\Omega_{\text{min}}^2} .
\end{align}
Therefore, the sequence of $f_T$ is equicontinuous, which concludes the continuity of $L_\infty$.

\subsection{Proof of \Cref{thm:consistency}} \label{apx:proof_consistency}

\begin{proof}
    We restrict the parameter space to $\Omega$ without loss of generality, since $\Omega$ contains the estimator for all sufficiently large $T$ by assumption.
	By \Cref{thm:pointwise_convergence}, the Wasserstein error $L_T(\theta)$ converges in probability to the limiting error $L_\infty(\theta) := A_\infty(\theta) (\theta - \theta_*)^2 + C_\infty$ pointwise, where $A_\infty(\theta)$ is a strictly positive function and $C_\infty$ is a non-negative constant.
    To establish the consistency, we apply Theorem 2.1 of \cite{Newey1994} for the pointwise convergent error $L_T(\theta)$.
    This requires verifying the four conditions (i)--(iv) in their theorem.
    Each condition is verified as follows:
    \begin{enumerate}[label=(\roman*)]
        \item the limiting error $L_\infty$ is uniquely minimized at $\theta = \theta_*$ since $A_\infty(\theta)$ is strictly positive;
        \item $\Omega$ is a compact set;
        \item the limiting error $L_\infty$ is continuous in $\Omega$, as previously established in \Cref{apx:limiting_con};
        \item the Wasserstein error $L_T$ uniformly converges in probability to $L_\infty$ by \Cref{thm:uniform_convergence}.
    \end{enumerate}
    Therefore, the desired consistency follows from Theorem 2.1 of \cite{Newey1994}.
\end{proof}

\section{Additional Experimental Details} \label{apx:experimental_detail}

This section provides additional details of the experiments.
\Cref{apx:simulation_maps} verifies that the random pushforward models used in \Cref{sec:simulation} satisfy the standing assumptions required for the theoretical analysis.
\Cref{apx:loss_computation} clarifies the computation of the Wasserstein loss evaluated in \Cref{sec:experiment}.

\subsection{Properties of Random Pushforward Maps in \Cref{sec:simulation}}
\label{apx:simulation_maps}

Throughout, the random variables defining the maps are assumed to be independent across time $t$.

First, the random shift map is defined by
\begin{align}
    \T_t^{\mathrm{Shift}}(x)=x+B_t, \qquad B_t \sim \mathrm{Normal}(0,s^2) .
\end{align}
In the simulation study, we used $s=1$.
For every $x\in\R$, we immediately have
\begin{align}
    \E\left[ \T_t^{\mathrm{Shift}}(x)-x \right] = \E[ B_t ] = 0 \quad \text{and} \quad \E\left[\left( \T_t^{\mathrm{Shift}}(x)-x \right)^2 \right] = \E[ B_t^2 ] = s^2,
\end{align}
The same argument applies to the fourth moment, which is uniformly bounded.
Furthermore, each realization of $\T_t^{\mathrm{Shift}}$ is monotone increasing, since
\begin{align}
    \frac{d}{dx}\T_t^{\mathrm{Shift}}(x)=1.
\end{align}
Thus $\T_t^{\mathrm{Shift}}$ is a valid random pushforward map satisfying the standing assumptions.

Next, the random sine map is defined by
\begin{align}
    \T_t^{\mathrm{Sine}}(x) = \sum_{j=1}^k W_{t,j} \left( x-\frac{a}{\pi}\sin\bigl(\pi(x-C_{t,j})\bigr) \right), \qquad C_{t,j}\overset{\mathrm{i.i.d.}}{\sim}\mathrm{Uniform}(-1,1),
\end{align}
where $0 < a \leq 1$ and the weights are random positive weights normalized to sum to one. 
In the simulation study, we use $a=0.3$ and $k=3$.
Since the weights sum to one,
\begin{align}
    \T_t^{\mathrm{Sine}}(x) -  x = -\frac{a}{\pi} \sum_{j=1}^k W_{t,j} Z_{t,j}(x) \quad \text{where} \quad Z_{t,j}(x) := \sin\bigl(\pi(x-C_{t,j})\bigr).
\end{align}
Because the integral is taken over one full period of the sine function, for any fixed $x\in\R$ we have
\begin{align}
    \E_{C_{t,j}}[ Z_{t,j}(x) ] & = \frac{1}{2} \int_{-1}^{1}\sin( \pi(x-c) )dc = 0; \\
    \E_{C_{t,j}}[ Z_{t,j}(x)^2 ] & = \frac{1}{2} \int_{-1}^{1}\sin( \pi(x-c) )^2 dc = \frac{1}{2}; \\
    \E_{C_{t,j}}[ Z_{t,j}(x)^4 ] & = \frac{1}{2} \int_{-1}^{1}\sin( \pi(x-c) )^4 dc \le 1 .
\end{align}
Therefore, we have
\begin{align}
    \E\left[ \T_t^{\mathrm{Sine}}(x) - x \right] = \E_{W_{t,1}, \dots, W_{t,k}}\bigg[ -\frac{a}{\pi} \sum_{j=1}^k W_{t,j} \E_{C_{t,j}}[ Z_{t,j}(x) ] \bigg] = 0 .
\end{align}
Since the $Z_{t,j}(x)$ are independent and centered, we further have
\begin{align}
    \E\left[ \left( \T_t^{\mathrm{Sine}}(x) - x \right)^2 \right] = \frac{a^2}{2\pi^2} \E_{W_{t,1}, \dots, W_{t,k}}\bigg[ \sum_{j=1}^k W_{t,j}^2 \bigg] = \sigma_t^2
\end{align}
where the cross terms are zero.
Since $W_{t,j}\geq 0$ and $\sum_{j=1}^k W_{t,j} = 1$, we have 
\begin{align}
    \sum_{j=1}^k W_{t,j}^2 \leq \bigg( \sum_{j=1}^k W_{t,j} \bigg)^2 = 1 .
\end{align}
Consequently, we have the time-independent bound $\sigma_t^2 \le a^2 / (2\pi^2) = \sigma^2$.
A similar argument applies to the fourth moment, yielding that
\begin{align}
    \E\left[ \left( \T_t^{\mathrm{Sine}}(x) - x \right)^4 \right] \le \left( \frac{a}{\pi} \right)^4 .
\end{align}
Finally, each realization of the map is monotone non-decreasing for $0 < a\leq 1$, since its derivative is
\begin{align}
    \frac{d}{dx}\T_t^{\mathrm{Sine}}(x)= \sum_{j=1}^k W_{t,j} \left( 1 - a \cos\bigl(\pi(x-C_{t,j})\bigr) \right) \geq \sum_{j=1}^k W_{t,j} ( 1 - a ) \geq 0 .
\end{align}
Therefore $\T_t^{\mathrm{Sine}}$ is a valid random pushforward map satisfying the standing assumptions.

\subsection{Computation of the Wasserstein Loss in \Cref{sec:experiment}} \label{apx:loss_computation}

Since all distributions are defined on $\mathbb{R}$, the Wasserstein loss can be computed through quantile functions. 
Let $V_{m,t|t-1}$ and $V_t$ denote the quantile functions of $\nu_{m,t|t-1}$ and $\nu_t$. 
Then we have
\begin{align}
    l_{m,t}^2 = W_2^2\left(\nu_{m,t|t-1},\nu_t\right) = \int_0^1 \left( V_{m,t|t-1}(q)-V_t(q) \right)^2 dq . \label{eq:experiment_loss_quantile}
\end{align}
We evaluated this expression of the Wasserstein loss, using the standard quadrature method on a common midpoint grid with $N_q$ points.
In \Cref{sec:equity} where the empirical distribution $\nu_t$ at each trading day $t$ consists of a large set of points, we evaluated the expression \eqref{eq:experiment_loss_quantile} on grid points of the sufficiently many number $N_q=200$. 
In \Cref{sec:smart_meter} where $\nu_t$ consists of 48 residual demands, we evaluated the expression \eqref{eq:experiment_loss_quantile} on grid points of the sample number $N_q=48$.

\section{Wesmooth: Python Package for WES} \label{appendix:software}

The accompanying Python package \texttt{wesmooth} is available at:
\begin{center}
\url{https://github.com/wilson-ye-chen/wesmooth/}
\end{center}
The package contains the implementation used in this paper, including simulation from the WES process, estimation of the smoothing parameter, one-step-ahead forecasting, the random pushforward maps used in \Cref{sec:simulation}, and visualization tools for distributional sample paths. 
It can be installed directly from the repository using
\begin{mdframed}
\begin{lstlisting}
pip install git+https://github.com/wilson-ye-chen/wesmooth
\end{lstlisting}
\end{mdframed}

The following code illustrates the main workflow. 
A sample path is simulated from the WES process, the smoothing parameter is estimated by minimizing the average squared Wasserstein one-step-ahead prediction loss, and the fitted WES path is then generated.

\begin{mdframed}
\begin{lstlisting}
from wesmooth.core import WassExpSmooth
from wesmooth.core import WassExpSmoothSampler
from wesmooth.maps import err_map_shift

theta = 0.5
n = 200

sampler = WassExpSmoothSampler(theta, err_map_shift)
y, mu = sampler.sample_path(n)

model = WassExpSmooth()
model.fit(y)

theta_hat = model.theta
y_hat = model.predict(y)
\end{lstlisting}
\end{mdframed}

The package also provides plotting utilities. 
The function \texttt{plot\_density\_path} displays a sequence of distributions as sideways density curves, as used in the simulation and empirical sections.

\begin{mdframed}
\begin{lstlisting}
import numpy as np
from wesmooth.plot import plot_density_path

n_grid = y.shape[1]
p_grid = (np.arange(n_grid) + 0.5) / n_grid

fig, ax = plot_density_path(
    y[:50],
    p_grid=p_grid,
    density_method='quantile',
    xlabel='t',
    ylabel='x'
)
\end{lstlisting}
\end{mdframed}

For empirical distributions, the same function can be used with kernel density estimates by setting \texttt{density\_method='kde'}. This package was used to reproduce the simulation and empirical results reported in the paper.

\end{document}